\documentclass[11pt]{article}
\usepackage{amsmath,amssymb,amsfonts,amsthm,epsfig}
\usepackage[usenames,dvipsnames]{xcolor}

\usepackage{bm,xspace}

\usepackage{cancel}
\usepackage{fullpage}
\usepackage{liyang}
\usepackage{framed}
\usepackage{verbatim}
\usepackage{enumitem}
\usepackage{array}
\usepackage{multirow}
\usepackage{afterpage}
\usepackage{mathrsfs} 
\usepackage[ruled,vlined]{algorithm2e}

\usepackage[normalem]{ulem}

\usepackage{accents}
\usepackage{pdfpages}

\renewcommand{\hat}{\widehat}

\usepackage{caption} 

\usepackage{todonotes}

\makeatletter
\newtheorem*{rep@theorem}{\rep@title}
\newcommand{\newreptheorem}[2]{
\newenvironment{rep#1}[1]{
 \def\rep@title{#2 \ref{##1}}
 \begin{rep@theorem}\itshape}
 {\end{rep@theorem}}}
\makeatother
\theoremstyle{plain}



\newcommand{\ignore}[1]{}

\def\colorful{0}

\ifnum\colorful=1

\newcommand{\red}[1]{{\color{red} {#1}}}

\fi
\ifnum\colorful=0

\newcommand{\red}[1]{{{#1}}}

\fi

\usepackage{boxedminipage}

\newreptheorem{theorem}{Theorem}
\newtheorem*{theorem*}{Theorem}
\newreptheorem{lemma}{Lemma}
\newreptheorem{proposition}{Proposition}
\newtheorem*{noclaim*}{Claim}

\newcommand{\weight}{\mathrm{weight}}

\newcommand{\tail}{\mathrm{tail}}

\usepackage{dsfont}

\renewcommand{\N}{\mathds{N}} 
\renewcommand{\R}{\mathds{R}}

\newcommand{\union}{\ensuremath{\cup}}
\newcommand{\card}[1]{\left| {#1} \right|}
\newcommand{\half}{\frac{1}{2}}

\newcommand{\pmo}{\set{-1, 1}}
\renewcommand{\vec}[1]{#1}
\newcommand{\rv}[1]{\bm{#1}}
\newcommand{\set}[1]{\{ {#1} \}}

\newcommand{\fh}{\ensuremath{\hat{f}}}
\newcommand{\gh}{\ensuremath{\hat{g}}}
\newcommand{\ft}{\ensuremath{\tilde{f}}}

\newcommand{\fb}{\ensuremath{\accentset{\bigtriangleup}{f}}}
\newcommand{\gb}{\ensuremath{\accentset{\bigtriangleup}{g}}}
\newcommand{\fc}{\ensuremath{\accentset{\diamond}{f}}}
\newcommand{\gc}{\ensuremath{\accentset{\diamond}{g}}}
\newcommand{\fs}{\ensuremath{f^*}}

\DeclareRobustCommand*{\dA}{\ensuremath{\accentset{\diamond}{A}}}
\DeclareRobustCommand*{\dB}{\ensuremath{\accentset{\diamond}{B}}}

\newcommand{\inner}[2]{\ensuremath{\langle #1, #2 \rangle}}
\newcommand{\vect}[1]{#1}
\newcommand{\indic}[1]{\mathbbm{1}\{#1\}}
\newcommand{\normone}[1]{\norm{#1}_1}
\newcommand{\normtwo}[1]{\norm{#1}_2}

\newcommand{\expectation}[1]{\E\left[#1\right]}
\newcommand{\gaussian}[2]{\ensuremath{ \mathcal{N}\left(#1,#2\right) }}
\newcommand{\bigO}[1]{O\left(#1\right)}
\newcommand{\ProbaOf}[1]{\Pr\left[#1\right]}
\newcommand{\approxEqual}[1]{\ensuremath{\stackrel{#1}{\approx}}}

\newcommand{\dChow}{\ensuremath{d_{\mathrm{Chow}}}}
\newcommand{\dPC}[1]{\ensuremath{d_{\mathrm{Chow}, #1}}} 
\newcommand{\dFourier}{\ensuremath{d_{\mathrm{Shapley-Fourier}}}}
\newcommand{\dShapley}{\ensuremath{d_{\mathrm{Shapley}}}}
\newcommand{\dPS}[1]{\ensuremath{d_{\textrm{Shapley}, #1}}} 

\newcommand{\Shap}{{\mathrm{Shap}}}
\newcommand{\DShap}{\mathcal{D}_{\Shap}}
\newcommand{\fail}{\ensuremath{\mathrm{fail}}}

\newcommand{\sgn}{\ensuremath{\mathrm{sgn}}}
\DeclarePairedDelimiter\iprod{\langle}{\rangle}

\newcommand{\deltaExpression}{$\delta=1/n^c$ for some constant $c>1$}
\begin{document}

\title{Reconstructing weighted voting schemes from\\ partial information about their power indices}
\author{Huck Bennett\thanks{University of Michigan, \texttt{hdbco@umich.edu}. Part of this work was performed while the author was at Northwestern University and supported by a Warren Postdoctoral Fellowship.} \and
Anindya De\thanks{University of Pennsylvania, {\tt anindyad@cis.upenn.edu}. Supported by NSF grants CCF-1926872 and CCF-1910534. Part of the work was done while the author was on the faculty at Northwestern University.} \and
Rocco A. Servedio\thanks{Columbia University, {\tt rocco@cs.columbia.edu}. Supported by NSF grants CCF-1814873, IIS-1838154, CCF-1563155, and by the Simons Collaboration on Algorithms and Geometry.
} \and Emmanouil-Vasileios Vlatakis-Gkaragkounis\thanks{Columbia University, \texttt{emvlatakis@cs.columbia.edu}. Supported by NSF grants CCF-1703925, CCF-1763970, CCF-1814873, CCF-1563155, and by the Simons Collaboration on Algorithms and Geometry.}}
\maketitle 

\begin{abstract}
A number of recent works \cite{goldberg2006,os2011,journals/geb/DeDS17,de2014} have considered the problem of approximately reconstructing an unknown weighted voting scheme given information about various sorts of ``power indices'' that characterize the level of control that individual voters have over the final outcome. In the language of theoretical computer science, this is the problem of approximating an unknown linear threshold function (LTF) over $\bn$ given some numerical measure (such as the function's $n$ ``Chow parameters,'' a.k.a.~its degree-1 Fourier coefficients, or the vector of its $n$ Shapley indices)
of how much each of the $n$ individual input variables affects the outcome of the function.

In this paper we consider the problem of reconstructing an LTF given only \emph{partial} information about its Chow parameters or Shapley indices; i.e.~we are given only the Chow parameters or the Shapley indices corresponding to a subset $S \subseteq [n]$ of the $n$ input variables.  A natural goal in this partial information setting is to find an LTF whose Chow parameters or Shapley indices corresponding to indices in $S$ accurately match the given Chow parameters or Shapley indices of the unknown LTF.  We refer to this as the \emph{Partial Inverse Power Index Problem}.

Our main results are a polynomial time algorithm for the ($\varepsilon$-approximate) Chow Parameters Partial Inverse Power Index Problem and a quasi-polynomial time algorithm for the ($\varepsilon$-approximate) Shapley Indices Partial Inverse Power Index Problem.

\end{abstract}
\clearpage
\tableofcontents

\thispagestyle{empty}

\newpage
\setcounter{page}{1}

\section{Introduction}

\subsection{Background: Power indices and inverse power index problems.}

A natural question that arises in voting theory is how to quantify the ``power'' of an individual voter in a collective decision-making scheme.  For simplicity, in this paper we consider only \emph{weighted voting games}; in the language
of theoretical computer science, these correspond to \emph{linear threshold functions (LTFs)} $f: \bn \to \bits$, $f(x) = \sign(w \cdot x - \theta)$, where each $w_i\geq 0$ is a non-negative \emph{voting weight}.  In such a weighted voting game there are $n$ binary voters, each with some amount of non-negative weight, and the collective decision is an affirmative one if and only if the total voting weight of the affirmative voters exceeds the threshold $\theta$.  

If all $n$ of the voting weights are the same then it is clear that all $n$ voters have the same amount of ``power'' over the final outcome, but it is much less clear how to measure the power of a single voter when the voting weights may vary.
As a simple example, consider a setting with three voters who have voting weights of $49,49$ and $2$, in which a total of 51 votes are required for the proposition to pass. 
While the disparity between voting weights may at first suggest that the two voters with 49 votes each have most of the ``power,'' any coalition of two voters is sufficient to pass the proposition and any single voter is insufficient, so the voting power of all three voters is in fact equal. 
Such examples are not merely hypothetical; in the first voting scheme used by the European Economic Community (the predecessor of the current European Union) in 1957, decisions were  accepted if they were supported by at least 12 out 17 votes, and the members' weight distribution was $\{$Germany : 4, France : 4, Italy : 4, The Netherlands : 2, Belgium : 2, Luxembourg : 1$\}$ \cite{Rome1957, leech2002}.  Inspection shows that there is no voting outcome in which Luxembourg could influence the result, and thus its real voting power was null.

A  number of different numerical measures, known as ``power indices,'' have been proposed to quantify how much power each voter has in a weighted voting election scheme.  These include the Deegan-Packel index \cite{deegan1978}, the Holler index \cite{holler1982,johnston1978}, and several others (see the extensive survey of de Keijzer \cite{deKeijer2008}). In the rest of this paper we confine our attention to two particularly well-studied power indices. The first of these are the \emph{Banzhaf indices}  \cite{Banzhaf1964}; in theoretical computer science these are more commonly known as the \emph{Chow parameters} \cite{chow1961} and we shall henceforth refer to them as such. There are $n+1$ Chow parameters of an $n$-variable LTF, and they are simply the constant- and degree-1 Fourier coefficients.\footnote{Since every LTF is a unate Boolean function, up to sign the degree-1 Fourier coefficients are the same as the $n$ coordinate influences of the LTF.} The second of these are the \emph{Shapley-Shubik indices} \cite{shapley1954}, henceforth referred to for brevity as the \emph{Shapley indices}; these are perhaps the best known, and certainly the oldest, power indices studied in the literature. Given an LTF $f: \bn \to \bits$ with non-negative weights that satisfy $f((-1)^n)=-1,$ $f(1^n)=1$, the Shapley indices are a vector of $n$ associated probabilities  $(\fc(1),\dots,\fc(n))$ that sum to 1.  The $i$th probability is the probability that $x_i$ is the ``pivotal index`` causing $f$'s value to flip from $-1$ to $1$, starting at the input $(-1)^n$ and flipping indices from $-1$ to $1$ in a random order.

The \#P-hardness of counting 0/1 knapsack solutions easily implies that it is \#P-hard to exactly compute the Chow parameters of a given LTF, and it turns out that the Shapley indices of LTFs are also \#P-hard to compute \cite{DP94}.  However, simple sampling-based approaches yield efficient algorithms for obtaining highly accurate estimates of the Chow parameters or the Shapley indices (see e.g.~\cite{Leech:03b,procaccia2010}).  Much more challenging are the \emph{inverse} problems, such as the \emph{(Inverse) Chow Parameters Problem}:  given as input the Chow parameters of an unknown LTF (or accurate approximations of the Chow parameters), construct an LTF whose Chow parameters are very close to the input provided.  
A beautiful result of C.-K. Chow from the 2nd FOCS conference \cite{chow1961} shows that given the exact Chow parameters of an LTF, it is information-theoretically possible to recover the LTF, but the proof is entirely non-constructive.  The algorithmic problem of appproximating an unknown LTF from its Chow parameters was studied in a number of recent works \cite{goldberg2006,os2011,de2014}, and more recently the analogous problem for Shapley indices (the \emph{Inverse Shapley Indices Problem)} was studied as well \cite{journals/geb/DeDS17}.  The current state of the art for the Inverse Chow Parameters Problem \cite{de2014} is an algorithm which, for any constant $\eps$, runs in fixed $\poly(n)$ time and outputs an LTF whose Chow parameters match the given input vector of Chow parameters up to $\eps$-accuracy (in a sense which we make precise later). A similar-in-spirit result (with some technical restrictions and somewhat weaker quantitative bounds; we defer a precise statement until later) was given for the inverse Shapley indices problem in \cite{journals/geb/DeDS17}. We also remark here that the problem of exactly reconstructing a LTF from its Chow parameters (or its Shapley indices) was recently~\cite{diakonikolas2019complexity} shown to be computationally intractable. 

\subsection{This work:  The Partial Inverse Power Index Problem.}

A drawback of the algorithms of \cite{goldberg2006,os2011,journals/geb/DeDS17} and \cite{de2014} for the Inverse Chow Parameters and Inverse Shapley Indices Problems is that they require \emph{full information} about the target vector of power indices:  none of these algorithms can be used unless \emph{all} of the target Chow parameters (or Shapley indices) are provided to the algorithm.  This is a potentially significant drawback for settings in which exhaustive information about the target power indices may not be available.

The current paper addresses this by studying algorithms for the \emph{Partial Inverse Power Index Problem}.  In this partial information version of the problem, the algorithm is only given a subset $S \subset [n]$ of the $n$ ``voters'' (i.e.~coordinates of the unknown LTF $f$) and the associated power indices (Chow parameters or Shapley indices) corresponding to those coordinates, and the goal is to output a weighted voting game $f$ (i.e.~an LTF) such that the power indices of $f$ in coordinates $S$ closely match the input that was provided.  We give algorithms for both the Chow Parameters and Shapley Indices version of this problem; to explain our results, we begin by giving a detailed definition of each of these problems below.

\subsubsection{The Partial Chow Parameters Problem}
We begin by recalling the definition of the Chow parameters:
\begin{definition}
Given $f: \pmo^n \to \pmo$, for $0 \leq i \leq n$ the $i$th \emph{Chow parameter} of $f$ is the value
\[
\fh(i) := \Ex_{\bx \sim \pmo^n}[f(\bx)\bx_i],
\]
where we define $\bx_0$ to be identically 1 and ``$\bx \sim \bn$'' indicates that $\bx$ is a uniform random element of $\bn.$

\label{def:chow-parameters}
\end{definition}
Thus the Chow Parameters of a Boolean function $f: \pmo^n \to \pmo$ are simply its $n + 1$ degree-0 and degree-1 Fourier coefficients.
The \emph{Chow Parameters Problem} is the problem of (approximately) recovering a weights-based representation of a linear-threshold function (LTF) $f$ given the Chow Parameters of $f$ as input.
The (approximate) \emph{Partial Chow Parameters Problem} is the partial information variant of the Chow Parameters Problem where only a \emph{subset} of the Chow Parameters of $f$ corresponding to some subset of indices $S$ are given as input, and the goal is to recover a weights-based representation of an LTF $f'$ such that the ``partial Chow distance with respect to $S$'' between $f$ and $f'$, namely $\big(\sum_{i \in S} (\fh(i) - \fh'(i))^2\big)^{1/2}$, is small:

\begin{definition}
The $\eps$-approximate \emph{Partial Chow Parameters Problem} is the promise problem defined as follows. Given $\set{(i, \fh(i)) : i \in S}$ for some LTF $f : \pmo^n \to \pmo$ and some $S \subseteq \set{0, 1, \ldots, n}$ as input, output weights $w_1, \ldots, w_n$ and a threshold $\theta$ such that $f'(x) := \sign(w \cdot  x - \theta)$ satisfies $\big(\sum_{i \in S} (\fh(i) - \fh'(i))^2 \big)^{1/2} \leq \eps$.
\label{def:partial-chow-param-problem}
\end{definition}

Note that we do not require $\fh(i)$ and $\fh'(i)$ to be close for $i \notin S$, and indeed this would be impossible for any algorithm to achieve (for example, the target LTF $f$ could be any LTF in the extreme case where $S = \emptyset$). We also note that the Partial Chow Parameters Problem generalizes the Chow Parameters Problem, since the latter is simply the special case of the former where $S = \set{0, 1, \ldots, n}$.

\subsubsection{The Partial Shapley Parameters Problem}

We begin by defining the notion of the Shapley indices. Given a permutation $\pi$ mapping $[n]$ to $[n]$, let $x(\pi, i) \in \pmo^n$  be the string that has a $1$ in each coordinate $j$ with $\pi(j) < \pi(i)$ and a $-1$ in all other coordinates.  Define $x^+(\pi, i)~\in~\pmo^n$ to be $x(\pi, i)$ but with the $i$th coordinate flipped from $-1$ to $1$.

\begin{definition}\label{def:shapley-values}
Given a monotone function $f: \pmo^n \to \pmo$, the $i$th (generalized) \emph{Shapley index} of $f$ is the value
\[
\fc(i) := \Ex_{\bpi \sim \mathbb{S}_n}[f(x^+(\bpi, i)) - f(x(\bpi, i))].
\]
\end{definition}
Thus for a non-constant, monotone LTF $f$, $\fc(i)$ is the probability that, starting from $x =(-1)^n$ and flipping randomly chosen coordinates of $x$ that are $-1$ one at a time to $1$, $i$ is the unique \emph{pivotal index} for which flipping $x_i$ changes $f(x)$ from $-1$ to $1$. The (approximate) \emph{Partial Shapley Indices Problem} is defined analogously to the Partial Chow Parameters problem:

\begin{definition}
The $\eps$-approximate \emph{Partial Shapley Indices Problem} is the promise problem defined as follows. Given $\set{(i, \fc(i)) : i \in S}$ for some LTF $f : \pmo^n \to \pmo$ and some $S \subseteq \set{1, \ldots, n}$ as input, output weights $w_1, \ldots, w_n$ and a threshold $\theta$  such that $f'(x) := \sign(w \cdot x - \theta)$ satisfies $\big(\sum_{i \in S} (\fc(i) - \fc'(i))^2 \big)^{1/2} \leq \eps$.
\label{def:partial-shapley-param-problem}
\end{definition}

\subsection{Our results}

Our first main result is an efficient algorithm for the Chow parameters version of the Partial Inverse Power Index Problem:~\begin{theorem}[Informal statement] \label{thm:main-chow-intro}
There is a $\poly(n, 2^{\poly(1/\eps)})$-time algorithm for the $\eps$-approximate Partial Chow Parameters Problem.
\end{theorem}
The algorithm of \Cref{thm:main-chow-intro} is an ``EPRAS,'' meaning that  its running time is a fixed polynomial in $n$ independent of $\eps$, but depends super-polynomially on $\eps$.

Our second main result is an efficient algorithm for the Shapley parameters version of the Partial Inverse Power Index Problem:~\begin{theorem}[Informal statement] \label{thm:main-shapley}
{There is a $2^{((\log n)/\eps)^c}$-time algorithm for some absolute constant $c > 0$ for the $\eps$-approximate Partial Shapley Indices Problem.}
\end{theorem}

Here our algorithm is an ``EQPRAS,'' meaning that its running time is a fixed quasi-polynomial function of $n$ for any value of $\eps$, but depends super-polynomially on $\eps.$

\subsection{Our techniques for the Chow problem}

We begin by giving a high level overview of our algorithm (and associated proof of correctness) for the partial Chow parameters problem. The techniques for the corresponding problem for Shapley indices build on the techniques for the Chow problem. 

We begin by recalling the  important notion of \emph{regularity} of a linear form  (e.g., see~\cite{Servedio:07cc}).  A linear form $w\cdot x$ (where $w \neq 0^n$) is said to be \emph{$\tau$-regular} if $\max_j |w_j| / \Vert w \Vert_2 \le \tau$. Regularity plays a crucial role in Boolean function analysis because of the Berry-Ess\'{e}en theorem, which says that for $\bx \sim \{-1,1\}^n$ the random variable $w \cdot \bx -\theta$ ``behaves like a Gaussian with mean $\theta$ and variance $\Vert w \Vert_2$''. In fact, the Berry-Ess\'{e}en theorem can be used to establish analogous statements whenever the $n$-dimensional random variable $\bx$ comes from a product distribution with bounded third moments. 

Moving on to halfspaces, the notion of regularity has played a crucial role in their analysis ever since it was first used in \cite{Servedio:07cc} (though implicitly it was used in the earlier works of \cite{KKMO07,MOO10}). The reason this notion is useful for us is as follows: Suppose $w \in \R^n$ is a $\tau$-regular vector and $f(x)=\sign(w \cdot x - \theta)$ is a corresponding LTF. Then a result first proven in \cite{MORS:10} (but which essentially can be derived from \cite{KKMO07}) is that there exists an (explicit) constant $c_\theta$ (depending on $\theta$) such that 
\begin{equation}~\label{eq:closeness-chow}
\sum_{i=1}^n (\widehat{f}(i) - c_\theta w_i)^2 = O(\tau^{\frac 1 2})
\end{equation} (see \Cref{prop:affine-weights-v1}).
 In other words, for $\tau$-regular LTFs, the Chow parameters are (close to) a linear rescaling of the LTF's weights. 
 
 Now, suppose we were given the promise, in the partial Chow parameters problem, that the target LTF $f$ is in fact $\tau$-regular for $\tau = O(\epsilon^2)$ (where $\epsilon$ is the desired accuracy parameter for the reconstruction). Then \Cref{eq:closeness-chow} suggests a very simple algorithm for the partial Chow parameters problem in this case:
\begin{enumerate}
\item While the threshold parameter $\theta$ is not known, we can \emph{guess} it. What this means precisely is the following:  since (without loss of generality) we may assume that $\Vert w \Vert_2=1$, it must be the case that $\theta \in [-\sqrt{n}, \sqrt{n}]$. 
It is not difficult to show that if instead of having the exact value of $\theta$, we had it up to an additive $\pm \delta$, this adds only a small inaccuracy (roughly, $O(\tau + \delta)$) to the error of the final reconstruction. Thus, what we can do is to try out all possible values of $\theta$ in a grid over $[-\sqrt{n}, \sqrt{n}]$ where the width of the grid is some sufficiently small $\delta$. While performing such a guessing step means that we will have a batch of several candidate LTFs, it is straightforward to do \emph{hypothesis testing} at the end, by simply estimating the Chow parameters of each hypothesis LTF and outputting the one which most closely matches the input, to identify a successful candidate from the batch.

(More generally, several times in this informal description of our algorithms we will employ such a ``guessing of parameters.'' Suppose that the total number of parameters is $P$ and that the grid space for each parameter is $L$; then iterating over all the possibilities and the subsequent hypothesis testing adds a multiplicative running time overhead of $\approx L^P$. Thus, as long as $P$ is small and $L$ is not too large, the total overhead incurred from guessing parameters is small. In the rest of this informal overview, whenever we ``guess a parameter", we will assume that we have its value exactly and not account for (i) either the small inaccuracy due to the granularity of the grid or (ii) multiplicative overhead created by iterating over the possibilities.)

\item Given the parameter $\theta$, we can explicitly compute the constant $c_\theta$. Given the value $c_\theta$, the most obvious approach to the partial Chow parameters problem is to define the quantities $v_1, \ldots, v_n$ as follows: For $i \in S$, we define $v_i := \frac{\widehat{f}(i)}{c_\theta}$. 
We further define $\mathsf{wt} := \sum_{i \in S} v_i^2$ and define $v_i = \sqrt{(1-\mathsf{wt})/(n-|S|)}$. Finally, we output the halfspace $g(x) = \sign (v \cdot x -\theta)$. The intuition behind this is that for coordinates $i \in S$,  (\ref{eq:closeness-chow}) suggests the correct value of $w_i$ (which is what we set $v_i$ to be). For all the other coordinates, we set $v_i$ to be ``as regular as possible".\footnote{Actually, in a slight deviation from what is described above, our actual algorithm for the regular case performs a slight numerical adjustment to avoid the pathological case in which $\mathsf{wt}$ slightly exceeds 1, which would make our algorithm meaningless.} Given (\ref{eq:closeness-chow}), it easily follows that $\sum_{i \in S}(\widehat{g}(i) -\widehat{f}(i))^2 = O(\tau^{\frac12})$, which is $O(\epsilon)$ by our choice of parameters. 
\end{enumerate}

To handle the case when the unknown LTF $f$ is not $\tau$-regular, we use the ``critical index'' machinery of Servedio~\cite{Servedio:07cc}. To explain how this is done, for ease of exposition let us assume that $f = \sign (w\cdot x - \theta)$ is such that $|w_1| \ge \ldots \ge |w_n|$. The $\tau$-critical index of the vector $w$ (or equivalently, of any linear form $w \cdot x-\theta$) is the smallest index $j$ such that $|w_{j+1}| / \sqrt{\sum_{k >j} w_k^2} \le \tau$. Thus, a vector $w$ is $\tau$-regular if and only if its $\tau$-critical index is zero. If the $\tau$-critical index of a vector $w$ is not defined, then we say it is $\infty$.  

We now discuss the general algorithmic strategy for the partial Chow parameters problem; as explained below, the strategy depends on the value of the critical index. 
(While the actual value of the critical index is not known to the algorithm, the algorithm can just ``guess" which of the following three cases it is in, followed by  hypothesis testing at the very end.) 

\begin{enumerate}
\item {\bf First case:  $\tau$-critical index is large:}  This is the case when the $\tau$-critical index $K(\tau)$ is at least $O(\tau^{-2} \log (1/\tau))$. In this case, Servedio~\cite{Servedio:07cc} shows that $f$ is $O(\tau)$-close in Hamming distance to a LTF $g$ which depends only on $K(\tau)$ variables. As $f$ and $g$ are close in Hamming distance, it follows that $\sum_{i \in S}(\wh{f}(i) - g(i))^2 = O(\tau)$. The algorithm in this case simply enumerates over all LTFs on $K(\tau)$ variables -- there are $2^{O(K^2(\tau) \log K(\tau))}$ such LTFs -- and  for each such LTF $g$, checks if it is a solution to the partial Chow parameters problem. 

\item {\bf Second case:  $\tau$-critical index is zero:} This is the case where the linear form $w \cdot x - \theta$ is $\tau$-regular. We have already described the reconstruction algorithm in this case. 

\item {\bf Third case:  $\tau$-critical index is small:} This is the case when the $\tau$ critical index is non-zero but is at most $O(\tau^{-2} \log (1/\tau))$.  This case, which is technically the most challenging, combines ingredients from the large and zero critical index cases. Let us assume that the unknown weights are $w_1, \ldots, w_n$ and that the critical index is $K(\tau)$. First of all, the algorithm will guess $K(\tau)$ (note that there are only $O(\tau^{-2} \log (1/\tau))$ possibilities). The algorithm will also guess $w_1, \ldots, w_{K(\tau)}$. Finally, given $\{\wh{f}(i)\}_{i \in S}$, the algorithm will also guess the subset $T = \{j \in S : j \ge K(\tau)\}$. (Note that the number of choices for $T$ which must be considered can be bounded by $K(\tau)$, since the weights before the critical index (the largest magnitude weights) must correspond to the largest Chow parameters.)  Having fixed all these choices, the crucial fact, analogous to \Cref{eq:closeness-chow}, is that there exists (an explicitly computable) $c = c(\theta, w_1, \ldots, w_{K(\tau)})$ such that $\sum_{i  \ge K} (\widehat{f}(i) -  c w_i)^2  = O(\tau^{1/2})$ (see \Cref{prop:reg}). The algorithm can now compute $c$ and finding a feasible $w_{K}, \ldots, w_n$ is essentially the same as case (ii), i.e.~the zero critical index case. 
\end{enumerate}

Finally, we note that the actual algorithm and its analysis is split into two cases, namely, the large versus small critical index cases (and not three cases as described above). In particular, the zero critical index case is subsumed by the small critical index case. However, the small critical index case is both conceptually and technically a combination of the ideas for the zero critical index and the large critical index cases. Thus, for expository reasons, in this introduction we have split the analysis into three cases.

While we are glossing over several technical details, the actual algorithm and analysis essentially follows the above description. We now turn to giving a high level overview of the techniques for the partial Shapley value problem.

\subsection{Our techniques for the Shapley problem.} \label{sec:shapley-techniques}

At the highest level, the structure of our algorithm for the Shapley problem is similar to our algorithm for the Chow problem --- a case split based on whether the critical index is large, zero, or small --- but the analysis and underlying structural results are considerably more involved. 
(Similar to the Chow problem, the actual algorithm and its analysis has only two cases,  the large and the small critical index. However, for the sake of exposition, similar to the Chow problem we describe a three case split here in the introduction.)

Let $g=\sign(v \cdot x - \theta')$ be the target LTF; as an initial pre-processing step, we argue (\Cref{thm:discretization}) that $g$ is close to an LTF $f=\sign(w \cdot x - \theta)$ in which all weights $w_i$ are not-too-large integer multiples of some fixed ``granularity'' value.  We then proceed with a case analysis based on the $\tau^\ast$-critical index (for a suitable regularity parameter $\tau^\ast$) of the LTF $f$.  As with the Chow problem, the algorithm proceeds differently depending on whether the target LTF $f$ has large, zero or small $\tau^\ast$-critical index. 

A significant challenge that arises in analyzing these cases for the Shapley problem is the fact that the probabilistic definition of the Shapley indices is much less ``clean'' than the definition of the Chow parameters.  
Recall that the $i$-th Chow parameter $\widehat{f}(i)$ is defined to equal $\Ex_{\bx \sim \bn}[f(\bx) \cdot \bx_i]$; the fact that the underlying distribution --- uniform over $\bn$ --- is a \emph{product} distribution makes this definition particularly amenable to harmonic analysis and the application of various tools from probability theory.
The Shapley indices, on the other hand, do not admit such a clean definition in terms of a product distribution. 
However, in an attempt to get a syntactically similar definition,  \cite{journals/geb/DeDS17} showed that the $i$-th Shapley index $\fc(i)$ is equal to $\alpha \cdot f^{\ast}(i) + \beta,$ where $\alpha$ and $\beta$ are fixed values (depending only on $n$ and not on $f$) and $f^\ast(i) = \Ex_{\bx \sim \DShap}[f(\bx) \cdot \bx_i]$, where $\DShap$ is a certain \emph{symmetric} distribution supported on $\{-1,1\}^n$. 
Here ``symmetric'' means that the distribution $\DShap$ is invariant under permutation of coordinates (the probability that $\DShap$ assigns to a string depends only on the number of 1's in the string).  
A significant technical complication is that $\DShap$ is not a product distribution, and thus several technical tools that are used to analyze the Chow parameters, and that rely on the product distribution structure of the uniform distribution over $ \bn$, are no longer available. 

In order to adapt our algorithm for the partial Chow parameters problem to the partial Shapley problem, the main technical statement that is required is that if $f = \sign(w\cdot x-\theta)$ is such that $w$ is $\tau$-regular, then the Shapley indices of $f$ are close to being an affine form of the weights. %
More precisely, we we want to prove that
 there are values  $\ensuremath{\accentset{\diamond}{A}} = \ensuremath{\accentset{\diamond}{A}}(\theta, \Vert w \Vert_1)$ and $\ensuremath{\accentset{\diamond}{B}}= \ensuremath{\accentset{\diamond}{B}}(\theta , \Vert w \Vert_1)$ such that
\begin{equation}~\label{eq:shapley-error}
\sum_{i=1}^n (\fc(i) - (\ensuremath{\accentset{\diamond}{A}}w_i + \ensuremath{\accentset{\diamond}{B}}))^2 \leq \varepsilon(\tau),
\end{equation}
where $\varepsilon(\tau) \rightarrow 0$ as $\tau \rightarrow 0$.
To prove this, we first show (\Cref{lem:approximation-of-shapley2}) that 
the Shapley distribution $\DShap$ can be approximated by a convex combination of $p$-biased product distributions $u^n_p$ on the hypercube (here $u^n_p$ is the product distribution in which each marginal $\bx_i$ has $\Pr[\bx_i = 1]=p$ and $\Pr[\bx_i = -1]=1-p$). 
While the distribution $\DShap$ cannot be exactly expressed as a convex combination of $p$-biased distributions on the cube, we show that for any parameter $\delta>0$, we can express $\fc(i)$ as a ``positive linear combination" of $\big\{\big(f^{\ast}_p(i)- \frac{\sum_{j=1}^n f^{\ast}_p(j)}{n}\big)\big\}_{p \in [\delta, 1-\delta]}$ up to an error of at most $O(n^2\delta)$. 
Here $f^{\ast}_p(i) = \mathbf{E}[f(\bx) \cdot \bx_i]$ where $\bx \sim u^n_{p}$, the $p$-biased distribution on the cube. As $\delta \rightarrow 0$, the error of approximating $\{\fc(i)\}$ goes to zero, but the ``positive linear" coefficients of $f^{\ast}_p(i)$ (for small values of $p$) diverge to infinity, thus rendering the expression meaningless.  We evade these difficulties by not allowing $\delta$ to be too close to 0; more precisely, we choose $\delta$ to be a particular $1/\poly(n)$ value, which ensures that $\fc(i)$ can be expressed as a positive linear combination of $\big\{\big(f^{\ast}_p(i)- \frac{\sum_{j=1}^n f^{\ast}_p(j)}{n}\big)\big\}_{p \in [\delta, 1-\delta]}$ up to an error of $o(1)$.

Establishing \Cref{eq:shapley-error} now reduces to showing 
that there are values $\ensuremath{\accentset{\diamond}{A}_p} = \ensuremath{\accentset{\diamond}{A}_p} (\Vert w \Vert_1, \theta, p)$ and $\ensuremath{\accentset{\diamond}{B}_p} = \ensuremath{\accentset{\diamond}{B}_p} (\Vert w \Vert_1, \theta, p)$ such that 
\begin{equation}~\label{eq:shapley-error-1}
\sum_{i=1}^n \bigg( \bigg(f^{\ast}_p(i)- \frac{\sum_{j=1}^n f^{\ast}_p(j)}{n}\bigg)-(\ensuremath{\accentset{\diamond}{A}_p}w_i + \ensuremath{\accentset{\diamond}{B}_p})\bigg)^2 = \varepsilon_p(\tau) , 
\end{equation}
where $\varepsilon_p(\tau)  \rightarrow 0$ as $\tau \rightarrow 0$. We note that when $p=1/2$, this follows from our analysis for the Chow problem.  We carry out a careful adaptation of the machinery developed in the context of LTF analysis for the uniform distribution ($p=1/2$), including results from \cite{MORS:10, DDS16}, to show  that for any $p \in (0,1)$, we have
\[
\sum_{i=1}^n \bigg( \bigg(f^{\ast}_p(i)- \frac{\sum_{j=1}^n f^{\ast}_p(j)}{n}\bigg)-(\ensuremath{\accentset{\diamond}{A}_p}w_i + \ensuremath{\accentset{\diamond}{B}_p})\bigg)^2  = O(\sqrt{\tau}). 
\]
This finishes the sketch of our high level approach for
establishing \Cref{eq:shapley-error}. 

We now turn to giving an overview of the algorithmic part. As with the partial Chow parameters problem, we choose a suitable value $\tau^\ast = \tau^\ast(\eps)$ of the regularity parameter (depending on the desired final accuracy $\eps$),  and the algorithmic strategy depends on whether the $\tau^\ast$-critical index is zero, ``large",  or ``small." As before, the case when the critical index is small is essentially a combination of the first two cases, so in the rest of this intuitive overview, we will just give the high level idea of the algorithmic strategy for the ``large" and zero critical index cases. 

\begin{enumerate}

\item {\bf Case 1:  $\tau^\ast$-critical index is large:} Similar  to the partial Chow parameters problem discussed earlier, we would like to argue that that for a suitable threshold $K = K(\tau^\ast)$, if the $\tau^\ast$-critical index of a LTF $f$ is larger than $K$ then $f$ is close to a LTF $f'$ on $K$ variables under $\DShap$. While the fact that $\DShap$ is not a product distribution presents some obstacles, we are able to leverage anti-concentration of certain linear forms under $\DShap$ (proved in \cite{journals/geb/DeDS17}) to argue that if the critical index is larger than essentially $K := O(\log n /(\tau^\ast)^2)$, then $f$ is close (under $\DShap$) to a junta on $K$ variables. Then, as in the partial Chow parameters problem, one can find a suitable LTF by just brute force search over all LTFs on $K$ variables. Note that the threshold $K$ has a $\log n$ dependence on $n$; this is in contrast with the partial Chow parameters problem, where the corresponding cutoff for ``large'' critical index is independent of $n$. This is a bottleneck that results in our algorithm for the partial Shapley problem running in quasipolynomial time (whereas for the partial Chow parameters problem the running time is polynomial in $n$). 

\item {\bf Case 2:  $\tau^\ast$-critical index is zero:} As stated at the beginning of this subsection, we can assume that all the weights in the target LTF are not-too-large integral multiples of some fixed granularity parameter $\gamma^{\ast}$, and thus we can also assume that the threshold $\theta$ is also an not-too-large integer multiple of $\gamma^\ast$. The algorithm guesses two parameters, namely $\theta$ and $W= \Vert w \Vert_1$; its analysis will exploit the fact that there are only polynomially many  possibilities for these parameters. Given these parameters, the algorithm can exactly compute the constants $\ensuremath{\accentset{\diamond}{A}} = \ensuremath{\accentset{\diamond}{A}}(\theta, \Vert w \Vert_1)$ and $\ensuremath{\accentset{\diamond}{B}}= \ensuremath{\accentset{\diamond}{B}}(\theta , \Vert w \Vert_1)$  from \Cref{eq:shapley-error}. Now let $S \subseteq [n]$ be the set of indices for which the algorithm is given Shapley indices. The algorithmic problem now reduces to finding a set of weights $w_1, \ldots, w_n$ to 
$$
\textrm{Minimize} \sum_{i \in S} (\fc(i) - (\ensuremath{\accentset{\diamond}{A}}w_i + \ensuremath{\accentset{\diamond}{B}}))^2 \ \ \text{subject to}  \ \ 
\sum_{i=1}^n w_i = W; \ \ \sum_{i=1}^n w_i^2 =1; \ \ \max_{1 \le i \le n} w_i \le \tau^\ast; \ \ 
$$
These constraints are non-linear and non-convex, and thus not amenable to techniques from convex programming in any obvious way. However, we show that by exploiting the granularity of the weights $w_i$ (recall that all of them are integral multiples of $\gamma^\ast$), it is possible to use a simple dynamic programming approach to solve this problem. 
\end{enumerate}

\subsection{Organization}

\Cref{sec:preliminaries} gives basic preliminary definitions and results on LTFs, regularity, various notions of distance between functions that we will use, and various distributions that will arise in our analysis.
\Cref{sec:prelim-gaussian-pbiased} gives background results from Gaussian analysis and $p$-biased Fourier analysis of LTFs, and \Cref{sec:useful-fourier} generalizes various technical results on Fourier analysis of regular LTFs under the uniform distribution from \cite{MORS:10,DDS16} to the $p$-biased case.  \Cref{section:head-and-tails} extends some of these results to the case of general LTFs by doing an analysis that works separately with the ``head'' portion and the (regular) ``tail'' portion of a general LTF.  \Cref{sec:chow} combines the $p=1/2$ case of these structural results with algorithmic arguments to prove \Cref{thm:main-chow-intro}, our main result for the Partial Chow Parameters problem, and \Cref{sec:shapley-intro} uses the general-$p$ version of these results (with additional analytic and algorithmic arguments) to prove \Cref{thm:main-shapley}, our main result for the Partial Shapley Indices problem.

\section{Background} \label{sec:preliminaries}

\subsection{Linear threshold functions, regularity, and critical index}

\begin{table}[th!]
\begin{center}
\def\arraystretch{1.5}
\begin{tabular}{|l|l|l|}
\hline
\textbf{Notation} & \textbf{Definition} & \textbf{Description}\\ \hline \hline
$\fh(i)$ & $\Ex_{\bx \sim \pmo^n}[f(\bx) \cdot \bx_i]$ & The $i$th Fourier coefficient/Chow parameter of $f$. \\ \hline
$\fh(i, p)$ & $\Ex_{\bx \sim u_p^n}[f(\bx) \cdot \psi_p(\bx_i)]$ & The $i$th $p$-biased Fourier coefficient of $f$. \\ \hline
$\fs(i, p)$ & $\Ex_{\bx \sim u_p^n}[f(\bx) \cdot \bx_i]$ & The $i$th $p$-biased coordinate correlation coefficient of $f$. \\ \hline
$\ft(i)$ & $\Ex_{\bx \sim N(0, 1)^n}[f(\bx) \cdot \bx_i]$ & The $i$th Hermite coefficient of $f$. \\ \hline
$\fc(i)$ & $\Ex_{\bpi \sim \mathbb{S}_n}[f(x^+(\bpi, i)) - f(x(\bpi, i))]$ & The $i$th Shapley index (value) of $f$. \\ \hline
$\fb(i)$ & $\Ex_{\bx \sim \DShap}[f(\bx) \cdot L_i(\bx)]$ & The $i$th Shapley Fourier coefficient of $f$. \\ \hline
$\fs(i)$ & $\Ex_{\bx \sim \DShap}[f(\bx) \cdot \bx_i]$ & The $i$th Shapley coordinate correlation coefficient of $f$. \\ \hline
\end{tabular}
\end{center}
\caption{Quantities associated with a linear threshold function $f : \pmo^n \to \pmo$ and index $i \in [n]$.}
\label{tbl:coefficients}
\end{table}

We recall that a linear threshold function (LTF) is a function $f: \bn \to \bits$ defined by $f(x)=\sign(w \cdot x - \theta)$ for some $w \in \R^n$, $\theta \in \R$ where $\sign(t)=1$ iff $t \geq 0.$
We say that a nonzero vector $w \in \R^n$ is \emph{$\tau$-regular} if $\norm{\vec{w}}_{\infty}/\norm{\vec{w}}_2 \leq \tau$, and we say that an LTF $f(x) = \sign(w \cdot  x - \theta)$ is $\tau$-regular if its weight vector $w$ is $\tau$-regular.

A key ingredient in our proofs is the notion of the \emph{critical index} of an LTF. The critical index was implicitly introduced and used in \cite{Servedio:07cc} and
was explicitly used in \cite{DS13,DGJ+10:bifh,os2011} and other works. Intuitively, the critical index of $w$ is the first index $i$ such that the sub-vector of $w$ obtained by deleting the $i$ largest-magnitude entries of $w$ is regular. A precise definition follows:
\begin{definition}[critical index]
Given a vector $w \in \R^n$ such that $\abs{w_1} \ge \ldots \ge \abs{w_n} > 0$, for $k \in [n]$ we
denote by $\sigma_k$ the quantity 
$\sqrt{\sum_{i=k}^n w_i^2}$. We define the $\tau$-critical index $c(w, \tau )$ of $w$ as the smallest index
$i \in [n]$ for which $\abs{w_i} \le  \tau \sigma_i$. If this inequality does not hold for any $i \in [n]$, we define $c(w, \tau ) = \infty.$
\label{def:critical-index}
\end{definition}

Finally, we will use the following lemma, which appears in a number of previous works. The result says that, for weight vectors $w$ with sorted weights, $\sigma_k=\sqrt{\sum_{i=k}^n w_i^2}$, also denoted as $\tail_k(w)$, decreases geometrically for $i$ less than the critical index. 
\begin{fact}[Fact 25 \cite{de2014})]
Let $w = (w_1, \ldots, w_n) \in \R^n$ be such that $\abs{w_1} \geq \cdots \geq \abs{w_n}$, and let $1 \leq a \leq b \leq c(w, \tau)$, where $c(w, \tau)$ is the $\tau$-critical index of $w$. Then $\tail_b(w) < (1 - \tau^2)^{(b - a)/2} \cdot \tail_a(w)$.
\label{lem:geom-decreasing-tail}
\end{fact}
\begin{proof}
By definition of the critical index, $\abs{w_i} > \tau \cdot \tail_i(w)$ for $i < c(w, \tau)$. Therefore for such an $i$, $\tail_i(w)^2 = w_i^2 + \tail_{i+1}(w)^2 > \tau^2 \cdot \tail_i(w)^2 + \tail_{i+1}(w)^2$, and so $\tail_{i+1} (w)< (1 - \tau^2)^{1/2} \cdot \tail_i(w)$. The result follows by applying this last inequality repeatedly.
\end{proof}

\subsection{Boolean functions and distance measures}

We assume familiarity with the basics of standard Fourier analysis of Boolean functions with respect to the uniform distribution over $\bn$, see \Cref{ap:fourier} for a brief overview.  (Later in this preliminaries section we will introduce notions of Fourier analysis with respect to other distributions such as product distributions and the ``Shapley distribution.'')

We will use a range of different notions of distance between Boolean functions $f,g : \bn \to \bits$.
Let 
\[
d(f, g) := \Prx_{\bx \sim \bn}[f(\bx) \neq g(\bx)]
\] 
denote the Hamming distance between $f$ and $g$.
Let 
\[
\dChow(f, g) := \left(\sum_{i=1}^n \big(\fh(i) - \gh(i)\big)^2 \right)^{1/2}\]
denote the \emph{Chow distance} between $f$ and $g$, and let 
\[
\dPC{S}(f, g) := \left(\sum_{i \in S} \big(\fh(i) - \gh(i)\big)^2 \right)^{1/2}\]
denote the \emph{partial Chow distance} between $f$ and $g$ with respect to a subset of indices $S \subseteq \set{0, 1, \ldots, n}$. 
We similarly define
\[
\dShapley(f, g) := \left(\sum_{i=1}^n \big(\fc(i) - \gc(i)\big)^2\right)^{1/2},
\]
the \emph{Shapley distance} between $f$ and $g$, and
\[
\dPS{S}(f, g) := \left(\sum_{i \in S} \big(\fc(i) - \gc(i)\big)^2\right)^{1/2},
\]
the \emph{partial Shapley distance} between $f$ and $g$,

It is clear that $\dPC{S}(f, g) \leq \dChow(f, g)$ and
$\dPS{S}(f, g) \leq \dShapley(f, g)$ for any $S \subseteq \set{0, 1, \ldots, n}$.
The following simple result relates Hamming distance and Chow distance:
\begin{proposition}[{\cite[Proposition 1.5]{os2011}}]
$\dChow(f, g) \leq 2 \sqrt{d(f, g)}$.
\label{prop:chow-ub-by-closeness}
\end{proposition}
\begin{proof}
For $f,g:  \pmo^n \to \pmo$ we have $d(f,g) = \frac{1}{4}\E[(f(\bx)-g(\bx))^2]= \frac{1}{4}\sum_{S\subseteq [n] }(\fh(S)-\gh(S))^2 \ge \frac{1}{4}\sum_{i\in [n]}(\fh(i)-\gh(i))^2=\frac{1}{4}\dChow(f,g)^2$, and hence $\dChow(f, g) \leq 2 \sqrt{d(f, g)}$.
\end{proof}

\subsection{Some useful distributions}

\subsubsection{The Shapley distribution \texorpdfstring{$\DShap$}{D\_Shap} and ``Fourier analysis'' for this distribution}

\cite{journals/geb/DeDS17} introduced a distribution over $\bn$, called the ``Shapley distribution'' (we write $\DShap$ for this distribution though it is denoted by $\mu$ in \cite{journals/geb/DeDS17}), which is very useful for analysis of the Shapley indices.
We recall the definition of this distribution: let $Q(n, k) := 1/k + 1/(n-k)$ for $0 < k < n$, and let $\Lambda(n) := \sum_{0 < k < n} Q(n, k) = 2H_{n-1}$, where $H_n$ denotes the $n$th harmonic number.
The distribution $\DShap$ over $\bn$ is defined as follows: it has support $\pmo^n \setminus \set{(-1)^n,1^n}$.
To sample a string $\bx \sim \DShap$, first sample $k \in \set{1, \ldots, n - 1}$ with probability $Q(n, k)/\Lambda(n)$. Then choose $\bx$ uniformly from the weight $k$ slice of the hypercube $\pmo^n$ (i.e.~the set of all ${n \choose k}$ many strings in $\bn$ with exactly $k$ many 1's).

Following~\cite{journals/geb/DeDS17}, we proceed to define a ``Fourier basis'' under the distribution $\mu$. We define the inner product $\iprod{f, g}_\mu := \E_{\rv{x} \sim \DShap}[f(\rv{x})g(\rv{x})]$, and we define orthonormal functions $L_i: \bn \to \R$ for $i = 0, 1, \ldots, n$ so that $\iprod{L_i, L_j}_\mu = 1$ if $i = j$ and $\iprod{L_i, L_j}_\mu = 0$ if $i \neq j$.
As shown in~\cite[Lemma 9]{journals/geb/DeDS17}, we can take $L_0(x)  \equiv 1$ and $L_i(x) = a \cdot (\sum_{j = 1}^n x_i) + b x_i$ for some values of $a = a(n)$ and $b = b(n)$ satisfying $a = -\Theta(\sqrt{\log n}/n)$ and $b = \Theta(\sqrt{\log n})$. Accordingly, we define \emph{Shapley Fourier coefficients} and \emph{Shapley Fourier distance} with respect to $\mu$ as follows. The $i$th Shapley Fourier coefficient for $i = 0, 1, \ldots, n$ is defined as
\[
\fb(i) := \Ex_{\bx \sim \DShap}[f(\bx) \cdot L_i(\bx)] \ ,
\]
and the Shapley Fourier distance between two LTFs $f$ and $g$ is defined as
\begin{equation} \label{eq:shapley-fourier}
\dFourier(f, g) := \Big(\sum_{i=0}^n (\fb(i) - \gb(i))^2 \Big)^{1/2} \ .
\end{equation}

\subsubsection{$p$-biased distributions and Fourier analysis}
We write $u_p$ to denote the $p$-biased distribution over $\bits$, i.e.~a random variable distributed according to $u_p$ takes the value $+1$ with probability $p$ and takes the value $-1$ with probability $1 - p$. Let 
\[
\mu_p := 2p - 1 \quad \quad \text{~and~}\quad \quad\sigma_p := 2 \sqrt{p(1-p)}
\] 
denote the mean and standard deviation respectively of such a random variable. 
We define $\psi_p: \bits \to \R$,
\[
\psi_p(x) := \frac{x - \mu_p}{\sigma_p},
\]
so if $\rv{x} \sim u_p$ is a $p$-biased random variable then $\psi_p(\rv{x})$ has mean 0 and variance 1.
We will overload the above notation, defining $\psi_p^{[w]}: \bits \to \R$,
\[
\psi_p^{[w]}(x) := \frac{x - \mu_p \cdot \sum_i w_i}{\sigma_p\normtwo{w}} \quad \text{ for }w\in\R^n,
\]
which gives that
\[
\psi_p^{[w]}(x) =  \frac{x - \mu_p \cdot \normone{w}}{\sigma_p\normtwo{w}} \quad \text{ for }w\in\R_{\ge 0}^n.
\]
\begin{fact}[Scaling Property]\label{fact:phi-equality}
We have that $\psi_p^{[w]}(x)=  \psi_p^{[\tfrac{w}{\sigma_p\normtwo{w}}]}(\tfrac{x}{\sigma_p\normtwo{w}})$.
\end{fact}

The $2^n$ functions $\{\lambda_{S,p}(x) := \prod_{i\in S}\psi_p(x_i)\}_{S \subseteq [n]}$ are easily seen to constitute an orthonormal basis for the vector space of all real-valued functions on $\bn$ under the distribution $u_p^n.$ We write $\fh(S,p)$ to denote the corresponding 
$p$-biased Fourier coefficients of a real-valued function $f$ under $u_p^n$, and we write the $p$-biased degree-1 coefficient as $\fh(i,p)$ rather than $\fh(\{i\},p)$. When $p=1/2$ and we are working with the uniform distribution, we simply write $\fh(S)$ or $\fh(i).$

\subsubsection{Gaussian distributions and Hermite analysis}

Let $N(\mu, \sigma^2)$ denote the Gaussian distribution with mean $\mu$ and variance $\sigma^2$.
We recall that the $n$-variable \emph{Hermite polynomials} $\{H_S\}_{S \in \N^n}$ form a complete orthonormal basis for the vector space of all square-integrable functions under the standard $n$-dimensional Gaussian distribution $N(0,1)^n$.  
 We write $\ft(S)$ to denote the $S$-th Hermite coefficient of a real-valued function $f$ under $N(0,1)^n$, and we will be particularly interested in $f$'s degree-1 coefficients, i.e., $\ft(e_i)$, where $e_i$ is the vector which is 1 in the $i$-th coordinate and 0 elsewhere.  See \Cref{ap:fourier} for a brief overview of the key notions.

\subsection{Miscellaneous notation, terminology, and inequalities} \label{sec:misc}

We recall that a function $f :\pmo^n\to\pmo$ is said to be a \emph{junta} on $J \subseteq [n]$ if $f$ only depends on the coordinates in $J$. If $|J|=k$ we say that $f$ is a \emph{$k$-junta}.

Following \cite{journals/geb/DeDS17}, we say that an LTF $f: \pmo^n \to \pmo$, $f(x) = \sign(w\cdot x - w_0)$ with $w \in \R^n$ is \emph{$\eta$-restricted} if $w_0 \in [-(1 - \eta)\norm{w}_1, (1 - \eta) \norm{w}_1]$.  When $\eta$ is small (as it will be in our Shapley result) this is a mild technical condition on the LTF $f$ (which was also present in \cite{journals/geb/DeDS17}).

We write ``$a \stackrel{k}{\approx} b$'' to indicate that $|a - b| \leq O(k)$.  For $v \in \R^n$ we write ``$\|v\|$'' to denote the 2-norm $(v_1^2 + \cdots + v_n^2)^{1/2}.$

At various point in our analysis we will need some useful but routine inequalities; we record these in \Cref{ap:inequalities}.

\section{Useful Fourier analytic results on $p$-biased Chow parameters of LTFs}
\label{sec:useful-fourier}

\subsection{Preliminary results from Gaussian analysis and $p$-biased Fourier analysis}
\label{sec:prelim-gaussian-pbiased}

\subsubsection{Background on LTFs and linear forms under the Gaussian distribution}
\label{sec:background:LTF-Gaussian}

Let $\phi$ denote the p.d.f. of a standard normal Gaussian $N(0, 1)$ and let $\Phi$ denote the corresponding c.d.f..  We extend the latter notation by writing $\Phi[a,b]$ to denote $\Phi(a)-\Phi(b)$, allowing $b<a$, and we will use the estimate $|\Phi[a,b]|\le |b-a|$ without comment.

Following \cite{MORS:10}, let us define the function $m : [-\infty, \infty] \to [-1, 1]$ by 
\begin{equation} \label{eq:m}
m(\theta) := \left(2 \int_{\theta}^{\infty} \phi(x) dx\right)-1
\end{equation}
and the function $W : [-1, 1] \to [0, 2/\pi]$ by 
\begin{equation} \label{eq:W}
W(\nu) = (2 \phi(m^{-1}(\nu)))^2
\end{equation} (the latter is well defined since the function $m$ is monotone decreasing with range $[-1,1]$; we remark that $W$ is a function symmetric about $0$, with a peak at $W(0) = 2/\pi$). To motivate these definitions, we observe that $m(\theta)$ corresponds 
to the expectation $\E_{\bx \sim N(0,1)}[h_\theta(\bx)]$ of the univariate function $h_\theta(x) = \sign(x-\theta).$  It is easily verified that 
\begin{equation} 
\label{eq:useful}
\tilde{h}_\theta(1) = \E[h_\theta(\bx)\bx] = 2\phi(\theta) \quad \quad \text{and}
\quad \quad
W(\E[h_\theta]) = \tilde{h}_\theta(1)^2.
\end{equation} The intuition is that given as input the expected value of some $h_\theta$, the function $W$ outputs the squared degree-1 Hermite coefficient of $h_\theta.$ This motivates the following definition, which will be useful later:

\begin{definition} \label{def:alpha}
Let $\alpha(\theta) := \sqrt{W(m(\theta))}$. View $[n]$ as partitioned into $[n] = H \sqcup T$.  For $p \in (0,1)$ and $w=(w_H,w_T) \in \R^n,$ let 
\begin{equation} \label{def:alpha-function}
\alpha(\theta, w_H,w_T,p) := \E_{\brho \sim u^{|H|}_p} [\alpha(\psi^{[w_T]}_p(\theta - w_H \cdot \brho))].
\end{equation} 
\end{definition}

Combining the above observations with the rotational invariance of $N(0,1)^n$, it is straightforward to establish the following (see Proposition~25 of \cite{MORS:10} for a proof):

\begin{fact}[Hermite Properties of LTFs]\label{fact:hermite-properties}
Let $f : \R^n \to \{-1, 1\}$ be an LTF $f(x) = \sign(w \cdot x -\theta)$, where
${w} \in \R^n$ has $\|w\|=1.$  Then the degree-0 and degree-1 Hermite coefficients of $f$ satisfy the following properties:
\begin{enumerate}
    \item $\ft(0) = \E_{\bx\sim N(0,1)^n}[f(\bx)] = m(\theta)$;
    \item 
    $\ft({e_i})=\sqrt{W(\E_{\bx \sim N(0,1)^n}[f(\bx)])}w_i
    =\sqrt{W(m(\theta))}w_i$;
    \item $\sum_{i=1}^n \ft(e_i)^2 = W(\E_{\bx \sim N(0,1)^n}[f(\bx)])$.
\end{enumerate}
\end{fact}

We further recall the following useful properties of the functions $m$ and $W$ (see Proposition~24 of \cite{MORS:10} for the simple proof):

\begin{proposition}\label{proposition:gaussian properties}
\begin{enumerate}
    \item $\E_{\bx \sim N(0,1)}[|\bx-\theta|]=2\phi(\theta)-\theta m(\theta)$;
    \item $|m^{\prime}|\le \sqrt{2/\pi}$ everywhere and $|W^\prime|<1$ everywhere;
    \item If $|\nu|=1-\eta$ then $W(\nu)=\Theta(\eta^2\log(1/\eta))$.
\end{enumerate}
\end{proposition}

\subsubsection{Gaussian versus $p$-biased linear forms}

The main reason why the Gaussian distribution is useful for us is because if $w$ is a regular linear form, then the distribution of $w \cdot \bx$ (when $\bx$ is uniform random over $\bn$ or is drawn from the $p$-biased distribution $u^n_p$) can be well approximated in c.d.f.~distance by a suitable Gaussian.  This is a consequence of the well-known Berry-Esseen theorem, which gives quantitative error bounds on the central limit theorem; in this subsection we state this fundamental result along with a range of consequences and extensions of it which we will use.

\begin{theorem}[Berry-Esseen Theorem, \cite{Feller}]
\label{thm:berry-esseen}
Let $\rv{X}_1, \ldots, \rv{X}_n$ be independent real-valued random variables with $\E[\rv{X}_i] = 0$, $\E[\rv{X}_i^2] = \sigma_i^2 > 0$ and $\rho_i = \E[\abs{\rv{X}_i}^3] < \infty$ for each $i \in [n]$. Let $\sigma = \left(\sum_{i=1}^n \sigma_i^2\right)^{1/2}$ and let $\rho = \sum_{i=1}^n \rho_i$. Let $F(x)$ denote the cumulative distribution function of $\sigma^{-1} \cdot \sum_{i=1}^n \vec{\rv{X}}_i$. Then for all $t \in \R$, it holds that
$
\abs{F(t) - \Phi(t)} \leq \frac{\rho}{\sigma^3},
$
or in more detail,
\[
| F(x)-\Phi(x)| \le C \cdot \frac{\rho}{\sigma^3} \cdot \frac{1}{1+|x|^3}
\]
for all real $x$, where $C$ is an absolute constant.
\end{theorem}

The following is a fairly straightforward consequence of the Berry-Esseen Theorem, and is essentially a $p$-biased version of~\cite[Fact 2.6]{DDS16}; it says that the value of a regular linear form with input sampled from $u_p^n$ is distributed like a Gaussian up to some small error. For completeness we give the proof in \Cref{ap:BE}.

\begin{fact} \label{fct:p-biased-linear-form}
Let $0^n \neq \vec{w} \in \R^n$ be $\tau$-regular, and let $p \in (0, 1)$. The we have the following:

\begin{enumerate}

\item For any interval $[a, b] \subseteq \R \union \set{\pm \infty}$, 
\[
\left|\Prx_{\vec{\rv{x}} \sim u_p^n}\left[\vec{w} \cdot \vec{\rv{x}} \in [a, b]\right] - \left(\Phi\left(\frac{b - \mu}{\sigma}\right) - \Phi\left(\frac{a - \mu}{\sigma}\right) \right)\right| \leq \frac{4 \tau}{\sigma_p} \ ,
\]
where $\mu = \mu_p\cdot\sum_{i=1}^n w_i$ and $\sigma = \sigma_p \cdot \norm{\vec{w}}_2$.

\item For any $\lambda$ and any $\theta\in\R$, we have
\[
\Prx_{\bx \sim u^n_p}\left[\left| w \cdot \bx-\theta \right|\le \lambda\right]\le 2\frac{\lambda}{\sigma_p\normtwo{w}}+2\frac{\tau}{\sigma_p}.
\]
In particular, if $\lambda=O(\tau)$ and $\|w\|_2=1,$ then we have
\[\Pr[|\displaystyle\sum_i w_i\rv{x_i}-\theta|\le \lambda]\le \frac{O(\tau)}{\sigma_p}.\]

\end{enumerate}

\end{fact}

As a $p$-biased analogue of the (simple) Proposition~31 of \cite{MORS:10}, we note that the Berry-Esseen theorem lets us easily approximate the expected value of a regular LTF under the $p$-biased distribution:

\begin{lemma}\label{lemma:mean-approximation}
For $f({x})=\sign(w \cdot \bx -\theta)$  a $\tau-$regular LTF, we have
$\Ex_{\bx \sim u^n_p}[f(\bx)] \overset{\tfrac{\tau}{\sigma_p}}{\approx} m( \psi_p^{[w]}(\theta) ).$
\end{lemma}

We also have a $p$-biased analogue of the (more involved) Proposition~32 of \cite{MORS:10}, which gives an approximation for the expected magnitude of the linear form ${w}\cdot{\bx}-\theta$ itself under the $p$-biased distribution (see \Cref{ap:BE} for the proof):

\begin{lemma}\label{lemma:marginal-approximation}
For $w$ a $\tau$-regular LTF, we have
\[
\Ex_{\bx \sim u^n_p}[| w \cdot \bx-\theta|] \overset{\tau\normtwo{{w}}}{\approx}  \normtwo{{w}}\sigma_p 
\Ex_{\bx \sim N(0,1)}\Big[|\bx-\psi_p^{[w]}(\theta)|\Big]=\normtwo{{w}}\sigma_p\Big(2\phi(\psi_p^{[w]}(\theta))-\psi_p^{[w]}(\theta)m(\psi_p^{[w]}(\theta))\Big).
\]
\end{lemma}

\medskip

{\bf Bivariate statements.} For technical reasons we will also require a two-dimensional analogue of \Cref{fct:p-biased-linear-form}. The proof, which uses a multivariate extension of the Berry-Esseen theorem, is sketched in \Cref{ap:BE} and is a $p$-biased generalization of Theorem~68 of~\cite{MORS:10}.

\begin{fact}\label{fct:d-clt:p-biased}
Let $\bx \sim u^n_p$ be a $p$-biased random vector in $\bn$, and let $\by$ be a random vector in $\bn$ that is $\rho$-correlated with $\bx$ (meaning that each coordinate $\by_i$ is independently set to equal $\bx_i$ with probability $\rho$ and is set to a random draw from $u_p$ with probability $1-\rho$) for some $\rho$ that is bounded away from 1. Let $w \in \R^n$ be $\tau$-regular, and let $\ell(x)$ denote the linear form $\sum_{i=1}^{n}w_ix_i$.
Then for any two intervals $[a,b]$ and $[c,d]$ in $\R$, we have
\[
\Big|\Pr[(\ell(\rv{x}),\ell(\rv{y})) \in [a,b]\times[c,d]]-
\Phi_{0,V}\Big([\psi_p^{[w]}(a),\psi_p^{[w]}(b)]\times  
[\psi_p^{[w]}(c),\psi_p^{[w]}(d)]\Big)
\Big|\le O\left(\frac{\tau}{\sigma_p}\right),
\]
where $V=\begin{bmatrix}
1&\rho\\
\rho&1
\end{bmatrix}$ and $\Phi_{0,V}$ denotes the distribution of the bivariate Gaussian with zero mean and covariance matrix $V$.
\end{fact}

\subsection{A structural theorem on regular LTFs under the $p$-biased distribution}

The following is a $p$-biased variant of Theorem~48 of \cite{MORS:10}; intuitively, it says that the level-1 Fourier weight of a regular $p$-biased LTF is captured by the $W(\cdot)$ function that was introduced in \Cref{sec:background:LTF-Gaussian}.

\begin{theorem}\label{theo:48th:pbiased}
Let $f_1: \bn \to \bits$ be a $\tau$-regular linear threshold function. Then
\[
\left|\sum_{i=1}^{n}\fh(i,p)^2-W\left(\Ex_{\bx \sim u^n_p}[f_1(\bx)] \right)\right|\le \bigO{\sqrt{\tfrac{\tau}{\sigma_p}}}.
\]
Further, suppose that $f_2:\bn \to \bits$ is another $\tau$-regular linear threshold function that can be expressed using the same linear form as $f_1$, i.e., $f_k=\sign({w}\cdot{x}-\theta_k)$ for some ${w},\theta_1,\theta_2$ and $k=1,2$. 
Then
\[
\left|\left(\sum_{i=1}^{n}\fh_1(i,p)\fh_2(i,p)\right)^2-W\left(\Ex_{\bx \sim u^n_p}[f_1(\bx)]\right)W\left(\Ex_{\bx \sim u^n_p}[f_2(\bx)]\right)
\right|\le \bigO{\sqrt{\tfrac{\tau}{\sigma_p}}}.
\]
\end{theorem}

\begin{proof}
We first note that we may assume that $\sqrt{\tau/\sigma_p}$ is bounded below 1, since otherwise the claimed bounds hold for trivial reasons.
Using \Cref{lemma:mean-approximation}, we have that for $k=1,2,$
\begin{equation}
\label{eq:aa}\Ex_{\bx \sim u^n_p}[f_k(\rv{x})]=\Ex_{\bx \sim u^n_p}[\sign(w \cdot \bx -\theta_k)] \overset{\tfrac{\tau}{\sigma_p}}{\approx}m( \psi_p^{[w]}(\theta) ). 
\end{equation}

Let $\bx \sim u^n_p$ and let $\rv{y} \in \bn$ be $\rho$-correlated with $\bx$ (as described in the statement of \Cref{fct:d-clt:p-biased}) where $\rho=\sqrt{{\tau}/{\sigma_p}}$ is bounded away from 1.
We have that
\[
\Ex[f_1(\rv{x})f_2(\rv{y})]=\Prx[f_1(\rv{x})=f_2(\rv{y})]-\Prx[f_1(\rv{x})\neq f_2(\rv{y})] =2\Prx[(w\cdot\rv{x},w\cdot\rv{y})\in A\cup B]-1\]
where $A=[\theta_1,+\infty)\times[\theta_2,+\infty)$ and $B=(-\infty,\theta_1]\times (-\infty,\theta_2]$. Applying \Cref{fct:d-clt:p-biased} and recalling that $w$ is $\tau$-regular, we have that
\[
\ProbaOf{(w\cdot\rv{x},w\cdot\rv{y})\in A\cup B}
\overset{\tfrac{\tau}{\sigma_p}}{\approx}
\ProbaOf{(\rv{X},\rv{Y})\in \widetilde{A}\cup \widetilde{B}}
\]
where $(\widetilde{\theta_1},\widetilde{\theta_2})=(\psi_p^{[w]}(\theta_1),\psi_p^{[w]}(\theta_2))$ and $\tilde{A}= [\widetilde{\theta_1},+\infty)\times[\widetilde{\theta_2},+\infty)$ and $\tilde{B}=(-\infty,\widetilde{\theta_2}]\times (-\infty,\widetilde{\theta_2}]$
and $(\rv{X},\rv{Y})$ are $\rho$-correlated $N(0,1)$ Gaussians.
 
It follows that 
\[
\Ex[f_1(\rv{x})f_2(\rv{y})]
\overset{\tfrac{\tau}{\sigma_p}}{\approx}
\Ex[h_{\widetilde{\theta_1}}(\rv{x}),h_{\widetilde{\theta_2}}(\rv{y})],
\]
where $h_{\widetilde{\theta_1}}(\cdot)$ is the function of one Gaussian variable defined as $h_{\widetilde{\theta_1}}(t):=\sign(t-{\widetilde{\theta_1}})$. 
Using the Fourier and Hermite expansions of $f_k$ and $h_{\widetilde{\theta_k}}$ and the fact that $\bx,\by$ are $\rho$-correlated, we may rewrite the above approximate equality as:
\begin{align*}
& \fh_1(\emptyset,p)
\fh_2(\emptyset,p)
+\rho \cdot \left(\sum_{i=1}^{n}
\fh_1(i,p)
\fh_2(i,p)\right)
+ \left(\sum_{|S|\ge 2}
\rho^{|S|}\fh_1(S,p) \fh_2(S,p)\right)\\
&\overset{\tfrac{\tau}{\sigma_p}}
{\approx}
\widetilde{h_{\widetilde{\theta_1}}}(0)
\widetilde{h_{\widetilde{\theta_2}}}(0)
+\rho 
\widetilde{h_{\widetilde{\theta_1}}}(1)
\widetilde{h_{\widetilde{\theta_2}}}(1)
+ \left(\sum_{k\ge 2}
\rho^{k}
\widetilde{h_{\widetilde{\theta_1}}}(k)
\widetilde{h_{\widetilde{\theta_2}}}(k)
\right).
\end{align*}
Now by Cauchy-Schwarz (and using the fact that $\rho\ge 0$) we have
\begin{align*}
    \left|\sum_{|S|\ge 2}\rho^{|S|}\fh_1(S,p)\fh_2(S,p)\right| &\le
    \sqrt{\sum_{|S|\ge 2}\rho^{|S|}\fh_1(S,p)^2}
    \sqrt{\sum_{|S|\ge 2}\rho^{|S|}\fh_2(S,p)^2}\\
    &\le\rho^2
    \sqrt{\sum_{|S|\ge 2}\fh_1(S,p)^2}
    \sqrt{\sum_{|S|\ge 2}\fh_2(S,p)^2}\\
    &\le\rho^2
    \sqrt{\sum_{|S|\ge 0}\fh_1(S,p)^2}
    \sqrt{\sum_{|S|\ge 0}\fh_2(S,p)^2} \le\rho^2.
\end{align*}

By a similar analysis, we have that
\begin{align*}
    \left|\sum_{k\ge 2}
\rho^{k}
\widetilde{h_{\widetilde{\theta_1}}}(k)
\widetilde{h_{\widetilde{\theta_2}}}(k)
\right|
    &\le
 \rho^2.
\end{align*}

We further have by \Cref{lemma:mean-approximation} that
\[
\widetilde{h_{\widetilde{\theta_k}}}(0)=\Ex_{\bx \sim N(0,1)}
[h_{\widetilde{\theta_k}}(\bx)]=m(\widetilde{\theta_k})
\overset{\tfrac{\tau}{\sigma_p}}{\approx}\Ex_{\bx \sim u^n_p}[f_k(\bx)]=\fh_k(\emptyset,p),
\]
and hence
\[
\rho \cdot \left(\sum_{i=1}^{n}
\fh_1(i,p)
\fh_2(i,p)\right)
\overset{\tfrac{\tau}{\sigma_p}+\rho^2}{\approx}
\rho \ \ 
\widetilde{h_{\widetilde{\theta_1}}}(1)
\widetilde{h_{\widetilde{\theta_2}}}(1)=
\rho \cdot \ 2\phi(\widetilde{\theta_1})\ \cdot\ 2\phi(\widetilde{\theta_2}),
\]
where the equality is by \Cref{eq:useful}.
Dividing by $\rho$ and using 
${\tfrac{\tau}{\rho \sigma_p}+\rho}\approx \tfrac{\tau^{1/2}}{\sigma_p^{1/2}}$ in the error estimate,
we get
\[
\sum_{i=1}^{n}
\fh_1(i,p)
\fh_2(i,p)
\overset{\sqrt{\tfrac{\tau}{\sigma_p}}}{\approx}
 \ \ 2\phi(\widetilde{\theta_1})\ \cdot\ 2\phi(\widetilde{\theta_2})
 =\sqrt{ W(m(\widetilde{\theta_1})) \cdot W(m(\widetilde{\theta_2}))}
\]
where the equality is by \Cref{eq:useful}.

Since we may apply this with $f_1$ and $f_2$ both equal to $f_k$, we may also conclude that
\[
\sum_{i=1}^{n}
\fh(i,p)^2
\overset{\sqrt{\tfrac{\tau}{\sigma_p}}}{\approx}
W(m(\widetilde{\theta_k})).
\]
Using the mean value theorem, the fact that $\abs{W'}\le 1$ on $[-1,1]$, and \Cref{eq:aa}, we can conclude that
\[
\sum_{i=1}^{n}
\fh_k(i,p)^2
\overset{\sqrt{\tfrac{\tau}{\sigma_p}}}{\approx}
 W(\expectation{f_k}),
\]
giving the first required approximate equality.
Similar reasoning yields that
\[
\left(\sum_{i=1}^{n}
\fh_1(i,p)
\fh_2(i,p)\right)^2
\overset{\sqrt{\tfrac{\tau}{\sigma_p}}}{\approx}
W(m(\widetilde{\theta_1}))\ \cdot\ W(m(\widetilde{\theta_2}))
\]
and the proof is complete.
\end{proof}

\subsection{$p$-biased Chow parameters are proportional to weights for regular LTFs}

The following is a $p$-biased analogue of Lemma~6.11 of \cite{DDS16}; intuitively, it says that for a regular LTF, the vector of weights is close (after a suitable scaling) to the vector of degree-1 Fourier coefficients.

\begin{proposition}\label{prop:affine-weights-v1}
Let $f: \bn \to \bits,$ $f(x) = \sign(w \cdot x - \theta)$ where $w$ is $\tau$-regular and $\norm{w}_2 = 1$.
Then
\[
\displaystyle\sum_{i \in [n]} (\fh(i,p) -  \alpha(\psi_p^{[w]}(\theta)) w_i)^2 \leq \bigO{\sqrt{\tfrac{\tau}{\sigma_p^2}}}.
\]

\end{proposition}

\begin{proof}
First we fix some notation:  we will write $\mu:=\displaystyle\sum_{i=1}^{n}w_i\mu_p$, and we observe that with this notation we have  $\psi_p^{[w]}(x) = \frac{x-\mu}{\sigma_p}$. Recalling that $\psi_p(x_i) = \frac{x_i - \mu_p}{\sigma_p},$ we begin by noting that
    \begin{align*}
        \sign(w \cdot x - \theta) &= 
         \sign\left( \left(\sum_{i=1}^n w_i \psi_p(x_i) \right) - \psi_p^{[w]}(\theta)\right).
    \end{align*} 
    As we will see, the latter expression is convenient because it contains a linear combination of functions $\psi_p(x_1),\dots,\psi_p(x_n)$ which are orthonormal under the $u_p^n$ distribution.
    
We consider two cases depending on the magnitude of $\psi_p^{[w]}(\theta) \sigma_p.$

\medskip

\noindent {\bf Case 1:}  The first case is that $|\psi_p^{[w]}(\theta) \sigma_p|>\sqrt{2\ln(4/\tfrac{\tau}{\sigma_p})}$, or equivalently, $-\frac{(\psi_p^{[w]}(\theta) \sigma_p)^2}{2}<\ln(\tfrac{\tau}{\sigma_p}/4)$.
In this case since $f$ is $\pm 1$-valued, we have $1-\E[f(\bx)]^2 =4\Pr[f(\bx) = 1]\Pr[f(\bx) = -1] \leq 4 \Pr[f(\bx)=1]$, and since
\[
1-\fh(\emptyset,p)^2=\displaystyle\sum_{|S|>0}\fh(S,p)^2\ge\displaystyle\sum_{i \in [n]}\fh(i,p)^2,
\]
this yields
\[
\displaystyle\sum_{i\in[n]}\fh(i,p)^2\le 4\Pr\left[  \sum_{i\in [n]}w_i \psi_p(x_i) \ge \psi_p^{[w]}(\theta) \right]. 
\]

		By Hoeffding's inequality and the assumption on $|\psi_p^{[w]}(\theta) \sigma_p|$ that put us in Case~1, we have that
\[
\Pr\left[  \sum_{i\in [n]}w_i \psi_p(x_i) \ge \psi_p^{[w]}(\theta) \right] \le \exp\left({\dfrac{-2(\psi_p^{[w]}(\theta))^2}{4\normtwo{w}^2/\sigma_p^2}}\right)\le \frac{1}{4} \cdot \frac{\tau}{\sigma_p},
\]
so consequently we have that  
\[  \displaystyle\sum_{i\in[n]}\fh(i,p)^2 \le \frac{\tau}{\sigma_p}. \]

		On the other hand, by \Cref{fact:hermite-properties} and the definition of $\alpha(\cdot)$ we also have that $\displaystyle\sum_{i\in[n]} \Big(  \alpha(\psi_p^{[w]}(\theta)) w_i \Big)^2= W(m(\psi_p^{[w]}(\theta))).$
  	    Applying  \Cref{lemma:mean-approximation}, we get that
\[
\Big|\Ex_{\bx \sim u^n_p}[f(\bx)]-m(\psi_p^{[w]}(\theta))\Big|\le O\left({\frac{\tau}{\sigma_p}}\right).
\] 
Observing that that the function $W(\cdot)$ is a contraction (it satisfies $|W'|<1$), we get that
\[\left|W\left(\Ex_{\bx \sim u^n_p}[f(\rv{x})]\right)-W(m(\psi_p^{[w]}(\theta)))\right|    \le O\left({\frac{\tau}{\sigma_p}}\right).\] 
Recalling \Cref{theo:48th:pbiased}, we further have that:
            \[
            \left|W\left(\Ex_{\bx \sim u^n_p}[f(\rv{x})]\right)- \displaystyle\sum_{i\in[n]}\fh(i,p)^2 \right|    \le O\left(\sqrt{\frac{\tau}{\sigma_p}}\right).\] 
Putting the pieces together and applying the inequality $(a-b)^2\le 2(a^2+b^2)$, we get that

            \[
            \displaystyle\sum_{i \in [n]} (\fh(i,p) -  \alpha(\psi_p^{[w]}(\theta)) w_i)^2 \le 
            2 \displaystyle\sum_{i \in [n]} (\fh(i,p))^2 +
            2 \displaystyle\sum_{i \in [n]} (\alpha(\psi_p^{[w]}(\theta)) w_i)^2
            \le \bigO{\sqrt{\frac{\tau}{\sigma_p}}}
            \]
as desired.

\medskip

\noindent {\bf Case 2:}  The remaining case is that $|\psi_p^{[w]}(\theta) \sigma_p|<\sqrt{2\ln(4/\tfrac{\tau}{\sigma_p})}$.
		
		To show that $\displaystyle\sum_{i \in [n]} (\fh(i,p) -  \alpha(\psi_p^{[w]}(\theta)) w_i)^2 \le \epsilon$
		, it suffices to show that 
		\begin{equation} \label{eq:goal}
		\displaystyle\sum_{i \in [n]} (\fh(i,p))^2 +
        \displaystyle\sum_{i \in [n]} (\alpha(\psi_p^{[w]}(\theta)) w_i)^2
		\approxEqual{\epsilon} 2 \displaystyle\sum_{i \in [n]} (\fh(i,p))(\alpha(\psi_p^{[w]}(\theta)) w_i)
		\end{equation}
		
		The analysis just given for Case~1 lets us control the LHS of \Cref{eq:goal} as
		\[
		\displaystyle\sum_{i \in [n]} (\fh(i,p))^2 +
        \displaystyle\sum_{i \in [n]} (\alpha(\psi_p^{[w]}(\theta)) w_i)^2
        =
		\displaystyle\sum_{i \in [n]} (\fh(i,p))^2 +
        \alpha(\psi_p^{[w]}(\theta))^2 \approxEqual{\sqrt{\frac{\tau}{\sigma_p}}} 2W(m(\psi_p^{[w]}(\theta))).
		\]
Turning to the RHS, we have that
		\begin{equation} \label{eq:potato}
		2\displaystyle\sum_{i \in [n]} \fh(i,p)\alpha(\psi_p^{[w]}(\theta)) w_i=
		2\alpha(\psi_p^{[w]}(\theta)) \displaystyle\sum_{i \in [n]} \fh(i,p) w_i=
		2\alpha(\psi_p^{[w]}(\theta)) \Ex_{\bx \sim u^n_p} \left[f(\rv{x})
		\sum_{i \in [n]} w_i \psi_p(\rv{x}_i) \right].
		\end{equation}
We can re-express the expectation above as
		\begin{align*}
		 \Ex_{\bx \sim u^n_p} \left[f(\rv{x})
		\sum_{i \in [n]} w_i \psi_p(\rv{x}_i) \right]&=
		\expectation{f(\rv{x})\left(\sum_{i \in [n]} w_i \psi_p(\rv{x}_i) - \psi_p^{[w]}(\theta)\right)}+
		\psi_p^{[w]}(\theta)\expectation{f(\rv{x})}\\
		&=
		\expectation{    \sign\left(\sum_{i \in [n]} w_i \psi_p(\rv{x}_i) - \psi_p^{[w]}(\theta)\right)
        \left(\sum_{i \in [n]}w_i \psi_p(\rv{x}_i) - \psi_p^{[w]}(\theta)\right)} \\ 
        & \ \ \ \ +
		\psi_p^{[w]}(\theta)\expectation{f(\rv{x})}\\
		&=
		\expectation{|\ell(\rv{x})|}+
		\psi_p^{[w]}(\theta)\expectation{f(\rv{x})},
		\end{align*}
where we write $\ell(x)$ to denote
\[
		    \ell(x):= \left(\sum_{i=1}^n w_i \psi_p(x_i) \right) - \psi_p^{[w]}(\theta).
		 \]
		By \Cref{lemma:marginal-approximation} we have that
		\[\E[|\ell(\rv{x})|]=\E\left[\left|\displaystyle\sum_{i=1}^{n} \frac{w_i}{\sigma_p\normtwo{w}}\bx_i -\frac{\theta}{\sigma_p\normtwo{w}} \right|\right]
		 \overset{\tau}{\approx} 
		 \Ex_{\bx \sim N(0,1)}\left[\left|\bx-\psi_p^{[\frac{w}{\sigma_p\normtwo{w}}]}\left(\frac{\theta}{\sigma_p\normtwo{w}}\right)\right|\right]
		 \overset{\text{\Cref{fact:phi-equality}}}{=}  \Ex\Big[|\bx-\psi_p^{[w]}(\theta)|\Big].
 \]
Recalling the last equality of \Cref{lemma:marginal-approximation}, this gives that
 \[
 \Ex_{\bx \sim u^n_p}[|\ell(\bx)|]  \approxEqual{\tau}  \Big( 2\phi(\psi_p^{[w]}(\theta))-\psi_p^{[w]}(\theta)m(\psi_p^{[w]}(\theta) \Big).\]

By \Cref{lemma:mean-approximation} we have that
		\[\psi_p^{[w]}(\theta)\Ex_{\bx \sim u^n_p}[f(\rv{x})]\approxEqual{\psi_p^{[w]}(\theta)\tfrac{\tau}{\sigma_p}}\psi_p^{[w]}(\theta)m(\psi_p^{[w]}(\theta)).
		\]
Putting these pieces together, we can re-express \Cref{eq:potato} as 
\begin{align*}
		2\alpha(\psi_p^{[w]}(\theta)) \Ex_{\bx \sim u^n_p} \left[f(\rv{x})
		\sum_{i \in [n]} w_i \psi_p(\rv{x}_i) \right]&=
		2\alpha(\psi_p^{[w]}(\theta))\Big(\expectation{|\ell(\rv{x})|}+
		\psi_p^{[w]}(\theta)\expectation{f(\rv{x})}\Big)\\
		&\approxEqual{(\max_{\theta \in \R}\{\alpha(\theta)\}) \cdot (\psi_p^{[w]}+1) \cdot \tfrac{\tau}{\sigma_p}} 2\alpha(\psi_p^{[w]}(\theta))\cdot \Big(2\phi(\psi_p^{[w]}(\theta))\Big)\\
		&\approxEqual{(\psi_p^{[w]}+1) \cdot \tfrac{\tau}{\sigma_p}} 2W(m(\psi_p^{[w]}(\theta))),
\end{align*}
where for the last line we used the definitions of $\alpha$ and $W$ and the fact that $W$ is uniformly bounded by $\sqrt{2/\pi}$.
		
Since the above analyses of the LHS and the RHS show that each of these can be approximated by the same quantity $2W(m(\psi_p^{[w]}(\theta))),$ we deduce that
		\[
				\displaystyle\sum_{i \in [n]} (\fh(i,p))^2 +
        \displaystyle\sum_{i \in [n]} (\alpha(\psi_p^{[w]}(\theta)) w_i)^2
		\approxEqual{(\psi_p^{[w]}+1) \cdot \tfrac{\tau}{\sigma_p} + \sqrt{\tfrac{\tau}{\sigma_p}}} 2 \displaystyle\sum_{i \in [n]} (\fh(i,p))(\alpha(\psi_p^{[w]}(\theta)) w_i).
		\]
It remains only to verify that
		\[
 	(\psi_p^{[w]}+1)\cdot \tfrac{\tau}{\sigma_p} + \sqrt{\tfrac{\tau}{\sigma_p}}
\leq 
O\left(
\frac{1}{\sigma_p}
\sqrt{\log\left({\frac 1 {\tau/\sigma_p}}\right)} \cdot \frac{\tau}{\sigma_p} + \sqrt{\frac{\tau}{\sigma_p}}
\right)
\leq\bigO{\sqrt{\frac{\tau}{\sigma_p^2}}},
		\]
and the proof is complete.
\end{proof}

The following corollary, which applies to the centralized version of the weight vector $w$ (see the definition immediately before \Cref{fact:centering}) and is an immediate consequence of \Cref{prop:affine-weights-v1} and \Cref{fact:centering},
will be useful for our analysis of the Shapley problem:

\begin{corollary}\label{prop:affine-weights-v1.5}
\[
\displaystyle\sum_{i \in [n]} \Big((\sigma_p\fh(i,p) - \frac{1}{n}\sum_{i\in [n]}\sigma_p\fh(i,p) ) - \sigma_p\alpha(\psi_p^{[w]}(\theta)) (w_i-\frac{1}{n}\sum_{i\in [n]}w_i)\Big)^2 \leq \bigO{\sqrt{\tau}}.
\]
\end{corollary}

\subsection{Structural results on heads and tails of LTFs (Chow version)}
\label{section:head-and-tails}

Let $f(x) = \sign(w \cdot  x - \theta)$ be an LTF, and to simplify the presentation let us assume that its weights are sorted in magnitude from largest to smallest, i.e., $|w_1| \geq |w_2| \geq \cdots \geq |w_n|$. (Of course this need not hold for the algorithmic problems on LTFs that we consider, but this can be assumed without loss of generality for the structural results we are concerned with in this section.)
Let $\tau > 0$. Although $f$ need not be $\tau$-regular, we can always partition its weights $w$ into ``head weights'' $w_H$ and ``tail weights'' $w_T$ such that $w_T$ is $\tau$-regular and any longer suffix of $w$ is not $\tau$-regular. Let $H=\{1,2,\dots,\}$ be the set of indices of head weights, and let $T=[n]\setminus H = \{\dots,n-1,n\}$ be the set of indices of tail weights.

In this section we prove a number of structural results on the $p$-biased Chow parameters of the head and tail variables in an arbitrary LTF.  For the original Chow parameters problem we will only use the $p=1/2$ case of these results, but the more general case of $p \in (0,1)$ will be used later in our approach to the Shapley problem.

\subsubsection{Regular tail weights are proportional to tail Chow parameters}
\label{sec:prop-of-tail}

We first show that there exists a value $\alpha = \alpha(\theta,w_H{,p})$ such that $\fh(i{,p}) \approx \alpha \cdot w_i$ for all $i \in T$. More precisely, we show that the vector of tail weights $w_T$ is proportional to the vector of tail Fourier coefficients $(\fh(i{,p}))_{i \in {T}}$. This characterization will be helpful for recovering the tail weights of an LTF from its tail Fourier coefficients, and as a corollary also gives an approximation of the Fourier weight of $f$ on $T$.

\begin{proposition} \label{prop:reg}
Let $f(\vec{x}) = f(\vec{x}_H, \vec{x}_T) = \sign(\vec{w}_H \cdot \vec{x}_H + \vec{w}_T \cdot \vec{x}_T - \theta)$ where $\vec{w}_T$ is $\tau$-regular and $\norm{\vec{w}_T}_2 = 1$. Fix $p \in (0, 1)$. Then
\[
\sum_{i \in T} (\fh(i{,p}) - { \alpha(\theta, w_H,w_T,p) }  \cdot w_i)^2 \leq { \bigO{\sqrt{\tfrac{\tau}{\sigma_p^2}}}} ,
\]
where {we recall that $\alpha(\theta, w_H,w_T,p) = \E_{\brho \sim {u^k_p}} [\alpha(\psi^{[w_T]}_p(\theta - w_H \cdot \brho))]$} and $\alpha(\theta) = \sqrt{W(\mu(\theta))}$ as in \Cref{def:alpha}.
\label{prop:weights-prop-chows}
\end{proposition}

\begin{proof}
For $\vec{\rho} \in \pmo^{|H|}$, let $f_{\rho}(\vec{x}_T) = f(\vec{\rho}, \vec{x}_T)$ denote $f$ with its head variables fixed to $\vec{\rho}$. Then
\begin{align*}
&\sum_{i \in T} (\fh(i{,p}) - {\alpha(\theta, w_H,w_T,p)}  \cdot w_i)^2\\
&= \sum_{i \in T} \big(\E_{\vec{\rv{\rho}} \sim {u^{\card{H}}_p}}[\fh_{\rv{\rho}}(i{,p}) - {\alpha(\psi^{[w_T]}_p(\theta - w_H \cdot \brho))} \cdot w_i]\big)^2 && \textrm{(By definition)} \\
&\leq \sum_{i \in T} \E_{\vec{\rv{\rho}} \sim {u^{\card{H}}_p}} [(\fh_{\rv{\rho}}(i{,p}) -{\alpha(\psi^{[w_T]}_p(\theta - w_H \cdot \brho))} \cdot w_i)^2] && \textrm{(By Jensen's inequality)} \\
&= \E_{\vec{\rv{\rho}} \sim {u^{\card{H}}_p}} \big[\sum_{i \in T} (\fh_{\rv{\rho}}(i{,p}) -{\alpha(\psi^{[w_T]}_p(\theta - w_H \cdot \brho))} \cdot w_i)^2 \big] && \textrm{(By linearity of expectation)} \\
&\leq \max_{\vec{\rho} \in \bits^{\card{H}}} \sum_{i \in T} (\fh_{\rho}(i{,p}) -{\alpha(\psi^{[w_T]}_p(\theta - w_H \cdot \rho))} \cdot w_i)^2  \\
&\leq { \bigO{\sqrt{\tfrac{\tau}{\sigma_p^2}}}} \ . && \parbox{5.5cm}{(By \Cref{prop:affine-weights-v1} applied to $f_\rho$)}
\end{align*}
\end{proof}

As a corollary, we get that the $\ell_2$-weight of the tail of the Chow parameters $\gamma := \norm{(\fh(i))_{i \in T}}_2$ and the constant of proportionality $\alpha(\theta, w_H{,w_T,p})$ are good approximations of each other for functions $f$ of the form described in \Cref{prop:weights-prop-chows}.
\begin{corollary}
Let $f(x) = f(\vec{x_H}, \vec{x}_T) = \sign(\vec{w}_H \cdot \vec{x}_H + \vec{w}_T \cdot \vec{x}_T - \theta)$ where $\vec{w}_T$ is $\tau$-regular and $\norm{\vec{w}_T}_2 = 1$. Then
$\norm{(\fh(i))_{i \in T}}_2 \stackrel{{(\tau/\sigma^2_p)^{1/4}}}{\approx} \alpha(\theta, \vec{w}_H{,w_T,p})$.
\label{cor:chow-tail-weight}
\end{corollary}

\begin{proof}

\Cref{prop:reg} says that the Euclidean distance between the vectors $(\hat{f}(i))_{i \in T}$ and $\alpha(\theta,w_H{,w_T,p})w_T$ is at most $O({(\tau/\sigma^2_p)^{1/4}}).$  The corollary follows from \Cref{fact:close-l2-norm} since the Euclidean length of the vector $\alpha(\theta,w_H{,w_T,p})w_T$ is $\alpha(\theta,w_H{,w_T,p})$.
\end{proof}

\subsubsection{Preserving the head Chow parameters}
\label{sec:preserving-head}

In this section we show that exchanging the tail weights $w_T$ of an LTF $f$ with other weights $w_T'$ of the same $\ell_1$ and $\ell_2$ norm does not change the head Fourier coefficients $(\fh(i{,p}))_{i \in H}$ by too much. 
This will be helpful for handling instances of the Partial Chow Parameters Problem with missing tail Fourier coefficients.

The following lemma is a generalization of~\cite[Lemma 6.13]{DDS16}, and its proof closely follows the proof given there.
It essentially says that as far as head Fourier coefficients are concerned, when the tail is regular it does not much matter whether the tail variables are $p$-biased Boolean random variables or Gaussian random variables with mean and variance matching the $p$-biased distribution over $\{-1,1\}.$
\begin{lemma} \label{lem:tailcouldbeeither}
Let $f(\vec{x}) = f(\vec{x}_H, \vec{x}_T) = \sign(\vec{w}_H \cdot \vec{x}_H + \vec{w}_T \cdot \vec{x}_T - \theta)$ where $(\vec{w}_H, \vec{w}_T) \in (\R^{\geq 0})^n$, and where $\vec{w}_T$ is $\tau$-regular. Fix $p \in (0, 1)$. 
Recall that $\hat{f}(i,p) = \E_{\vec{\rv{x}} \sim u_p^n}[f(\vec{\rv{x}}) \cdot \psi_p(\rv{x}_i)]$ and let $\check{f}(i,p) = \E_{\vec{\rv{x}}_H \sim u_p^{|H|}, \vec{\rv{x}}_T \sim N(\mu_p, \sigma_p^2)^{|T|}}[f(\vec{\rv{x}}) \cdot \psi_p(\rv{x}_i)]$.
Then $$\sum_{i \in H} ( \hat{f}(i,p) - \check{f}(i,p))^2 \leq O\left(\frac{\tau^2}{\sigma_p^2}\right)\ .$$
\label{lem:p-biased-fourier-hermite}
\end{lemma}

\begin{proof}
Define the functions $f', f'' : \set{-1, 1}^{\card{H}} \to [-1, 1]$ as follows:
\[
f'(\vec{x}_H) = \Ex_{\vec{\rv{x}}_T \sim u_p^{\card{T}}}[f(\vec{x}_H, \vec{\rv{x}}_T)], \quad
f''(\vec{x}_H) = \Ex_{\vec{\rv{x}}_T \sim N(\mu_p, \sigma_p^2)^{|T|}}[f(\vec{x}_H, \vec{\rv{x}}_T)] \ .
\]
Then
\begin{align*}
\sum_{i \in H} ( \hat{f}(i,p) - \check{f}(i,p) )^2 &= \sum_{i \in H} (\hat{f'}(i,p) - \hat{f''}(i,p) )^2 \\
&\leq \sum_{S \subseteq H} (\hat{f}_p'(S) - \hat{f}_p''(S))^2 \\
&= \Ex_{\vec{\rv{x}} \sim u_p^{\card{H}}} (f'(\vec{\rv{x}}) - f''(\vec{\rv{x}}))^2 \\
&\leq \max_{\vec{x} \in \set{-1, 1}^{\card{H}}} (f'(\vec{x}) - f''(\vec{x}))^2 \ ,
\end{align*}
where the first equality holds since $\hat{f'}(i,p) = \hat{f}(i,p)$ and $\hat{f''}(i,p)  = \check{f}(i,p)$ for all $i \in H$ by definition, and the second equality is by Parseval's identity.

To upper bound the last expression, we observe that for every $\vec{\rho} \in \set{-1, 1}^{\card{H}}$ it holds that
\begin{align*}
\abs{f'(\vec{\rho}) - f''(\vec{\rho})} &= 
2 \cdot \big|\Prx_{\vec{\rv{x}}_T \sim u_p^{\card{T}}}[\vec{w}_T \cdot \vec{\rv{x}}_T \geq \theta - \vec{w}_H \cdot \vec{\rho}] -
\Prx_{\vec{\rv{x}}_T \sim N(\mu_p, \sigma_p^2)^{|T|}}[\vec{w}_T \cdot \vec{\rv{x}}_T \geq \theta - \vec{w}_H \cdot \vec{\rho}]\big| \\
& \leq O(\tau/\sigma_p) \ ,
\end{align*}
where the inequality uses the assumption that $\vec{w}_T$ is $\tau$-regular to apply \Cref{fct:p-biased-linear-form}.
\end{proof}

The following theorem is essentially a corollary of \Cref{lem:p-biased-fourier-hermite}. Its proof is similar to~\cite[Lemma 6.15]{DDS16}.
\begin{theorem}
Let $f(\vec{x}) = f(\vec{x}_H, \vec{x}_T) = \sign(\vec{w}_H \cdot \vec{x}_H + \vec{w}_T \cdot \vec{x}_T - \theta)$ and let $g(\vec{x}) = f'(\vec{x}_H, \vec{x}_T) = \sign(\vec{w}_H \cdot \vec{x}_H + \vec{w}_T' \cdot \vec{x}_T - \theta)$ where $(\vec{w}_H, \vec{w}_T) \in (\R^{\geq 0})^n$, and where $\vec{w}_T, \vec{w}_T'$ are $\tau$-regular, satisfy $\norm{\vec{w}_T}_1 = \norm{\vec{w}'_T}_1$, and satisfy $\norm{\vec{w}_T}_2 = \norm{\vec{w}'_T}_2$. Fix $p \in (0, 1)$. As in
\Cref{lem:tailcouldbeeither}, for $h \in \{f,g\}$ recall that $\hat{h}(i,p) = \E_{\vec{\rv{x}} \sim u_p^n}[h(\vec{\rv{x}})\psi_p(\rv{x}_i)]$. Then \[
\sum_{i \in H} (\hat{f}(i,p) - \hat{g}(i,p))^2 \leq O\left(\frac{\tau^2}{\sigma_p^2}\right) \ .
\]
\label{thm:p-biased-head-coeffs}
\end{theorem}

\begin{proof}
As in \Cref{lem:tailcouldbeeither}, for $h \in \{f,g\}$ let $\check{h}(i,p) = \E_{\vec{\rv{x}}_H \sim u_p^{|H|}, \vec{\rv{x}}_T \sim N(\mu_p, \sigma_p^2)^{|T|}}[h(\vec{\rv{x}})\psi_p(\rv{x}_i)]$.
By the triangle inequality we have that
\[
\sqrt{\sum_{i \in H} (\hat{f}(i,p)  - \hat{g}(i,p))^2} \leq 
\sqrt{\sum_{i \in H} (\hat{f}(i,p) - \check{f}(i,p))^2} + \sqrt{\sum_{i \in H} (\check{f}(i,p))^2- \check{g}(i,p))^2} + \sqrt{\sum_{i \in H} (\hat{g}(i,p)- \check{g}(i,p))^2)^2} \ .
\]
We prove the theorem by upper bounding each term on the right hand side. By \Cref{lem:p-biased-fourier-hermite} and the $\tau$-regularity of $\vec{w}_T$ and $\vec{w}_T'$, we have that the first and third terms are upper bounded by $O(\tau/{\sigma_p})$. Furthermore, we claim that the second term is equal to $0$. This follows from the fact that for every $i \in H$,
\begin{align*}
\check{f}(i,p) &= \E_{\vec{\rv{x}}_H \sim u_p^{|H|}, \vec{\rv{x}}_T \sim N(\mu_p, \sigma_p^2)^{|T|}}[\sign(\vec{w}_H \cdot \vec{\rv{x}}_H + \vec{w}_T \cdot \vec{\rv{x}}_T - \theta) \cdot \psi_p(\rv{x}_i)] \\
&= \E_{\vec{\rv{x}}_H \sim u_p^{|H|}, \vec{\rv{x}}_T \sim N(\mu_p, \sigma_p^2)^{|T|}}[\sign(\vec{w}_H \cdot \vec{\rv{x}}_H + \vec{w}_T' \cdot \vec{\rv{x}}_T - \theta) \cdot \psi_p(\rv{x}_i)] \\
&= \check{g}(i,p)\ .
\end{align*}
Here the first and third equalities are by definition. The second equality uses the fact that $\vec{w}_T \cdot \vec{\rv{x}}_T$ and $\vec{w}_T' \cdot \vec{\rv{x}}_T$ are both (exactly) distributed as $N(\mu_p \cdot \norm{\vec{w}}_1, \sigma_p^2 \cdot \norm{\vec{w}}_2^2)$ when $\vec{\rv{x}}_T \sim N(\mu_p, \sigma_p^2)^{|T|}$, which holds since $\norm{\vec{w}_T}_1 = \norm{\vec{w}'_T}_1$ and $\norm{\vec{w}_T}_2 = \norm{\vec{w}'_T}_2$.
\end{proof}

\begin{corollary}
When $p = \half$, \Cref{thm:p-biased-head-coeffs} holds without the assumption that $\norm{\vec{w}_T}_1 = \norm{\vec{w}'_T}_1$.
\label{cor:half-biased-head-coeffs}
\end{corollary}

\begin{proof}
When $p = \half$, $\vec{w}_T \cdot \vec{\rv{x}}_T$ and $\vec{w}_T' \cdot \vec{\rv{x}}_T$ are both (exactly) distributed as $N(0, \norm{\vec{w}}_2^2)$ when $\vec{\rv{x}}_T \sim N(\mu_{\half}, \sigma_{\half}^2)^{|T|}$ since $\mu_{\half} = 0$ and $\sigma_{\half}=1.$
\end{proof}

\section{The Partial Chow Parameters Problem} \label{sec:chow}

In this section we give an EPRAS for solving the Partial Inverse Chow Parameters Problem. Our algorithm leverages a variant of the following structural theorem of~\cite{os2011}. The variant (\Cref{thm:main-structural}) defines a relatively small set of candidate LTFs and asserts that at least one of these must have Chow Parameters that are close to the input Chow Parameters. Our algorithm works by enumerating all of these candidate LTFs, and then checking for each one whether it has Chow Parameters that are close to the input Chow Parameters.

\begin{theorem}[{\cite[Theorem 7.3]{os2011}}]
Let $\eps \in (0, \half)$ and let $\tau = \tau(\eps)$ be a certain $\poly(\eps)$ value.
Let $f$ be an LTF $f(\vec{x}) = f(\vec{x}_H, \vec{x}_T) = \sign(\vec{w}_H \cdot \vec{x}_H + \vec{w}_T \cdot \vec{x}_T -\theta)$
where $H$ contains all indices $i \in [n]$ with $|\fh(i)| \geq \tau^2$.
Then at least one of the following holds:
\begin{enumerate}
\item $f$ is $O(\eps)$-close to a linear threshold function junta $f'$ over $\vec{x}_H$, or \label{en:os-junta}
\item $f$ is $O(\eps)$-close to a linear threshold function $f'$ of the following form:
\begin{equation}
f'(\vec{x}) = f'(\vec{x}_H, \vec{x}_T) = \sign( \sum_{i \in H} v_i \cdot x_i + \gamma^{-1} \sum_{i \in T} \fh(i) \cdot x_i -\theta')
\label{eq:os-ltf1}
\end{equation}
where ${\theta'} \in \R, \vec{v}_{H} = (v_1, \ldots, v_{|H|}) \in \R^{|H|}$ are such that each $v_i$ is an integer multiple of $\sqrt{\tau}/|H|$ and has magnitude at most $2^{O(|H| \log |H|)} \sqrt{\ln(1/\tau)}$,
and where $\gamma := \big(\sum_{i \in T} \fh(i)^2\big)^{1/2}$.
\label{en:os-small-crit-ind}
\end{enumerate}

\label{thm:OS-LTF-approx}
\end{theorem}

We adapt \Cref{thm:OS-LTF-approx} to obtain our main structural result, which is stated below. It is syntactically similar, but has the key conceptual difference that it works for the ``partial information'' case: given only the Chow Parameters of an LTF $f$ corresponding to a \emph{subset} of indices $S \subseteq \set{0, 1, \ldots, n}$, it specifies a relatively small set of LTFs which is guaranteed to include one that is close to $f$ in Partial Chow Distance with respect to $S$.
\begin{theorem}
Let $\eps \in (0, \half)$ and let $\tau = \tau(\eps) = \eps^{1000}$. %
Let $f$ be an LTF $f(\vec{x}) = f(\vec{x}_H, \vec{x}_T) = \sign( \vec{w}_H \cdot \vec{x}_H + \vec{w}_T \cdot \vec{x}_T - \theta)$
where $H$ contains all indices $i \in [n]$ with $|\fh(i)| \geq \tau^2$.
Let $S \subseteq \set{0, 1, \ldots, n}$.
Then one of the following holds:
\begin{enumerate}
\item $\dPC{S}(f, g) \leq O(\eps)$ for some linear threshold function junta $g$ over $\vec{x}_H$, or \label{en:main-junta}
\item $\dPC{S}(f, g) \leq O(\eps)$ for some linear threshold function $g$ of the form
\begin{equation}
g(\vec{x}) = g(\vec{x}_H, \vec{x}_T) = \sign\Big( \sum_{i \in H} v_i \cdot x_i + 
(\gamma')^{-1} \cdot \Big(\sum_{i \in T \cap S} \fh(i) \cdot x_i +
\sum_{i \in T \cap \bar{S}} r \cdot x_i \Big) -\theta'\Big)
\label{eq:main-ltf}
\end{equation}
where ${\theta'} \in \R, \vec{v}_{H} = (v_1, \ldots, v_{|H|}) \in \R^{|H|}$ are such that each $v_i$ is an integer multiple of $\sqrt{\tau}/|H|$ and has magnitude at most $2^{O(|H| \log |H|)} \sqrt{\ln(1/\tau)}$, where $\gamma'$ satisfies $\gamma' = \gamma$ if $T \subseteq S$ and $\gamma \leq \gamma' \leq \gamma + \tau$ if $T \nsubseteq S$ for $\gamma := \big(\sum_{i \in T} \fh(i)^2\big)^{1/2}$,
and where $r := \big(\frac{(\gamma')^2 - \sum_{i \in T \cap S} \fh(i)^2}{|T \cap \bar{S}|}\big)^{1/2}$.
\label{en:main-small-crit-ind} 
\end{enumerate}

\label{thm:main-structural}
\end{theorem}

We note the close analogy between \Cref{thm:OS-LTF-approx} and \Cref{thm:main-structural}, and the similarity between \Cref{eq:os-ltf1} and \Cref{eq:main-ltf}: both have tail weight vectors whose $\ell_2$ norm is equal to 1, and when $T \subseteq S$ the tail weight vectors are the same. The differences are that in~\Cref{eq:main-ltf}, we use the same weight $r$ for all of the variables corresponding to missing tail Chow Parameters, and use $\gamma'$ as a slight overestimate of $\gamma$, the $\ell_2$ norm of the tail Chow parameters of $f$ (the slight overestimate is to ensure that $r$ is the square root of a non-negative number).

The main idea behind the proof of \Cref{thm:main-structural} is to upper bound $\dPC{S}(f, g)$ using the triangle inequality by
\begin{equation}
\dPC{S}(f, g) \leq \dPC{H}(f, f') + \dPC{H}(f', g) + \dPC{T \cap S}(f, g)
\label{eq:triangle-pc-ub}
\end{equation}
for a function $f'$ of the form described in \Cref{eq:os-ltf1} and a function $g$ of the form described in \Cref{eq:main-ltf}, and then to upper bound each term in the right-hand side of \Cref{eq:triangle-pc-ub}.
Roughly speaking, by \Cref{thm:OS-LTF-approx} there is a function $f'$ of the form described in \Cref{eq:os-ltf1} so that $\dPC{H}(f, f')$ will be small; by the head weight stability result described in \Cref{cor:half-biased-head-coeffs} $\dPC{H}(f', g)$ will be small; and by the proportionality of the tail weights and Chow Parameters described in \Cref{prop:weights-prop-chows}, $\dPC{T \cap S}(f, g)$ will be small.

We also crucially rely on the regularity of the tails of the weight vectors of the functions $f'$ and $g$, which is established in \Cref{clm:reg-tail-of-fp} and \Cref{clm:reg-tail-of-g} below. 
Note that we will not explicitly find a function $f'$ of the form described in \Cref{eq:os-ltf1}, but merely use its existence to prove \Cref{thm:main-structural}.

\subsection{Useful facts about tail weights}

We will use the following lower bound, from~\cite{os2011}, on the tail weight of the LTF in \Cref{eq:os-ltf1}. 
\begin{fact}[{\cite[Equation (8.8)]{os2011}}]
Let $f : \pmo^n \to \pmo$ be an LTF which has $d(f,g)= \Omega(\eps^2)$ for every Boolean function $g$ which is a junta on the coordinates $H = \set{i \in [n]: \fh(i) \geq \tau^2}$, for $\tau = \Omega(\eps^{288})$. Let $T = [n] \setminus H$. Then $\gamma = \big(\sum_{i \in T} \fh(i)^2\big)^{1/2} \geq \Omega(\tau^{1/72})$.
\label{fct:tail-chow-param-lower-bound}
\end{fact}

The next two claims establish the regularity of the tails of the weight vectors of the functions $f'$ and $g$ defined in \Cref{eq:os-ltf1} and~\Cref{eq:main-ltf}.

\begin{claim}
The vector $\vec{v}_T$ of tail weights of each function $f'$ of the form in \Cref{eq:os-ltf1} is $O(\tau^{143/72})$-regular.
\label{clm:reg-tail-of-fp}
\end{claim}

\begin{proof}
As stated above, let $\gamma = \big(\sum_{i \in T} \fh(i)^2\big)^{1/2}$. Since $\norm{\vec{v}_T}_2 = 1$ it suffices to upper bound $\norm{\vec{v}_T}_{\infty}$ in order to establish regularity. We then have that
\begin{equation}
\norm{\vec{v}_T}_{\infty} = \max_{i \in T} \frac{|\fh(i)|}{\gamma} \leq \frac{\tau^2}{\Omega(\tau^{1/72})} \leq O(\tau^{143/72}) \ .
\label{eq:os-ltf-tail-reg}
\end{equation}
Here the equality uses the definition of $\vec{v}_T$, and the first inequality uses the assumption that $|\fh(i)| \leq \tau^2$ for all $i \in T$ to upper bound the numerator and uses \Cref{fct:tail-chow-param-lower-bound} to lower bound the denominator.
\end{proof}

\begin{claim} \label{claim:vprimeTregular}
The vector $\vec{v}_T'$ of tail weights of each function $g$ of the form in \Cref{eq:main-ltf} is $O(\tau^{71/144})$-regular.
\label{clm:reg-tail-of-g}
\end{claim}

\begin{proof}
If $T \subseteq S$, then $\vec{v}_T' = \vec{v}_T$ and the result follows by \Cref{clm:reg-tail-of-fp}. So, assume that $T \nsubseteq S$.
Let $\gamma$, $\gamma'$, and $r$ be as in \Cref{thm:main-structural}. By the definition of $r$, $\norm{\vec{v}_T'}_2 = 1$ so it suffices to upper bound $\norm{\vec{v}_T'}_{\infty}$ in order to show regularity. By definition,
\[
\norm{\vec{v}_T'}_{\infty} = (\gamma')^{-1} \cdot \max (\{|\fh(i)| : i \in T \cap S\} \union \{r\}) \ .
\]
By assumption, $\gamma' \geq \gamma$, and so $(\gamma')^{-1} \cdot \max (\{|\fh(i)| : i \in T \cap S\} \leq \gamma^{-1} \cdot \max (\{|\fh(i)| : i \in T \cap S\} \leq O(\tau^{143/72}) \leq O(\tau^{71/144})$ by \Cref{clm:reg-tail-of-fp}.

It remains to upper bound $r/\gamma'$.
Let $n' := \card{T \cap \bar{S}}$. We will  use the fact that
\begin{equation}
\tau^2 \geq \Omega\Big(\frac{1}{\gamma^2} \max_{i \in T} \fh(i)^2\Big) \geq \Omega\Big(\frac{1}{\gamma^2 n'} \sum_{i \in T \cap \bar{S}} \fh(i)^2 \Big) = \Omega\Big(\frac{1}{n'} \cdot \Big(1 - \frac{1}{\gamma^2}\sum_{i \in T \cap S} \fh(i)^2\Big)\Big) \ ,
\label{eq:small-avg-tail-chow}
\end{equation}
where the first inequality is from \Cref{eq:os-ltf-tail-reg} and the equality follows by the definition of $\gamma$. We have:

\begin{align*}
n' \cdot (r/\gamma')^2 &= 1 - \frac{1}{(\gamma')^2}  \sum_{i \in T \cap S} \fh(i)^2 && \textrm{(By definition of $r$)}\\
&\leq \big(1 - \frac{1}{\gamma^2} \sum_{i \in T \cap S} \fh(i)^2\big) + \big(\big(\frac{\gamma'}{\gamma}\big)^2 - 1\big) && \textrm{(Multiplying by $(\gamma'/\gamma)^2$)} \\
&\leq n' \cdot O(\tau^2) + \big(\big(\frac{\gamma'}{\gamma}\big)^2 - 1\big) && \textrm{(By \Cref{eq:small-avg-tail-chow})} \\
&\leq n' \cdot O(\tau^2) + O(\tau/\gamma + (\tau/\gamma)^2) && \textrm{(Since $\gamma' \leq \gamma + \tau$)} \\
&\leq n' \cdot O(\tau^2) + O(\tau^{71/72}). && \textrm{(Since $\gamma \geq \Omega(\tau^{1/72})$)}
\end{align*}
Dividing both sides by $n'$ and taking square roots, we get that $r/\gamma' \leq O(\tau^{71/144})$, proving the claim.
\end{proof}

\subsection{Proof of \Cref{thm:main-structural}}
We now prove the main structural theorem.

\begin{proof}[Proof of \Cref{thm:main-structural}]
Fix $\tau = \eps^{1000}$.
By \Cref{thm:OS-LTF-approx}, there exists a function $f'$ that satisfies $d(f, f') \leq O(\eps^2)$ and is either a junta over $\vec{x}_H$ or is of the form in \Cref{eq:os-ltf1}.
By \Cref{prop:chow-ub-by-closeness}, $\dPC{S}(f, f') \leq \dChow(f, f') \leq 2 \sqrt{d(f, f')} = O(\eps)$. 
If $f'$ is a junta on $H$ or $T \subseteq S$ then the set of functions $g$ defined in \Cref{thm:main-structural} will be the same as the set of functions $f'$ defined in \Cref{thm:OS-LTF-approx}, and therefore we get that $\dPC{S}(f, g) \leq O(\eps)$ for some function $g$ defined in \Cref{thm:main-structural} as well.

It therefore remains to show that $\dPC{S}(f, g) \leq O(\eps)$ for some function $g$ defined in \Cref{thm:main-structural} in the case where $f$ is $\Omega(\eps^2)$-far from any junta on $H$, and where $T \nsubseteq S$.
Let $f'$ again be a function of the form in \Cref{eq:os-ltf1} that satisfies $\dChow(f, f') \leq O(\eps)$, and let $g$ be the unique function of the form in \Cref{eq:main-ltf} with the same threshold ${\theta'}$ and same head weights $v_H$ as $f'$.
By the definition of partial Chow distance and the triangle inequality,
\begin{align}
\dPC{S}(f, g) &\leq \dPC{H \cap S}(f, g) + \dPC{T \cap S}(f, g) \nonumber \\
              &\leq \dPC{H \cap S}(f, f') + \dPC{H \cap S}(f', g) + \dPC{T \cap S}(f, g) \nonumber \\
	      &\leq \dPC{H}(f, f') + \dPC{H}(f', g) + \dPC{T \cap S}(f, g). \label{eq:partial-chow-f-g}
\end{align}
We will upper bound each of the terms on the right-hand side of \Cref{eq:partial-chow-f-g}. 
For the first term, using \Cref{thm:OS-LTF-approx} and \Cref{prop:chow-ub-by-closeness}, we have that
\begin{equation}
\dPC{H}(f, f') \leq \dChow(f, f') \leq O(\eps) \ .
\label{eq:partial-chow-H-f-fp}
\end{equation}

For the second term, let $\vec{v}_T$ and $\vec{v}_T'$ denote the tail weight vectors of $f'$ and $g$, respectively.
Using the $O(\tau^{143/72})$-regularity of $\vec{v}_T$ (\Cref{clm:reg-tail-of-fp}), the $O(\tau^{71/144})$-regularity of $\vec{v}_T'$ (\Cref{clm:reg-tail-of-g}), and the fact that $\norm{\vec{v}_T} = \norm{\vec{v}_T'} = 1$, we get by \Cref{cor:half-biased-head-coeffs} (recalling the setting of $\tau$ in terms of $\eps$ given in the statement of \Cref{thm:main-structural})
\begin{equation}
\dPC{H}(f', g) \leq O(\tau^{71/144}) \leq O(\eps) \ .
\label{eq:partial-chow-H-fp-g}
\end{equation}

Finally we upper bound $\dPC{T \cap S}(f, g)$. Let $\alpha = \alpha(-v_0, v_T)$ be the constant of proportionality defined in \Cref{prop:weights-prop-chows}. Then
\begin{align}
\begin{split}
\dPC{T \cap S}(f, g) &\leq \sqrt{\sum_{i \in T \cap S} \big(\gh(i) - \frac{\alpha}{\gamma'} \fh(i) \big)^2} + \Big|1 - \frac{\alpha}{\gamma'}\Big| \sqrt{\sum_{i \in T \cap S} \fh(i)^2} \\
                     &= \sqrt{\sum_{i \in T \cap S} \big(\gh(i) - \alpha \cdot v_i' \big)^2} + \Big|\frac{\gamma' - \alpha}{\gamma'}\Big| \sqrt{\sum_{i \in T \cap S} \fh(i)^2} \\
&\leq O(\tau^{71/1728}) + \Big|\frac{\gamma' - \alpha}{\gamma'}\Big| \sqrt{\sum_{i \in T \cap S} \fh(i)^2} \\
&\leq O(\tau^{71/1728}) + O\Big(\frac{\tau^{1/12}}{\gamma'}\Big) \sqrt{\sum_{i \in T \cap S} \fh(i)^2} \\
&\leq O(\tau^{71/1728}) + O(\tau^{5/72}) \\
&\leq O(\eps) \ .
\end{split}
\label{eq:partial-chow-TS-f-g}
\end{align}
The first inequality is the triangle inequality; the equality follows by the definition of the weights in $\vec{v}_T'$ for $i \in T \cap S$ as $v_i' = \fh(i)/\gamma'$; the second inequality uses the $O(\tau^{71/144})$-regularity of $\vec{v}_T'$ to apply \Cref{prop:weights-prop-chows}; the third inequality holds by the triangle inequality since $\gamma' \stackrel{\tau}{\approx} \gamma$ and $\gamma \stackrel{\tau^{{1/4}}}{\approx} \alpha$ (the former approximation holds by assumption and the latter by \Cref{cor:chow-tail-weight}); the fourth inequality holds since $\gamma' \geq \gamma \geq \tau^{1/72}$ by \Cref{fct:tail-chow-param-lower-bound} and $\sum_{i \in T \cap S} \fh(i)^2 \leq \gamma \leq 1$ (by Parseval's Theorem), and the last inequality uses the setting of $\tau$ as a function of $\eps$ given in the statement of \Cref{thm:main-structural}.

The theorem follows by upper bounding the terms in the right-hand side of \Cref{eq:partial-chow-f-g} using \Cref{eq:partial-chow-H-f-fp},~\Cref{eq:partial-chow-H-fp-g},and~\Cref{eq:partial-chow-TS-f-g}.
\end{proof}

\subsection{Main algorithm for the Partial Chow Parameters Problem}
\label{subsec:main-alg-chow}

We next present the main algorithm of this section, which is an EPRAS for solving the Partial Inverse Chow Parameters Problem, and which works by leveraging this section's main structural result, \Cref{thm:main-structural}.

\begin{theorem}
There exists an algorithm for the Partial Inverse Chow Parameters Problem with the following guarantees. It takes as input four things: (1) a set $\set{(i, \fh(i)) : i \in S}$ for some LTF $f : \pmo^n \to \pmo$ and some $S \subseteq \set{0, 1, \ldots, n}$, (2) the length $n$ of the input to $f$, (3) an error parameter $\eps \in (0, \half)$, and (4) a confidence parameter $\delta > 0$. It outputs a weights-based representation of an LTF $g : \pmo^n \to \pmo$ such that $\dPC{S}(f, g) \leq O(\eps)$ with probability $1 - \delta$, and runs in time $n^2 \log n \cdot \log(1/\delta) \cdot 2^{\poly(1/\eps)}$.
\label{thm:main-algorithmic}
\end{theorem}

The algorithm consists of three steps: a parameter guessing step, an LTF enumeration step, and an LTF verification step. The second two steps are similar to those in the main algorithm in~\cite{os2011}. The two cases in the enumeration step correspond to the two cases in \Cref{thm:main-structural}.
The algorithm is as follows.\bigskip\ \\

{\centering 
\fbox{
\begin{minipage}{16cm}
\begin{enumerate}
\item Guess the size of the head $\card{H}$ and a value $\gamma'$ satisfying the conditions in \Cref{thm:main-structural}.
\begin{enumerate}
\item Compute $H \cap S := \set{i \in S : |\fh(i)| \geq \tau^2}$, $T \cap S := \set{i \in S : |\fh(i)| < \tau^2}$, $\card{T} := n - \card{H}$, $r$ from these guesses.
\item Set $H$ equal to the union of $H \cap S$ and $\card{H} - \card{H \cap S}$ arbitrary indices not in $S$.
\end{enumerate}
\label{en:main-guess-params}
\item For each guess of $H$ and $\gamma'$ in Step~\ref{en:main-guess-params}, enumerate candidate LTFs using the two cases in \Cref{thm:main-structural}:
\begin{enumerate}
\item Enumerate all junta LTFs $g$ over $\vec{x}_H$. \label{en:main-alg-enumerate-candidate-ltfs-juntas}
\item Enumerate all LTFs $g$ of the form given in \Cref{eq:main-ltf}. \label{en:main-alg-enumerate-candidate-ltfs-small-crit-ind}
\end{enumerate} 
\label{en:main-alg-enumerate-candidate-ltfs}
\item For each candidate LTF $g$ generated in Step~\ref{en:main-alg-enumerate-candidate-ltfs}, compute an empirical estimate $\bar{g}(i)$ of each of the Chow Parameters $\gh(i)$ for $i \in S$ so that $|\gh(i) - \bar{g}(i)| \leq \eps/\sqrt{|S|}$ with confidence $1 - \delta/(\card{S} \cdot M)$, where $M$ is the total number of LTFs enumerated in Step~\ref{en:main-alg-enumerate-candidate-ltfs}.
Output (a weights-based representation of) the first $g$ such that\\ $\norm{(\fh(i))_{i \in S} - (\bar{g}(i))_{i \in S}} \leq O(\eps)$.
\label{en:main-alg-check-candidate-ltfs}
\end{enumerate}
\end{minipage}
}}
\bigskip
\begin{proof}[Proof of \Cref{thm:main-algorithmic}]
We start by arguing that the above algorithm is correct.
By taking a union bound, it holds that all $|S| \cdot M$ estimates $\bar{g}(i)$ of the Chow Parameters $\gh(i)$ of candidate LTFs $g$ with $i \in S$ in Step~\ref{en:main-alg-check-candidate-ltfs} will be accurate to within a $\eps/\sqrt{|S|}$ additive error factor with probability at least $1 - \delta$.
In this case our estimates will all satisfy $\norm{(\gh(i))_{i \in S} - (\bar{g}(i))_{i \in S}} \leq \eps$, and hence by the triangle inequality $\norm{(\fh(i))_{i \in S} - (\bar{g}(i))_{i \in S}} - \eps \leq \dPC{S}(f, g) \leq \norm{(\fh(i))_{i \in S} - (\bar{g}(i))_{i \in S}} + \eps$ for every candidate LTF $g$. So, in this case, we will output a candidate LTF $g$ if and only if it satisfies $\dPC{S}(f, g) \leq O(\eps)$.
Furthermore, by \Cref{thm:main-structural}, one of the candidate LTFs $g$ enumerated in Step~\ref{en:main-alg-enumerate-candidate-ltfs} will satisfy $\dPC{S}(f, g) \leq O(\eps)$, and so with probability at least $1 - \delta$ we will output such a function.

We turn to analyzing the runtime of the algorithm.
We start by analyzing the number of guesses that we need for $\card{H}$ and $\gamma'$ in Step~\ref{en:main-guess-params}.
Each $\fh(i)$ with $i \in H$ satisfies $|\fh(i)| \geq \tau^2$, and because $\sum_{i \in [n]} \fh(i)^2 \leq 1$ this implies that $|H| \leq 1/\tau^4$.
Because $0 \leq \gamma \leq 1$, setting $\gamma'$ to be either $1$ or one of the $O(1/\tau)$ integer multiples of $\tau$ between $0$ and $1$ will satisfy the condition $\gamma \leq \gamma' \leq \gamma + \tau$. So, we need $O(1/\tau)$ guesses for $\gamma'$. Computing all other quantities given the guesses of $\card{H}$ and $\gamma'$ is efficient.

We next upper bound the number $M$ of functions enumerated in Step~\ref{en:main-alg-check-candidate-ltfs}.
By~\cite{MTT:61}, any junta LTF on $\card{H}$ variables can be represented using integer weights of magnitude at most $2^{O(\card{H} \log \card{H})}$, meaning that there are at most $2^{O(\card{H}^2 \log \card{H})} = 2^{O(1/\tau^8 \cdot \log(1/\tau))}$ such functions total (where we have used the fact that $|H| \leq 1/\tau^4$).

We next consider functions of the form specified in \Cref{eq:main-ltf}. For fixed $H$ and $\gamma'$ each such function is uniquely specified by a threshold ${\theta'}$ and head weights $v_H$, and so the total number of such functions is equal to the total number of possibilities for ${\theta'}, v_H$. Each of the $|H| + 1$ weights ${\theta'}$ and $v_i$ for $i \in H$ is an integer multiple of $\sqrt{\tau}/|H|$ and has magnitude at most $2^{O(|H| \log |H|)} \sqrt{\ln(1/\tau)}$. Therefore, the total number of such functions is upper bounded by $(2^{O(|H| \log^2 |H|)} \sqrt{\ln(1/\tau)/\tau})^{|H| + 1} \leq 2^{O(1/\tau^8 \cdot \log^2(1/\tau))}$, where we have again used the fact that $|H| \leq 1/\tau^4$. Combining the upper bounds on the number of juntas on $|H|$ variables and on the number of functions of the form in \Cref{eq:main-ltf} we get that $M \leq 2^{O(1/\tau^8 \cdot \log^2(1/\tau))}$.

Finally, we upper bound the amount of time necessary to obtain estimates $\bar{g}(i)$ of the Chow Parameters $\gh(i)$ with the desired error and confidence.
The following standard Chernoff bound holds for $\pm 1$-valued, i.i.d. {Bernoulli} random variables $\bX_1, \ldots, \bX_N$ each with mean $\mu$:
\begin{equation}
\Pr\Big[\Big|\Big(\frac{1}{N} \sum_{i=1}^N \bX_i \Big) - \mu \Big| \geq \Delta\Big] \leq 2 \exp(-\Delta^2 N/2) \ .
\label{eq:pmo-chernoff}
\end{equation}
Therefore, using $N := O\Big(\frac{|S|}{\eps^2} \log\Big(\frac{|S| \cdot M}{\delta}\Big)\Big)$ uniformly random samples $\rv{x}_1, \ldots, \rv{x}_N \sim \pmo^n$, the estimator $\bar{g}(i) := \sum_{j=1}^N g(\rv{x}_j) \cdot (\rv{x}_j)_i$ approximates $\mu := \gh(i)$ to within $\Delta := \eps/\sqrt{|S|}$ additive error with confidence $1 - \delta/(\card{S} \cdot M)$.

Computing each estimator $\bar{g}(i)$ requires $N$ evaluations of $g$ and uses $O(N)$ additional arithmetic operations. Each function evaluation uses $O(n)$ arithmetic operations, for a total of $O(n \cdot N)$ arithmetic operations. We must compute estimators $\bar{g}(i)$ for $\card{S} \cdot M$ many Chow Parameters $\gh(i)$, so in total Step~\ref{en:main-alg-check-candidate-ltfs} requires
\begin{equation}
O(n N \card{S} M) = O\Big(\frac{n \cdot |S|^2}{\eps^2} \log\Big(\frac{|S| \cdot M}{\delta}\Big)\Big) M \Big)
\label{eq:chow-step-3-time}
\end{equation}
time, which also subsumes the amount of time it takes to enumerate all $M$ functions in Step~\ref{en:main-alg-enumerate-candidate-ltfs}.

We conclude by upper bounding the overall runtime of the algorithm by the right-hand side of \Cref{eq:chow-step-3-time} times the number of guesses we need to make for $H$ and $\gamma'$ in Step~\ref{en:main-guess-params}:
\begin{align*}
&O\Big((1/\tau^5) \cdot \frac{n \cdot |S|^2}{\eps^2} \log\Big(\frac{|S| \cdot M}{\delta}\Big)\Big) M \Big)  \\
&= \frac{|S|^2}{\eps^2} \cdot \Big(\log \frac{|S|}{\delta} + (1/\tau^8 \cdot \log^2(1/\tau)) \Big)\Big) \cdot 2^{O(1/\tau^8 \cdot \log^7(1/\tau))} \\
&\leq n^2 \log n \cdot \log(1/\delta) \cdot 2^{\poly(1/\eps)} \ ,
\end{align*}
where we have used the fact that $\tau = \poly(\eps)$.
\end{proof}

\section{The Partial Shapley Indices Problem}
\label{sec:shapley-intro}

In this section we give a quasi-polynomial time algorithm for the Partial Shapley Indices problem by proving the following theorem, which is our second main result.  (See \Cref{sec:misc} for the definition of ``$\eta$-restricted,'' and recall from Table~1 that $\fc(i)$ is the $i$-th Shapley value of $f$.)

\begin{theorem}
\label{thm:main-shapley}
Let $f(x) = \sign(\ell(x))$ be an $\eta$-restricted LTF where $\ell(x) = \displaystyle\sum_{i=1}^n v_i x_i - t$ is a linear form with $v_1,\dots,v_n \geq 0$ and $\eta$ is an absolute constant in $(1/4,1]$.
There is an algorithm that, on input $\set{(i, \fc(i)) : i \in S}$ for some $S \subseteq [n]$ and a desired accuracy parameter $\eps$ satisfying $\eps \geq 1/n^{1/12}$, 
with high probability outputs a weights-based representation of an LTF $g$ that satisfies $\dPS{S}(f,g) \leq O(\eps)$ and runs in $2^{\tilde{O}(\log^{18} n/\eps^{24})}$ time
\end{theorem}

At the highest level the proof is by a case analysis. In \Cref{sec:useful-shapley} we first establish (\Cref{thm:discretization}) a preliminary structural result showing that the target LTF $f$ is closely approximated by an LTF $f'$ with ``well-structured'' (discretized) weights; having such weights is useful for our algorithm and analysis. We then proceed by a case analysis based on the $\tau^\ast$-critical index of the approximating LTF $f'$, where 
\begin{equation}
\label{eq:tau-ast}
\tau^\ast :=  \red{\frac{\eps^2}{(\log n)^4}}
\end{equation}

There are two cases: The first case, which corresponds to case 1 of \Cref{thm:main-structural-shapley}, is that the $\tau^\ast$-critical index of $f'$ is ``large,'' more precisely at least 
\begin{equation}
\label{eq:tau-ast}
k^\ast := \red{\max\left\{ \frac{4\log n}{\tau^2}, \frac{1}{\eps^{12}} \right\}}.
\end{equation}
The algorithm for this case is a relatively straightforward enumeration over candidate junta LTFs.
The second and more involved case, which corresponds to case 2 of \Cref{thm:main-structural-shapley},  is that the $\tau^\ast$-critical index of $f'$ is between 0 and $k^\ast$. This case has an analysis which incorporates ingredients from the aforementioned structural results of both \Cref{sec:preserving-head} and \Cref{sec:head-and-tails} (the latter of which in turn relies on 
technical results on approximating $\DShap$ by a mixture of $p$-biased product distributions which are given in \Cref{sec:approx-shapley}), and the algorithm for this case uses dynamic programming. We refer the reader to \Cref{sec:shapley-techniques} for further high-level description.

\subsection{Useful facts for Shapley indices} \label{sec:useful-shapley}
In this subsection we establish some tools which will be used for the proof of \Cref{thm:main-shapley}.

\subsubsection{Background results}
We recall several useful facts about LTFs and Shapley indices, starting with the definition of the Shapley indices:
\begin{equation}\label{eq:shapley-values}
\fc(i) := \Ex_{\bpi \sim \mathbb{S}_n}[f(x^+(\bpi, i)) - f(x(\bpi, i))].
\end{equation}

We begin with a useful and elementary observation which shows that larger weights in an LTF correspond to larger Shapley indices; the proof is given in \Cref{ap:rank}.
\begin{lemma}\label{lem:monotonicity-shapley}
Let $f(x) = \sign(\ell(x))$ be an LTF where $\ell(x) = \displaystyle\sum_{i=1}^n w_i x_i - \theta$ is a linear form with $w_1,\dots,w_n \geq 0$. Then for all $i \neq j \in [n]$, it holds that if $w_i \geq w_j$ then $\fc(i) \geq \fc(j).$
\end{lemma}

We continue by recalling a theorem about the anti-concentration of measure under the Shapley distribution $\DShap$ and some bounds for the Shapley distance between two function $f$ and $f'$.
Recall that an LTF $f: \pmo^n \to \pmo$, $f(x) = \sign(w\cdot x - \theta)$ with $w \in \R^n$ is said to be \emph{$\eta$-restricted} if $\theta \in [-(1 - \eta)\norm{w}_1, (1 - \eta) \norm{w}_1]$.

\begin{theorem}[{\cite[Theorem 15]{journals/geb/DeDS17}}]
Let $\ell(x) = \sum_{i=1}^n v_i x_i - \theta'$ be a monotone non-decreasing, $\eta$-restricted affine form where $\eta$ is an absolute constant in $(1/4,1]$, so $v_i \geq 0$ for $i \in [n]$ and $\abs{\theta'} \leq (1-\eta) \cdot \sum_{i=1}^n \abs{v_i}$. Let $12 \leq k \leq n$, and let $r \in \R^+$ be such that $|S| \geq k$ where $S = \set{i \in [n] : \abs{v_i} \geq r}$. Then
\[
\Prx_{\rv{x} \sim \DShap}[|\ell(\rv{x})| < r] = O\Big(\frac{1}{\log n} \cdot \frac{1}{k^{1/6}}\Big) \ .
\]
\label{thm:shapley-anti-concentration}
\end{theorem}

We require the following results from \cite{journals/geb/DeDS17} relating Shapley distance, Shapley Fourier distance, and the Shapley distribution:

\begin{lemma}[{\cite[Lemma 11]{journals/geb/DeDS17}}]
Let $f, g : \pmo^n \to [-1, 1]$ be LTFs. Then
\[
\dShapley(f, g) \leq \sqrt{2 H_{n-1}} \cdot \dFourier(f, g) + \frac{4}{\sqrt{n}} \ , 
\] 
where $H_k=\Theta(\log k)$ is the $k$-th harmonic number.
\label{lem:fourier-upper-bounding-shapley}
\end{lemma}

\begin{lemma}[{\cite[Special case of Fact 7]{journals/geb/DeDS17}}]
Let $f, g : \pmo^n \to \R$ be LTFs. Then
\[
\displaystyle \dFourier(f, g) \leq 2 \sqrt{\Prx_{\rv{x} \sim \DShap}[f(\rv{x}) \neq g(\rv{x})]} \ .
\]
\label{lem:ell1-upper-bounding-fourier}
\end{lemma}

By combining \Cref{lem:fourier-upper-bounding-shapley} and \Cref{lem:ell1-upper-bounding-fourier}, we get the following.

\begin{corollary}
Let $f, g : \pmo^n \to [-1, 1]$ be LTFs. Then
\[
\displaystyle \dShapley(f, g) \leq O\parens*{\sqrt{\log n \cdot \Prx_{\rv{x} \sim \DShap}[f(\rv{x}) \neq g(\rv{x})]} + \frac{1}{\sqrt{n}}} \ .
\]
\label{cor:cor-ub-shapley-by-l1}
\end{corollary}

\subsubsection{A discretization lemma}
As described earlier, we will perform our case analysis on an LTF $f'$ which approximates the target LTF $f$ (with respect to Shapley distance) and whose weights (after a suitable rescaling) are integers that are not too large.  The following theorem provides the necessary structural result ensuring the existence of such an approximation.

\begin{theorem}
\label{thm:discretization}
Let $\eps \in (\frac{1}{n^{1/12}}, \half)$ and let $f: \pmo^n \to \pmo$ be a monotone increasing, $\eta$-restricted LTF where $\eta$ is an absolute constant in $(1/4,1]$.
Then there exists an LTF $f'(x) = \sign(\sum_{i=1}^n w_i \cdot x_i -\theta)$
where $\theta, w_1, \ldots, w_n$ are integer multiples of $1/(n^2 \cdot k^{k/2})$ for $k=1/\eps^{12}$, with $\theta, w_1, \ldots, w_n \in [0, 1]$ and $\max_{i \in [n]} w_i \geq 1/2$,
such that $\dShapley(f,f') \leq O(\eps)$.
\end{theorem}

Before giving the proof, to motivate the first structural lemma we will use, consider a linear form $\ell(x) = \vec{v} \cdot \vec{x}-\theta'$ and its corresponding LTF $f(x)=\sign(\ell(x))$.  We note that given a probability distribution $\chi$ over $\bn$,
\begin{itemize}
\item If $\Prx_{\rv{x}\sim \chi }[|\ell(\rv{x})| \text{ is small}]$ is large, then for $\ell'(x) = v' \cdot x - \theta'$ a slight perturbation of $\ell(x)$, the corresponding LTF $f(x)=\sign(\ell'(x))$ could be far from $g$ with respect to $\chi$, i.e.~$\Prx_{\rv{x}\sim \chi }[f(\rv{x})\neq f'(\rv{x}) ]$ could possibly be large.
\item On the other hand, if $\Prx_{\rv{x}\sim \chi }[|\ell(\rv{x})| \text{ is small}]$ is small, then any slight perturbation $\ell'(x)$ will be such that for the corresponding perturbed LTF $f$, the probability
$\Prx_{\rv{x}\sim \chi }[f(\rv{x})\neq f'(\rv{x}) ]$ must be small.
\end{itemize}

The following lemma formalizes the above observations; for completeness we give its simple proof.
\begin{lemma}\label{lem:anti-closeness}
Let $\ell(x),\ell'(x)$ be two linear forms such that $|\ell(x)-\ell'(x)| \le \epsilon$ for all $x \in \bn$. Suppose that $\ell$ satisfies $\Prx_{\rv{x}\sim \DShap}[|\ell(\rv{x})|\le \epsilon ]\le \delta.$ Then it holds that $\Prx_{\rv{x} \sim \DShap}[f(\rv{x}) \neq f'(\rv{x})] \le \delta$, where $f(x)=\sign(\ell(x))$ and $f'(x)=\sign(\ell'(x)).$
\end{lemma}
\begin{proof}
\begin{align*}
\Prx_{\rv{x} \sim \DShap}[f(\rv{x}) \neq f'(\rv{x})] &=  \Prx_{\rv{x} \sim \DShap}[\sign(\ell(\rv{x})) \neq \sign(\ell'(\rv{x}))] \\
& \leq \Prx_{\rv{x} \sim \DShap}[|\ell(\rv{x})| < |\ell(\rv{x}) - \ell'(\rv{x})|] \\
&\leq \Prx_{\rv{x} \sim \DShap}[|\ell(\rv{x})| < \epsilon] \leq \delta \ . \qedhere
\end{align*}
\end{proof}

\Cref{thm:shapley-anti-concentration} provides the desired upper bound on the probability that $\bx \sim \DShap$ has $|\ell(\bx)|$ being ``too small,'' but to apply it we need to ensure that ``many'' weights $w_i$ are ``not too small.'' This is ensured by the following lemma:

\begin{theorem}[{\cite[Theorem 3]{journals/geb/DeDS17}}]\label{thm:LParg} 
Let $g: \pmo^n \to \pmo$ be an $\eta$-restricted LTF where $\eta=\Theta(1)$, and let $k \in [2, n]$.
There exists a representation of $g$ as $g(x) = \sign(\sum_{i=1}^n v_i x_i-\theta)$ such that (after reordering coordinates so that condition (i) below  holds) we have:   (i) $|v_{i}| \geq |v_{i+1}|$, $i \in [n-1]$; (ii) $|\theta| \leq (1-\eta) \sum_{i=1}^n |v_i|$; and (iii) for all $i \in [0, k-1]$ we have $|v_{i}| \leq (2/\eta) \cdot \sqrt{n} \cdot k^{\frac{k}{2}} \cdot \sigma_k$, where $\sigma_k := \sqrt{\sum_{j \geq k} v_j^2}$.
\end{theorem}

Rescaling the weights so that the largest weight has magnitude 1, \Cref{thm:LParg} easily yields the following corollary:
\begin{corollary}\label{coro:Multiples}
Let $g: \pmo^n \to \pmo$ be an $\eta$-restricted LTF where $\eta=\Theta(1)$, and let $12 \leq k \leq n.$  Then $g$ has a representation as $g(x) = sign(v \cdot x -\theta)$ 
where the largest-magnitude weight has magnitude 1, the $\Omega(k)$ many largest-magnitude weights each have magnitude at least $r:= \frac{1}{n \cdot k^{k/2} }$, and $|\theta| \leq (1-\eta) \sum_{i=1}^n |v_i|$.
\end{corollary}

Now we can give the proof of \Cref{thm:discretization}.

\begin{proof}[Proof of \Cref{thm:discretization}]
Recall that $f(\vec{x}) = \sign(v \cdot x - \theta')$. Let  $k=1/\eps^{12}$ (as in the statement of \Cref{thm:discretization}).
Applying \Cref{coro:Multiples}, we may express $f(x)$ as $\sign(\ell(x))$, where $\ell(x)=v \cdot x - \theta$ and the weight vector $v$ satisfies the properties stated in that corollary. Since $\ell(x)$ is guaranteed to have ``many'' weights that are ``not too small'' we may apply \Cref{thm:shapley-anti-concentration} to it, and we get that
\[
\Prx_{\rv{x} \sim \DShap}[|\ell(\rv{x})| < r] = O\Big(\frac{1}{\log n} \cdot \frac{1}{k^{{1/6}}}\Big),
\]
where $ r= \frac{1}{n \cdot k^{k/2}}$. 

Now for each $i \in [n]$ we define a rounded version $w_i$ of the weight $v_i$ which is obtained by rounding it to the closest integer multiple of $\frac{1}{n^2 k^{k/2}}$, and we let $\ell'(x)$ be the linear form $w \cdot x - \theta.$ It is immediate that for all $x \in \bn$ we have $|\ell(x)-\ell'(x)|\le \frac{1}{n{k^{k/2}}}=r$, and that $\max_{i \in [n]} w_i \geq 1/2$.  Letting $f'(x)=\sign(\ell'(x))$, by \Cref{lem:anti-closeness} we have that
\[
\Prx_{\rv{x} \sim \DShap}[f(\rv{x}) \neq f'(\rv{x})] \le O\Big(\frac{1}{\log n} \cdot \frac{1}{k^{{1/6}}}\Big).
\]
Finally, applying \Cref{cor:cor-ub-shapley-by-l1}, we get that 
\[
\displaystyle \dShapley(f, f') \leq O\Big(\sqrt{\log n \cdot \Prx_{\rv{x} \sim \DShap}[f(\rv{x}) \neq f'(\rv{x})]} + \frac{1}{\sqrt{n}}\Big) \le O(\epsilon)
\]
as was to be shown.
\end{proof}

\subsubsection{Approximating $\DShap$ by a mixture of $p$-biased product distributions}
\label{sec:approx-shapley}
We will use the following lemma from \cite{journals/geb/DeDS17} in order to express the Shapley indices in terms of the coordinate correlation coefficients:
\begin{lemma} [{\cite[Lemma 11]{journals/geb/DeDS17}}] \label{lem:expshap}
For $f : \{-1,1\}^n \rightarrow \{-1,1\}$ any monotone function, for each 
$i=1,\dots,n$ we have
$$
\fc(i) = \frac{f(1^n) - 
f((-1)^n)}{n} + \frac{\Lambda(n)}{2} \cdot \left( f^*(i)  - {\frac 1 n}
\displaystyle\sum_{j=1}^n f^*(j)\right),
$$
where $$f^*(i)=\Ex_{\rv{x}\sim \DShap}[f(\rv{x})\bx_i].$$
\end{lemma}

In \Cref{section:head-and-tails} we proved two structural results, \Cref{prop:weights-prop-chows}  (``tail weights are proportional to tail Chow parameters'') and \Cref{thm:p-biased-head-coeffs} (``exchanging a regular tail vector for another regular tail vector with the same $\ell_1$ and $\ell_2$ norm doesn't change  the head Chow parameters by much''), for general $p$-biased input distributions. To analyze our algorithm for the Partial Chow Parameters Problem we only needed the $p=1/2$ case of these results, but now we will use those structural results in their full generality.

\Cref{lem:expshap} shows that the Shapley indices are closely related to the distribution $\DShap$. 
Towards the goal of employing the results of \Cref{section:head-and-tails} for the Shapley problem, ideally we would like to define the Shapley distribution $\DShap$ as a mixture of $p$-biased product distributions $u^n_p$. (As a sanity check of the feasibility of doing this, we note that $\DShap$ and each $u^n_p$ are all  \emph{exchangeable} distributions:  for any one of these distributions, the probability weight assigned to an $n$-bit string depends only on the number of 1s in the string.)
Recall that the Shapley distribution $\DShap$ is defined as follows:  it puts zero weight on the strings $1^n$ and $(-1)^n$, and for every other $x \in \bn$, it assigns weight
\[
\Prx_{\rv{x}\sim \DShap }[\rv{x}=x]=
\dfrac{\frac{1}{i}+\frac{1}{n-i}}{\Lambda(n)\binom{n}{i}}\ ,
\quad \quad \text{ where } i=\weight(x) := |\{k\in[n]: x_k=1\}|.
\]
How can we draw $\bx \sim \DShap$ via a random procedure that uses the $p$-biased product distributions $u_p^n$?  Towards answering this question, we observe that for $i \in \{1,\cdots,n-1\}$, routine calculus yields that
\[
\int_{0}^{1}p^{i}(1-p)^{n-i}\frac{(\frac{1}{p}+\frac{1}{1-p})}{\Lambda(n)}dp=
\frac{\frac{1}{i}+\frac{1}{n-i}}{\Lambda(n)\binom{n}{i}} \ .
\]
Therefore $\DShap$ can be alternatively defined as follows:
\[
\Prx_{\rv{x}\sim \DShap }[\rv{x}=x]=\indic{i\not\in \{0,n\}} \cdot \displaystyle\int_{0}^{1}\frac{(\frac{1}{p}+\frac{1}{1-p})}{\Lambda(n)} p^{i}(1-p)^{n-i} dp,\text{ where $ i=\weight(x)$.}
\]
This leads to the following natural first attempt to define a new sampling mechanism for making a draw from the Shapley distribution $\DShap$ (where we write $\{-1,1\}_{=k}^{n}$ to denote $\{ x \in \pmo^n: \weight ( x )=k\}$):

\medskip

\noindent \fbox{
\begin{minipage}{0.45\textwidth}
\textsc{Original Shapley distribution \\ sampling mechanism:}
\begin{itemize}
    \item Sample layer $\bk\in \{1,\cdots,n-1\}$ with probability proportional to $\frac{1}{k}+\frac{1}{n-k}$.
    \item Then sample a uniformly random point from
$\pmo_{=\bk}^{n}$.
$$$$
\end{itemize}
\end{minipage}
}$\quad \Rightarrow$
	\fbox{
\begin{minipage}{0.45\textwidth}
\textsc{First attempt at new Shapley \\ distribution sampling mechanism:}
\begin{enumerate}
    \item Sample $\bp\in(0,1)$ with probability $\mathcal{K}(p)$ proportional to
    $\frac{(\frac{1}{p}+\frac{1}{1-p})}{\Lambda(n)}.$
    \item Then sample layer $\bk\in \{1,\cdots,n-1\}$ with probability proportional to $\mathcal{K}(\bp)\binom{n}{k}\bp^k(1-\bp)^{n-k}=\frac{1}{k}+\frac{1}{n-k}$.
    \item Finally sample a uniformly random point from $\{-1,1\}_{=\bk}^{n}$.
\end{enumerate}
\end{minipage}
}
\medskip
\noindent

Unfortunately, there is a crucial flaw in the above new hoped-for sampling mechanism. The flaw is in Step (1):  a trivial verification shows that  $\int_{0}^{1}\frac{(\frac{1}{p}+\frac{1}{1-p})}{\Lambda(n)}dp=\infty$, and so it is not possible to actually sample $\bp$ as described in that step.

We get around this challenge by restricting the sampling space in Step~(1) above to $[\delta,1-\delta]$ instead of $(0,1)$  where as we will see soon, we take $\delta$ to be a very small value (a value which is $1/\poly(n)$ and at most $o(1/n)$). 
Thus, it is natural for us to consider the continuous probability distribution $\mathcal{K}(\delta)$ supported on $[\delta,1-\delta]$, which is defined as follows: 
\begin{definition}[$\mathcal{K}(\delta)$-distribution]\label{def:Kappa-delta-distribution}
A random variable $\rv{p}$ is \emph{$\mathcal{K}(\delta)$-distributed} if its density is given by
$f_{\mathcal{K}(\delta)}(p)=C_\delta \frac{(\frac{1}{p}+\frac{1}{1-p})}{\Lambda(n)}$ for any $p\in[\delta,1-\delta]$, where $C_\delta:=\left(\int_{\delta}^{1-\delta}\frac{(\frac{1}{p}+\frac{1}{1-p})}{\Lambda(n)} {d}p\right)^{-1} = \frac{\Lambda(n)}{2\ln(\delta^{-1}-1)}$. Notice that if $\delta$ is inverse polynomial in $n$ then $C_\delta=\Theta(1)$.
\end{definition}

Following the above first attempt at a new Shapley Distribution sampling procedure, we can also define $\mathcal{Q}(\delta)$, a continuous mixture of $u_p^n$ distributions, defined as follows:

\begin{definition}[$\mathcal{Q}(\delta)$-distribution]\label{def:Kappa-delta-distribution}
The \emph{$\mathcal{Q}(\delta)$ distribution} is supported in $\pmo^n$ and is defined as $\Pr_{\rv{x}\sim \mathcal{Q}(\delta)}[\rv{x}=x]=\int_{\delta}^{1-\delta}f_{\mathcal{K}(\delta)}(p)\Pr_{\rv{x}\sim u_p^n}[\rv{x}=x]dp$.
\end{definition}

Using the above definitions, we establish a useful approximation for the expectation of Boolean functions over $\DShap$ in terms of the $\mathcal{Q}(\delta)$ distribution:

\begin{lemma}\label{lem:approximation-of-shapley1}
Let \deltaExpression and let $f:\pmo^n\to \R$ be such that $\|f\|_\infty \le O(1)$ and  $f(-1)^n=f(1^n)=0$.
Then we have that
\[  \Big|\Ex_{\rv{x}\sim \DShap}[f(\rv{x})]-\frac{1}{C_\delta}\Ex_{\rv{x}\sim \mathcal{Q}(\delta)}[f(\rv{x})]\Big| \le \dfrac{\bigO{n\delta}}{\Lambda(n)}.
\]
\end{lemma}

Combining \Cref{lem:expshap} and \Cref{lem:approximation-of-shapley1} and the definition of the $\mathcal{Q}(\delta)$ distribution (which implies that $\E_{\rv{x}\sim\mathcal{Q}(\delta)}[f(\rv{x})]=
\E_{\rv{p}\sim\mathcal{K}(\delta)}[\E_{\rv{x}\sim u_{\rv{p}}^n} [f(\rv{x})] ]
$), an immediate consequence is the following approximation of Shapley indices which will be the starting point for various structural lemmas in later sections:

\begin{lemma}\label{lem:approximation-of-shapley2}
Let $f$ be any nontrivial monotone LTF (so $f((-1)^n)=-1$ and $f(1^n) = 1$).  Then for \deltaExpression,  for each $i \in [n]$, the value $\fc(i)$ is additively $O(n \delta)$-close to the quantity $\upsilon_i$ defined below:\begin{align*}
\fc(i) &{\stackrel{ \bigO{n\delta} }{\approx}} \frac{2}{n} + \frac{\Lambda(n)}{2} \cdot 
\dfrac{1}{C_\delta}\cdot \Ex_{\rv{p}\sim \mathcal{K}(\delta)}\left[ f^*(i,\rv{p})  - {\frac 1 n} \sum_{j=1}^n f^*(j,\rv{p}) \right] =:\upsilon_i \\
& = \frac{2}{n} + \frac{\Lambda(n)}{2} \cdot 
\dfrac{1}{C_\delta}\cdot \Ex_{\rv{p}\sim \mathcal{K}(\delta)}\left[ \sigma_{\bp} \fh(i,\rv{p})  - {\frac 1 n} \sum_{j=1}^n \sigma_{\rv{p}}\fh(j,\rv{p}) \right].
\end{align*}
\end{lemma}

The proofs  of \Cref{lem:approximation-of-shapley1} and \Cref{lem:approximation-of-shapley2} are a sequence of routine calculations and are given in \Cref{ap:approx}.

\subsubsection{Estimating Shapley indices}

For completeness we close this subsection with a quick description of a simple sampling-based scheme to approximate the Shapley indices of a given monotone LTF.

\begin{proposition}\label{prop:shapley-estimation}
There is a procedure \textsc{EstimateShapley} with the following properties:
The procedure is given oracle access to a monotone LTF
$f : \{-1,1\}^n \rightarrow \bits$, a desired accuracy parameter $\gamma$, 
and 
a desired failure probability $\delta_{\fail}$.  The procedure makes 
$O(n \log (n/\delta_{\fail}) /\gamma^2)$ oracle calls to $f$ and runs in 
time $O(n^2  \log (n/\delta_{\fail}) /\gamma^2)$ (counting
each oracle call to $f$ as taking one time step).
With probability $1-\delta_{\fail}$ it outputs a list of numbers $\tilde{a}(1),
\ldots, \tilde{a}(n)$ such that $$\sum_{i\in [n]}\Big( \tilde{a}(i) - \fc(i) \Big)^2 \leq \gamma.$$
\end{proposition}

\begin{proof}
The procedure empirically estimates each $\fc(j)$, $j=1,\dots,n$,
to additive accuracy $\gamma/\sqrt{n}$ using the definition of Shapley indices, \Cref{eq:shapley-values}.
This is done by generating a uniform random $\bpi \sim \mathbb{S}_n$ and then,
for each $i=1,\dots,n,$
constructing the two inputs $x^+(\bpi,i)$ and $x(\bpi,i)$ and calling
the oracle for $f$ twice to compute $f(x^+(\pi,i))-f(x(\pi,i)).$
Since $|f(x^+(\bpi,i))-f(x(\bpi,i))| \leq 2$ always, a straightforward application of Hoeffding bounds gives that a
sample of $m=O(n \log(n/\delta_{\fail})/\gamma^2)$ permutations suffices
to estimate all the $\fc(i)$ values to additive
accuracy $\pm \gamma/\sqrt{n}$ with total failure probability
at most $\delta_{\fail}$.  If each estimate $\tilde{a}(i)$ is additively
accurate to within $\pm \gamma/\sqrt{n}$, then $\dShapley(a,f)
\leq \gamma$ as desired.
\end{proof}

\subsection{Structural results on heads and tails of LTFs (Shapley version)}
\label{sec:head-and-tails}
Working in the same fashion as in \Cref{section:head-and-tails}, let $f(x) = \sign(w \cdot  x - \theta)$ be an LTF, and to simplify presentation let us assume that its weights are sorted in magnitude from largest to smallest, i.e., $|w_1| \geq |w_2| \geq \cdots \geq |w_n|$. 
Let $\tau > 0$. Although $f$ need not be $\tau$-regular, we can always partition its weights $w$ into ``head weights'' $w_H$ and ``tail weights'' $w_T$ such that $w_T$ is $\tau$-regular and any longer suffix of $w$ is not $\tau$-regular. Let $H=\{1,2,\dots,\}$ be the set of indices of head weights, and let $T=[n]\setminus H = \{\dots,n-1,n\}$ be the set of indices of tail weights.

\subsubsection{Regular tail weights are approximately affinely related to tail Shapley indices}
\label{sec:aff-of-tail}

We first show that the vector of tail weights $w_T$ is approximately affinely related to the vector of tail Shapley indices $(\fc(i))_{i \in  T }$; more precisely, there exist real values $\ensuremath{\accentset{\diamond}{A}}=\ensuremath{\accentset{\diamond}{A}}(w_H,\|w_T\|_1,\|w_T\|_2,\theta,\delta), \ensuremath{\accentset{\diamond}{B}}=\ensuremath{\accentset{\diamond}{B}}(w_H,\|w_T\|_1,\|w_T\|_2,\theta,\delta)$ such that $\fc(i) \approx \ensuremath{\accentset{\diamond}{A}} \cdot w_i +\ensuremath{\accentset{\diamond}{B}}$ for all $i \in T$.  This characterization will be helpful for recovering the tail weights of an LTF from the Shapley indices, and as a corollary also gives an approximation of the sum of Shapley indices of $f$ on $T$.

\begin{theorem}\label{thm:shapley-affine-tail}
Let $f(\vec{x}) = f(\vec{x}_H, \vec{x}_T) = \sign(\vec{w}_H \cdot \vec{x}_H + \vec{w}_T \cdot \vec{x}_T - \theta)$  be an LTF satisfying $f(1^n)=1,f((-1)^n)=-1$ where $\vec{w}_T$ is $\tau$-regular and $w_1,\dots,w_n \geq 0.$
There exist \vec{real values} $\ensuremath{\accentset{\diamond}{A}}(w_H,\|w_T\|_1,\|w_T\|_2,\theta,\delta), \ensuremath{\accentset{\diamond}{B}}(w_H,\|w_T\|_1,\|w_T\|_2,\theta,\delta)$ such that \vec{for \deltaExpression, }
\[
\sum_{i\in T} \parens*{\fc(i) - ( \ensuremath{\accentset{\diamond}{A}}(w_H,\|w_T\|_1,\|w_T\|_2,\theta,\delta) w_i  + \ensuremath{\accentset{\diamond}{B}}(w_H,\|w_T\|_1,\|w_T\|_2,\theta,\delta)}^2 \le  \bigO{n^3\delta^2 + \ln^2(\delta^{-1}) \cdot \parens*{|H|/n + \sqrt{\tau}}}.
\]
\end{theorem}

Initially, we will prove a simplified version of our theorem asserting the extra assumption $\|w_T\|_2=1$.

\begin{lemma}
Let $f(\vec{x}) = f(\vec{x}_H, \vec{x}_T) = \sign(\vec{w}_H \cdot \vec{x}_H + \vec{w}_T \cdot \vec{x}_T - \theta)$  be an LTF satisfying $f(1^n)=1,f((-1)^n)=-1$ where $\vec{w}_T$ is $\tau$-regular and $\|w_T\|_2=1$, and $w_1,\dots,w_n \geq 0.$
There exist \vec{real values} $\ensuremath{\accentset{\diamond}{\Gamma}}(w_H,\|w_T\|_1,\theta,\delta), \ensuremath{\accentset{\diamond}{\Delta}}(w_H,\|w_T\|_1,\theta,\delta)$ such that \vec{for \deltaExpression,}
\[
\sum_{i\in T} \parens*{\fc(i) - ( \ensuremath{\accentset{\diamond}{\Gamma}}(w_H,\|w_T\|_1,\theta,\delta) w_i  + \ensuremath{\accentset{\diamond}{\Delta}}(w_H,\|w_T\|_1,\theta,\delta)}^2 \le  \bigO{n^3\delta^2 + \ln^2(\delta^{-1}) \cdot \parens*{|H|/n + \sqrt{\tau}}}.
\]
\end{lemma}

\begin{proof}

By \Cref{lem:approximation-of-shapley2}, we have that
\[
\fc(i) \stackrel{ \bigO{n\delta} }{\approx} \frac{2}{n} + \frac{\Lambda(n)}{2} \cdot 
\dfrac{1}{C_\delta}\cdot \Ex_{\rv{p}\sim \mathcal{K}(\delta)}\left[ \sigma_{\rv{p}} \parens*{\fh(i,\rv{p})  - {\frac 1 n} \sum_{j=1}^n \fh(j,\rv{p}) } \right]=\upsilon_i.
\]
Squaring and summing this difference over all $i \in T$, it follows that we have 
\begin{equation}
  \sum_{i\in T} \Big(\fc(i)- \upsilon_i
   \Big) ^2 \le  O(n^3\delta^2).
\label{eq:approx-in-measure-tail}
\end{equation}

We now define the quantities $\vartheta_i,\varpi_i,\varrho_i$: The quantity $\vartheta_i$ is the same as $\upsilon_i$, but with  the summation over all of $[n]$ inside the expectation operator being instead a sum over all of $T$.
The second is the approximation of $\vartheta_i$ that results from using the affine transformation of the weights of the linear form (recall \Cref{prop:weights-prop-chows}) instead of the actual $p$-biased Fourier coefficients $\fh(\cdot,p)$.  More precisely, 
\begin{align*}
\vartheta_i&:=\frac{2}{n} + \frac{\Lambda(n)}{2} \cdot \dfrac{1}{C_\delta}\cdot \Ex_{\rv{p}\sim \mathcal{K}(\delta)}\left[ \sigma_{\rv{p}}\left( \fh(i,\rv{p})  - {\frac 1 n} \sum_{j\in T} \fh(j,\rv{p})\right) \right], \\
\varpi_i&:=\frac{2}{n} + \frac{\Lambda(n)}{2} \cdot \dfrac{1}{C_\delta}\cdot \Ex_{\rv{p}\sim \mathcal{K}(\delta)}\left[  \sigma_{\rv{p}}\alpha(\theta, w_H,w_T,\bp) \left(w_i-\frac{\sum_{k\in T}w_k}{n}\right) \right].
\end{align*}

To bound the error incurred by using $\varpi_i$ instead of $\upsilon_i$, we observe that
\[
\sum_{i\in T}(\upsilon_i - \varpi_i  )^2 =\frac{\Lambda(n)^2}{4C_\delta^2} \cdot \sum_{i\in T} \left(\Ex_{\rv{p}\sim \mathcal{K}(\delta)}\left[  \sigma_{\rv{p}}\alpha(\theta, w_H,w_T,\rv{p}) \left(w_i-\frac{\sum_{k\in T}w_k}{n}\right) -  \sigma_{\rv{p}}\left(\fh(i,\rv{p}) - \frac{1}{n}\sum_{i\in [n]}\fh(i,\rv{p}) \right)  \right] \right)^2.
\]
Since \deltaExpression, it is easy to see that $\frac{\Lambda(n)^2}{4C_\delta^2}\le \bigO{\ln^2(\delta^{-1})}$. Using Jensen's inequality and linearity of expectation, we get that
\begin{align}
& \sum_{i\in T}(\upsilon_i - \varpi_i)^2 \nonumber\\
&\le 
\bigO{\ln^2(\delta^{-1})}
 \cdot \Ex_{\rv{p}\sim \mathcal{K}(\delta)}\left[ \sigma_{\rv{p}}^2 \sum_{i\in T} \left( \alpha(\theta, w_H,w_T,\rv{p}) \left(w_i-\frac{\sum_{k\in T}w_k}{n}\right) -  \left(\fh(i,\rv{p}) - \frac{1}{n}\sum_{i\in [n]}\fh(i,\rv{p}) \right)   \right)^2 \right]. \label{eq:biggie}
\end{align}

By applying \Cref{prop:weights-prop-chows} to $f$ we get that 
\[
\displaystyle\sum_{i \in T} (\fh(i\vec{,p}) - \vec{ \alpha(\theta, w_H,w_T,p) }  \cdot w_i)^2 \leq \vec{ \bigO{\sqrt{\tfrac{\tau}{\sigma_p^2}}}} \tag{a}
\]
holds for each $p \in (0,1).$
Next by applying for any $p$ the first claim of \Cref{fact:centering} to (a) with the rescaling factor ``$c=\frac{|T|}{n}$''
, we have 
\[
\displaystyle\sum_{i \in T} \parens*{ {\frac{1}{n}}\sum_{j \in T}\fh(\vec{j,p}) - \vec{ \alpha(\theta, w_H,w_T,p) }  \cdot {\frac {1}{n}}\sum_{j \in T}w_j }^2 \leq \vec{ \bigO{\sqrt{\tfrac{\tau}{\sigma_p^2}}}}.  \tag{b}
 \]
Since the sum of all squared $p$-biased Fourier coefficients is at most 1, by Cauchy-Schwarz we have that
\[
 \sum_{i\in T} \left(  \frac{1}{n}\sum_{j\in H}\fh(j,p)   \right)^2 \le \bigO{\dfrac{|H| \cdot |T|}{n^2}}\le \bigO{|H|/n} \tag{c}.
\]
Using \Cref{fact:triangular-in-square} to add the last two inequalities, we get:
\[
\displaystyle\sum_{i \in T} \parens*{ {\frac{1}{n}}\sum_{j \in [n]}\fh(\vec{j,p}) - \vec{ \alpha(\theta, w_H,w_T,p) }  \cdot {\frac {1}{n}}\sum_{j \in T}w_j }^2 \leq \vec{ \bigO{\sqrt{\tfrac{\tau}{\sigma_p^2}}}+ |H|/n} .  \tag{d}
 \]

Multiplying (a), (d) by $\sigma_p^2=\bigO{1}$ and using \Cref{fact:triangular-in-square} to combine them we get:
\[
\sum_{i\in T} \left(  \sigma_{p}\alpha(\theta, w_H,w_T,p) \left(w_i-\frac{\sum_{k\in T}w_k}{n}\right) -  \sigma_{p}\left(\fh(i,p) - \frac{1}{n}\sum_{i\in [n]}\fh(i,p) \right)   \right)^2 \leq \bigO{|H|/n + \sqrt{\tau}}.
\]
Observing that the RHS above has no dependence on $p$, we can plug this into \Cref{eq:biggie} and we get that
\begin{equation}
\sum_{i\in T}(\upsilon_i - \varpi_i)^2 \le \bigO{\ln^2(\delta^{-1})} \Ex_{\rv{p}\sim \mathcal{K}(\delta)}\left[ \bigO{|H|/n} + \bigO{\sqrt{\tau}} \right]  \le   \bigO{\ln^2(\delta^{-1}) \cdot \parens*{|H|/n + \sqrt{\tau}}}.
\label{eq:approx-in-weights-tail-one}
\end{equation}

Finally, combining the bounds from \Cref{eq:approx-in-weights-tail-one} and \Cref{eq:approx-in-measure-tail} using \Cref{fact:triangular-in-square}, we get that
\begin{align*}
\sum_{i\in T} \left(\fc(i) - \varpi_i \right)^2 \le  \bigO{n^3\delta^2 + \ln^2(\delta^{-1}) \cdot \parens*{|H|/n + \sqrt{\tau}}}.
\end{align*}

To finish the proof it remains only to verify that $\varpi_i$ can be written as $ \ensuremath{\accentset{\diamond}{\Gamma}} w_i+\ensuremath{\accentset{\diamond}{\Delta}}$,
where ${\ensuremath{\accentset{\diamond}{\Gamma}}=\ensuremath{\accentset{\diamond}{\Gamma}}(w_H,\|w_T\|_1,\theta,\delta)}$ and ${\ensuremath{\accentset{\diamond}{\Delta}}=\ensuremath{\accentset{\diamond}{\Delta}}(w_H,\|w_T\|_1,\theta,\delta)}$. This holds because
\begin{align}\varpi_i&= \frac{2}{n} + \frac{\Lambda(n)}{2} \cdot \int_{\delta}^{1-\delta }\dfrac{1/p + 1/(1-p)}{\Lambda(n)}\left[  \sigma_p\alpha(\theta, w_H,w_T,p) \left(w_i-\frac{\sum_{k\in[n]}w_k}{n}\right) \right] dp \nonumber \\
&= \underbrace{\parens*{  \frac{1}{2}  \int_{\delta}^{1-\delta } \sigma_p  \cdot \alpha(\theta, w_H,w_T,p)\cdot \dfrac{1/p + 1/(1-p)}{n} \, dp}}_{\ensuremath{\accentset{\diamond}{\Gamma}}} w_i \nonumber \\
& \ \ \ +
      \underbrace{\parens*{  \frac{2}{n} - \frac{1}{2}  \int_{\delta}^{1-\delta } \sigma_p \cdot \alpha(\theta, w_H,w_T,p) \cdot \dfrac{1/p + 1/(1-p)}{n} \cdot \|w\|_1 \, dp }}_{\ensuremath{\accentset{\diamond}{\Delta}}}. 
\label{eq:old-affine-constants}
\end{align}
It is important to mention that by \Cref{def:alpha}, it holds that \[
\alpha(\theta, w_H,w_T,p) := \E_{\brho \sim u^{|H|}_p} [\alpha(\psi^{[w_T]}_p(\theta - w_H \cdot \brho))]\]
where $\psi_p^{[w]}(x) =  \frac{x - \mu_p \cdot \normone{w}}{\sigma_p\normtwo{w}} \quad \text{ for }w\in\R_{\ge 0}^n$.
Therefore, the quantity $\alpha(\theta, w_H,w_T,p)$ depends actually only on $(\theta, w_H,\|w_T\|_1,p)$ since $\|w_T\|_2=1$.
\end{proof}

\begin{proof}[Proof of \Cref{thm:shapley-affine-tail}]
Finally it is easy to see that we can relax the assumption of $\|w_T\|_1=1$ by setting 
\begin{equation}
\begin{cases}
\ensuremath{\accentset{\diamond}{A}}(w_H,\|w_T\|_1,\|w_T\|_2,\theta,\delta):=
\frac{1}{\|w_T\|_2}\ensuremath{\accentset{\diamond}{\Gamma}}(\frac{w_H}{\|w_T\|_2},\frac{\|w_T\|_1}{\|w_T\|_2},\frac{\theta}{\|w_T\|_2},\delta)
\\
\ensuremath{\accentset{\diamond}{B}}(w_H,\|w_T\|_1,\|w_T\|_2,\theta,\delta):=
\ensuremath{\accentset{\diamond}{\Delta}}(\frac{w_H}{\|w_T\|_2},\frac{\|w_T\|_1}{\|w_T\|_2},\frac{\theta}{\|w_T\|_2},\delta)
\end{cases}
\label{eq:affine-constants}
\end{equation}
where $\ensuremath{\accentset{\diamond}{\Gamma}},\ensuremath{\accentset{\diamond}{\Delta}}$ are the affine constants of \Cref{eq:old-affine-constants}.
\end{proof}

In the special case in which the entire weight vector $w$ is regular, we get the following: 
\begin{corollary}\label{corollary:shapley-affine-regular}
Let $f(\vec{x}) = f(\vec{x}_H, \vec{x}_T) = \sign(\vec{w}_H \cdot \vec{x}_H + \vec{w}_T \cdot \vec{x}_T - \theta)$  be an LTF satisfying $f(1^n)=1,f((-1)^n)=-1$ where $\vec{w}$ is $\tau$-regular and $w_1,\dots,w_n \geq 0.$
There exist \vec{real values} $\ensuremath{\accentset{\diamond}{A}}(w_H,\|w_T\|_1,\|w_T\|_2,\theta,\delta), \ensuremath{\accentset{\diamond}{B}}(w_H,\|w_T\|_1,\|w_T\|_2,\theta,\delta)$ such that \vec{for \deltaExpression,}
\[
\sum_{i\in T} \parens*{\fc(i) - ( \ensuremath{\accentset{\diamond}{A}}(w_H,\|w_T\|_1,\|w_T\|_2,\theta,\delta) w_i  + \ensuremath{\accentset{\diamond}{B}}(w_H,\|w_T\|_1,\|w_T\|_2,\theta,\delta)}^2 \le  \bigO{n^3\delta^2 + \ln^2(\delta^{-1}) \cdot \parens*{|H|/n + \sqrt{\tau}}}.
\]
\end{corollary}

\subsubsection{Preserving the head Shapley indices}

The last structural result we require on Shapley indices is an analogue of  \Cref{thm:p-biased-head-coeffs} for the head Shapley indices. More precisely, the following theorem shows that exchanging the tail weights $w_T$ of an LTF $f$ with other weights $w_T'$ of
the same $\ell_1$ and $\ell_2$ norm does not change the head Shapley coefficients $(\fc(i))_{i\in H}$ by too much.
 
\begin{theorem}\label{thm:shapley-head-coeffs}
Let $f(\vec{x}) = f(\vec{x}_H, \vec{x}_T) = \sign(\vec{w}_H \cdot \vec{x}_H + \vec{w}_T \cdot \vec{x}_T - \theta)$ and let $g(\vec{x}) = f'(\vec{x}_H, \vec{x}_T) = \sign(\vec{w}_H \cdot \vec{x}_H + \vec{w}_T' \cdot \vec{x}_T - \theta)$ where $(\vec{w}_H, \vec{w}_T) \in (\R^{\geq 0})^n$, and where $\vec{w}_T, \vec{w}_T'$ are $\tau$-regular, satisfy $\norm{\vec{w}_T}_1 = \norm{\vec{w}'_T}_1$, and satisfy $\norm{\vec{w}_T}_2 = \norm{\vec{w}'_T}_2$.
Suppose that $f (1^n )= g (1^n ) = 1, f ((-1)^n )= g ((-1)^n )  =-1$. Then for \deltaExpression, we have that
\[
\sum_{i \in H} (\fc(i) - \gc(i))^2 \leq  \bigO{|H|n^2 \delta^2+ \Lambda^2(n) (\tau^2 + \frac{\sqrt{\tau} |H|}{n}) }  \ .
\]
\end{theorem}
\begin{proof}
From \Cref{lem:approximation-of-shapley2}, we have that each $i \in [n]$ satisfies
\[
\fc(i) \stackrel{ \bigO{n\delta} }{\approx} \frac{2}{n} + \frac{\Lambda(n)}{2} \cdot 
\dfrac{1}{C_\delta}\cdot \Ex_{\rv{p}\sim \mathcal{K}(\delta)}\left[ \sigma_{\rv{p}} ( \fh(i,\rv{p})  - {\frac 1 n} \sum_{j=1}^n \fh(j,\rv{p}) ) \right]=\upsilon_i(f).
\]
Applying the above equation with \Cref{fact:triangular-in-square}, we have
 \begin{align}\label{eq:Shap-head-diff-1}
& \sum_{i \in H} (\fc(i) - \gc(i))^2 %
\leq  O(|H|n^2 \delta^2) + \sum_{i \in H} (\upsilon_i(f) - \upsilon_i(g))^2
\end{align}
 
Using the above equation with Jensen's inequality, we have  
\begin{align}\label{eq:Shap-head-diff}
& \sum_{i \in H} (\upsilon_i(f) - \upsilon_i(g))^2\leq \parens*{\frac{\Lambda(n)}{2C_\delta}}^2\Ex_{\rv{p}\sim \mathcal{K}(\delta)}\left[ \sum_{i \in H}\bigg( \sigma_p \bigg( \fh(i,\rv{p})  - \gh(i,\rv{p}) - {\frac 1 n} \sum_{j=1}^n \parens*{ \fh(j,\rv{p}) - \gh(j,\rv{p})} \bigg)\bigg)^2 \right].
\end{align}

 To bound the right hand side, we will leverage two facts.  First, by applying \Cref{prop:weights-prop-chows} to $f$ and \Cref{fact:centering}, for any $p$, we get that 
 \[
 \bigg| \dfrac{1}{n} \sum_{i \in T}  \hat{f}(i,p) - \vec{ \alpha(\theta, w_H,w_T,p) }  \cdot \dfrac{1}{n } \sum_{i \in T} w_i  \bigg| \le O\parens*{{\frac 1 {\sqrt{n}}} \cdot \parens*{{\frac {\tau}{\sigma_p^2}}}^{1/4}}. 
 \]
 Similarly, applying \Cref{prop:weights-prop-chows} to $g$ and \Cref{fact:centering}, for any $p$, we get that
 \[
 \bigg| \dfrac{1}{n} \sum_{i \in T}  \hat{g}(i,p) - \vec{ \alpha(\theta, w_H,w'_T,p) }  \cdot \dfrac{1}{n } \sum_{i \in T} w'_i  \bigg| \le O\parens*{{\frac 1 {\sqrt{n}}} \cdot \parens*{{\frac {\tau}{\sigma_p^2}}}^{1/4}}.
 \]
Recalling that the dependence of $\vec{ \alpha(\theta, w_H,w_T,p) } $ on $w_T$ is only through the quantities $\Vert w_T \Vert_1$ and $\Vert w_T \Vert_2$, and recalling that $\Vert w_T \Vert_1 = \Vert w'_T \Vert_1$ and $\Vert w_T \Vert_2 = \Vert w'_T \Vert_2$, we can combine the last two inequalities to obtain
 \begin{equation}~\label{eq:Shapley-small-crit-head-1}
 \bigg| \dfrac{1}{n} \sum_{i \in T}  \hat{g}(i,p) - \dfrac{1}{n} \sum_{i \in T}  \hat{f}(i,p)  \bigg| \le 
O\parens*{{\frac 1 {\sqrt{n}}} \cdot \parens*{{\frac {\tau}{\sigma_p^2}}}^{1/4}}
 \end{equation}   
 Next, recalling \Cref{thm:p-biased-head-coeffs}, we have that
\[
\sum_{i \in H} (\hat{f}(i,p) - \hat{g}(i,p))^2 \leq O\left(\frac{\tau^2}{\sigma_p^2}\right). \ 
\]
Applying the third statement of \Cref{fact:centering} with its scaling factor ``$c$'' set to be $\frac{|H|}{n}$, we get that
\begin{equation}\label{eq:head-bound}
\sum_{i \in H} \parens*{ \hat{f}(i,p) -  \hat{g}(i,p) -
\dfrac{1}{n} \sum_{j \in H}  \hat{f}(j,p) -  \hat{g}(j,p) }^2 \leq O\left(\frac{\tau^2}{\sigma_p^2} \right).
\end{equation}
By combining  \Cref{eq:Shapley-small-crit-head-1} and \Cref{eq:head-bound} , we get that for all $p \in (0,1)$, 
\[
\sum_{i \in H}\bigg( \sigma_p \bigg( \fh(i,{p})  - \gh(i,{p}) - {\frac 1 n} \sum_{j=1}^n \fh(j,{p}) - \gh(j,{p}) \bigg)\bigg)^2 = \bigO{\tau^2 + {\frac{\sqrt{\tau} |H|}{n}}} .
\]
Plugging this back into \Cref{eq:Shap-head-diff-1} and \Cref{eq:Shap-head-diff}, we get 
\[
\sum_{i \in H} (\fc(i) - \gc(i))^2 \leq  \bigO{|H|n^2 \delta^2} + \parens*{ \frac{\Lambda(n)}{2C_\delta}}^2 \cdot \bigO{\tau^2 + \frac{\sqrt{\tau} |H|}{n}}. 
\]
\end{proof}

\subsection{Structural Theorem for LTFs under $\dPS{S}$}
In this section we establish a structural result which is at the heart of our algorithm for the Partial Shapley Indices Problem.  This result may be viewed as an analogue of \Cref{thm:main-structural}, the structural result that was the core of our algorithm for the Partial Chow Parameters problem.

We will use the following lemma, which appears in a number of previous works (e.g.,~\cite[Fact 25]{de2014}). Given a vector $w$ of non-negative weights that are sorted by magnitude, so $|w_1| \geq \cdots \geq |w_n| \geq 0$, for $i \in [n]$ let us write $\tail_i(w)$ to denote $(\sum_{j = i}^n w_j^2)^{1/2}$. The lemma says that, for weight vectors $w$ with sorted weights as above, this quantity decreases geometrically for $i$ less than the critical index:
\begin{lemma}
Let $w = (w_1, \ldots, w_n) \in \R^n$ be such that $\abs{w_1} \geq \cdots \geq \abs{w_n}$, and let $1 \leq a \leq b \leq c(w, \tau)$, where $c(w, \tau)$ is the $\tau$-critical index of $w$. Then $\tail_b(w) < (1 - \tau^2)^{(b - a)/2} \cdot \tail_a(w)$.
\label{lem:geom-decreasing-tail}
\end{lemma}

\begin{proof}
By definition of the critical index, $\abs{w_i} > \tau \cdot \tail_i(w)$ for $i < c(w, \tau)$. Therefore for such an $i$, $\tail_i(w)^2 = w_i^2 + \tail_{i+1}(w)^2 > \tau^2 \cdot \tail_i(w)^2 + \tail_{i+1}(w)^2$, and so $\tail_{i+1} (w)< (1 - \tau^2)^{1/2} \cdot \tail_i(w)$. The result follows by applying this last inequality repeatedly.
\end{proof}

Let us sketch the high level idea for the proof of the large critical index case in the main structural theorem below. (The small critical index case will be an immediate corollary of the previous sections' results.) First, we argue that if the $\tau^\ast$-critical index $k_{\text{critical}}(\tau^\ast)$
is sufficiently large, specifically $k_{\text{critical}} > k^\ast$, then the single weight $w_{k^\ast/2}$ will have larger magnitude than the $\ell_1$ weight of the entire tail $\sum_{i=k^\ast}^n \abs{w_i}$. Then we apply the anti-concentration results \Cref{thm:shapley-anti-concentration} and \Cref{lem:anti-closeness} 
to conclude that in this case the tail weights will rarely affect the sign of the affine form $\ell(x) = \sum_{i=1}^n w_i x_i-\theta$ and hence that $f(x) = \sign(\ell(x))$ is close in $\ell_1$ distance to the junta $f'(x) = \sign(\ell'(x))$ where $\ell'(x) = \sum_{i=1}^{k^\ast-1} w_i x_i -\theta$. Finally, we show using \Cref{cor:cor-ub-shapley-by-l1} that the closeness of two functions in $\ell_1$ distance implies closeness in (partial) Shapley distance.

\begin{theorem} \label{thm:main-structural-shapley}
Let $\eps \in (\frac{1}{n^{1/14}}, \half)$.  Define $\tau^\ast := (\tfrac{\eps^2}{\log^4 n})$ and $k^\ast := \red{\max\{4\tfrac{\log n}{(\tau^\ast)^2},\tfrac{1}{\eps^{12}}\}}$.
 Let $f(x) = \sign(w \cdot x - \theta)$ be a monotone increasing, $\eta$-restricted LTF where $\eta$ is an absolute constant in $(1/4,1]$.
There is a value $0 \leq \hat{k} \leq k^\ast$ such that, taking $H \subset [n]$  to be the indices of the $\hat{k}$ largest-magnitude weights in $f$  and $T := [n] \setminus H$ to be the complementary $n-\hat{k}$ remaining weights, at least one of the following holds: 
\begin{enumerate}
\item $\dShapley(f, f')= O(\eps)$ for some LTF junta $f'$ over the variables in $H$; or \label{en:main-large-crit-shapley}
\item %
The tail Shapley indices are close to an affine transform of the tail weights in the following sense: \label{en:main-small-crit-shapley}
\[
\sum_{i\in T} \parens*{\fc(i) - \dA w_i  + \dB}^2 \le  \bigO{\eps} \ ,
\]
where $\dA, \dB$ are the values defined in~\eqref{eq:affine-constants} and $w_T$ vector is $\tau^\ast$-regular.
\end{enumerate}
\end{theorem}

\begin{proof}

The proof is split into two main cases, according to whether the critical index is large (in which case we show that~\Cref{en:main-large-crit-shapley} holds) or the critical index is small (in which case we show that~\Cref{en:main-small-crit-shapley} holds).

\paragraph{Large critical index case $(k_{\text{critical}}(\tau^\ast)>k^\ast)$:}\ \\
Without loss of generality, we can assume that the coordinates of the vector $w$ are sorted by magnitude, so $|w_1| \geq \cdots \geq |w_n|$.
In this case, we will show that \Cref{en:main-large-crit-shapley} holds with $\hat{k}:= k^\ast-1$ and $H:=\{1,\cdots,k^\ast-1\}$.
Indeed, we have that
\begin{align*}
\norm{(w_j)_{j=k^\ast}^n}_1^2 &\leq n \cdot \norm{(w_j)_{j=k^\ast}^n}_2^2 & (\text{Cauchy-Schwarz Inequality})\\
                             &\leq n \cdot (1 - (\tau^\ast)^2)^{k^\ast/2} \cdot \norm{(w_j)_{j=k^\ast/2}^n}_2^2 &(\text{\Cref{lem:geom-decreasing-tail} with $\alpha=k^\ast/2,\beta=k^\ast$})\\
                             &\leq n^2 \cdot (1 - (\tau^\ast)^2)^{k^\ast/2} \cdot w_{k^\ast/2}^2 \\
														 &\leq w_{k^\ast/2}^2  & (k^\ast\ge 4\tfrac{\log n}{(\tau^\ast)^2} \ge 4\log_{(1\text{-}(\tau^\ast)^2)^{\text{-}1}}(n)) \ .
\end{align*}
It follows that
\begin{equation}
\norm{(w_j)_{j=k^*}^n}_1 \leq |w_{k^*/2}| \ .
\label{eq:wk2-vs-w1wt}
\end{equation}
Having established \Cref{eq:wk2-vs-w1wt}, we are ready to show that if we ``zero the tail weights'' in $f(x)$ to obtain a $(k^\ast-1)$-junta $f'(x)$ then the $\ell_1$ distance (with respect to the Shapley distribution $\DShap$) between $f(x)$ and $f'(x)$ is not too large. 

In more detail, we define the junta $\ell'(x) =\sum_{i=1}^{k^\ast-1} w_i x_i -\theta$ and $f'(x) = \sign(\ell'(x))$. We can assume without loss of generality that $|\ell(x)| \neq |\ell(x) - \ell'(x)|$ for all $x \in \bn$. If not, we can ensure this by perturbing the threshold in one of $\ell(x)$, $\ell'(x)$ slightly without changing the values of $f(x)$, $f'(x)$ for any $x$.
Then
\begin{align*}
    \Prx_{\rv{x} \sim \DShap}[f(\rv{x}) \neq f'(\rv{x})] 
    &= \Prx_{\rv{x} \sim \DShap}[\sign(\ell(\rv{x})) \neq \sign(\ell'(\rv{x}))]  \\
    &\leq \Prx_{\rv{x} \sim \DShap}[|\ell(\rv{x})| < |\ell(\rv{x}) - \ell'(\rv{x})|] \\
    &\leq \Prx_{\rv{x} \sim \DShap}[|\ell(\rv{x})| < {\textstyle \sum_{i=k^\ast}^n |w_i|}]\\  
    &\leq \Prx_{\rv{x} \sim \DShap}[|\ell(\rv{x})| < |w_{k^\ast/2}|] \\
    &\leq O((\log n)^{-1} \cdot (k^\ast/2)^{-1/6}) \ , 
\end{align*}
where the penultimate inequality follows by \Cref{eq:wk2-vs-w1wt} and the final inequality follows by applying \Cref{thm:shapley-anti-concentration}, with its parameters set to $r := |w_{k^\ast/2}|$, $k := k^\ast/2$, and $\eta := \eta$.
Therefore, by~\Cref{cor:cor-ub-shapley-by-l1},
\[
\displaystyle \dShapley(f, f') \leq O\parens*{\sqrt{\log n \cdot \Prx_{\rv{x} \sim \DShap}[f(\rv{x}) \neq f'(\rv{x})]} + \frac{1}{\sqrt{n}}} \leq \bigO{
(k^\ast)^{-1/12} + \frac{1}{\sqrt{n}}
} \leq \bigO{\eps} .
\]

\paragraph{Small critical index case $(k_{\text{critical}}(\tau^\ast) \leq k^\ast)$:}\ \\
It remains to analyze the case that the $(\tau^\ast)$-critical index $k_{\text{critical}}(\tau^\ast)$ is at most $k^\ast$.
In this case, we will show that \Cref{en:main-small-crit-shapley} holds with $\hat{k}:= k_{\text{critical}}(\tau^\ast)$, $H:= \{1,\cdots,k_{\text{critical}}(\tau^\ast)\}$, and $T:= \{k_{\text{critical}}(\tau^\ast)+1,\cdots,n\}$. 

By the definition of the critical index, it is easy to check that $\vec{w}_T$ is $(\tau^\ast)$-regular.
Thus, as an immediate application of \Cref{thm:shapley-affine-tail} with $\delta=\frac{1}{n^{2}}$, $|H|=k_{\text{critical}}\le k^\ast$ and $\tau^\ast=\tfrac{\eps^2}{\log^4 n}$, we get that there exist real values $\dA=\dA(w_H,\|w_T\|_1,\|w_T\|_2,\theta,\delta=\frac{1}{n^2}), \dB=\dB(w_H,\|w_T\|_1,\|w_T\|_2,\theta,\delta=\frac{1}{n^2})$ such that 
\begin{align*}
\sum_{i\in T} \parens*{\fc(i) - \dA w_i  + \dB}^2 &\le  \bigO{\tfrac{1}{n} + \log^2(n) \cdot \parens*{\frac{k^\ast}{n} + \sqrt{\tau^\ast}}} \\
 &\le \bigO{\eps} \ ,
\end{align*}
 where the second inequality holds by the definition of $\tau^\ast$, because $\eps\ge\tfrac{1}{n^{1/14}}\ge\tfrac{1}{n}$, and because $k^\ast\le\max\{\log^{14} n  ,n^{12/14}\}$.
\end{proof}

\subsection{An algorithm for the Partial Inverse Shapley Index Problem}

In this section, we show how to leverage the structural result~\Cref{thm:main-structural-shapley} to give an algorithm for recovering the weights of an LTF that are very close to being consistent with a subset $S$ of its Shapley indices.

\subsubsection{Recovering tail weights by dynamic programming}

We start by presenting a subroutine, $\textsc{RecoverWeights}$, for recovering weights $w = (w_1, \ldots, w_n)$ corresponding to an LTF that (approximately) minimizes the objective function 
\[
\sum_{i \in S} (\fc(i) - \dA w_i + \dB)^2,
\] subject to certain constraints. Using the characterization in~\Cref{thm:main-structural-shapley},~\Cref{en:main-small-crit-shapley} of the tail Shapley values as affine functions of their corresponding input weights, this will allow us to output an LTF $g$ with small $\dPS{S}(f, g)$ for functions $f$ with low critical index.

The algorithm takes the following as input: 

\begin{itemize}

\item A set $\mathcal{S} = \set{(\alpha_i, i) : i \in S}$ of values $\alpha_i$ (to be thought of as approximations of Shapley indices $\fc(i)$ of an LTF $f(x) = \sgn(\sum_{i=1}^n v_i x_i - \theta)$) and corresponding indices for some subset $S \subseteq [n]$;

\item  the number of weights $n$;

\item a granularity parameter $\gamma$; 

\item  target $\ell_1$ and $\ell_2$ norm values $W_1$ and $W_2$ for the weight vector $w$, with $W_1$ an integer multiple of $\gamma$ and $W_2$ an integer multiple of $\gamma^2$;

\item  a regularity parameter $\tau$;

\item and constants $\dA, \dB$.

\end{itemize}
The algorithm outputs a vector of non-negative weights $w = (w_1, \ldots, w_n)$ that minimizes the objective function $\sum_{i \in S} (\alpha_i - \dA w_i + \dB)^2$ subject to the constraints that each $w_i$ is an integer multiple of $\gamma$, each $w_i \leq \tau W_2$, $\norm{w}_1 = W_1$, and $\norm{w}_2 = W_2$.

The algorithm works by dynamic programming. It constructs a table $T$ indexed by three values: an index $i \in \set{1, \ldots, n}$, a target $\ell_1$ norm value $W_1' \in \set{0, \gamma, 2\gamma, \ldots, W_1}$, and a target $\ell_2$ norm value $W_2' \in \set{0, \gamma, \sqrt{2} \gamma, \ldots, W_2}$, where $W_1$ is an integer multiple of $\gamma$ and $W_2^2$ is an integer multiple of $\gamma^2$. Each entry $T(k, W_1', W_2')$ contains a weight vector prefix  $w = (w_1, \ldots, w_k)$ of length $k$ that minimizes the objective function $\sum_{i \in S \cap [k]} (\alpha_i - \dA w_i + \dB)^2$ over all weight vector prefixes $w' = (w_1', \ldots, w_k')$ satisfying the constraints that each $w_i'$ is an non-negative integer multiple of $\gamma$, $w_i' \leq \tau W_2$, $\norm{w'}_1 = W_1'$, and $\norm{w'}_2 = W_2'$. (The entry in $T(k, W_1', W_2')$ contains $\perp$ if no such vector $w'$ exists.)

The algorithm works by constructing ``layers'' of $T$ indexed by $k$, starting from $k = 1$, and constructing layer $k$ from layer $k - 1$ for $k = 2, \ldots, n$. Its final output is $w := T(n, W_1, W_2)$.

\begin{algorithm}
Initialize an $n \times (W_1/\gamma + 1) \times (W_2^2/\gamma^2 + 1)$ table $T$, by setting each of its entries to $\perp$.
\bigskip

\For{$w_1 \in \set{0, \gamma, 2\gamma, \ldots, \tau W_2}$}{
Set $T[1, w_1, w_1] \gets w_1$.
}
\bigskip
\For{$k = 2, \ldots, n$}{
\For{$W_1' \in \set{0, \gamma, 2\gamma, \ldots, W_1}$, $W_2' \in \set{0, \gamma, \sqrt{2} \gamma, \ldots, W_2}$}{
\tcp{Identify a feasible set of weight vectors of length $k$ achieving the desired $\ell_1$ and $\ell_2$ norm.}
$V \gets \set{w = (w^*, w_k) : w_k \in \{0, \gamma, 2\gamma, \ldots, \tau W_2}, w_k \leq W_2'$,\\ \qquad\qquad\qquad\qquad\qquad $w^* = T[k - 1, W_1' - w_k, ((W_2')^2 - w_k^2)^{1/2}] \neq \perp\}$. \\
\bigskip
\tcp{Identify a feasible weight vector minimizing the objective function.}
$T[k, W_1', W_2'] \gets \textrm{arg}\,\min_{w \in V} \sum_{i \in S \cap [k]} (\alpha_i - \dA w_i + \dB)^2$ (or $\perp$ if $V = \emptyset$).

}
}
\bigskip
\Return $T[n, W_1, W_2]$.
\caption{$\textsc{RecoverWeights}(\mathcal{S}, n, \gamma, W_1, W_2, \tau, \dA, \dB)$}
\label{alg:dp-weights}
\end{algorithm}

\begin{theorem}
The procedure $\textsc{RecoverWeights}(\mathcal{S}, n, \gamma, W_1, W_2, \tau, \dA, \dB)$ outputs a weight vector $w$ satisfying the conditions that $w = \gamma z$ for some $z \in (\Z^{\geq 0})^n$, $\norm{w}_1 = W_1$, $\norm{w}_2 = W_2, \norm{w}_{\infty}/\norm{w}_2 \leq \tau$ that minimizes the objective function $\sum_{i \in S} (\alpha_i - \dA w_i + \dB)^2$ over all weight vectors satisfying those conditions.
Moreover, $\textsc{RecoverWeights}$ runs in $\poly(n, 1/\gamma, W_1)$ time.
\label{thm:recover-weights-correctness}
\end{theorem}

\begin{proof}
We prove by induction on $k$ that $T[k, W_1', W_2']$ contains a vector $w = (w_1, \ldots, w_k)$ with $\norm{w}_1 = W_1'$ and $\norm{w}_2 = W_2'$ that minimizes $\sum_{i \in S \cap [k]} (\alpha_i - \dA w_i + \dB)^2$ if such a vector exists. The base case where $k = 1$ is clear.

For the inductive case, assume that there exists a weight vector that satisfies all of the required conditions, and let $w = (w_1, \ldots, w_k)$ denote a vector that minimizes the quantity $\sum_{i \in S \cap [k]} (\alpha_i - \dA w_i + \dB)^2$ among all satisfying vectors.
Consider
$w^* = (w_1^*, \ldots, w_{k-1}^*) = T[k - 1, W_1' - w_k^*, ((W_2')^2 - (w_k^*)^2)^{1/2}]$,
which must exist and satisfy
$\sum_{i \in S \cap [k-1]} (\alpha_i - \dA w_i^* + \dB)^2 \leq \sum_{i \in S \cap [k-1]} (\alpha_i - \dA w_i + \dB)^2$
by the induction hypothesis. The algorithm will therefore consider the pair $w' = (w^*, w_k)$, which is optimal by the assumption that $w$ is optimal, as needed.

We next turn to analyzing the algorithms's runtime. The table $T$ used in $\textsc{RecoverWeights}$ has $O(n \cdot W_1/\gamma \cdot W_2^2/\gamma^2) = O(n \cdot W_1^3/\gamma^3)$ entries. Updating each of these entries (other than those in the first layer) requires computing $V$, which takes $O(\tau W_2/\gamma) = O(W_1/\gamma)$ time, and checking which $w \in V$ minimizes the objective function, which takes $O(\card{V} \cdot \card{S}) = O(n \cdot W_1/\gamma)$ time. The algorithm's runtime is dominated by the total time required to update these entries, which is at most $O(n^2 \cdot W_1^4/\gamma^4)$.
\end{proof}

\subsubsection{Main algorithm for the Partial Shapley Values Problem}

Now we are ready to present the main algorithm for the Partial Shapley Values Problem.
This algorithm takes as input a set of Shapley values and corresponding indices $\mathcal{S} = \set{(\fc(i), i) : i \in S}$ for some $S \subseteq [n]$ of an LTF \red{$f(x) = f(x_1, \ldots, x_n) = \sgn(\sum_{i=1}^n v_i x_i - \theta)$}. The algorithm is analogous to the algorithm in~\Cref{subsec:main-alg-chow} for the Partial Chow Parameters Problem, and works in three steps.

In the first step, the algorithm sets parameters and guesses the size of the head $H$ and tail $T$ indices of $f$. As in the first step of the Chow algorithm, the algorithm will only need to know (guess) $\card{H}$ and $\card{H \cap S}$; note that fixing a guess for $\card{H \cap S}$ fixes the corresponding set $H \cap S$ by \Cref{lem:monotonicity-shapley}, and also fixes the sizes and identities of $T \cap S$ and $\card{T}$.  (How the indices not in $S$ are partitioned between $H$ and $T$ is irrelevant since any permutation of indices not in $S$ will result in candidate LTFs $g$ with the same Partial Shapley Distance with respect to $S$, $\dPS{S}(f, g)$.)
 In the second step, the algorithm enumerates all LTFs in a relatively small (quasipolynomial size) set based on the structural result in~\Cref{thm:main-structural-shapley}. Enumerating the LTFs in this set is more nuanced than in the corresponding step in the Partial Chow Parameters Problem, and requires guessing additional values and calling the dynamic programming routine $\textsc{RecoverWeights}$ from the previous section.
In the final step, the algorithm checks which of the candidate LTFs $g$ generated in the previous step satisfies $\dPS{S}(f, g) \leq O(\eps)$, and outputs one of them. This final step domaintes the algorithm's runtime, which is again quasipolynomial.

The idea behind the algorithm's correctness corresponds to the two cases in~\Cref{thm:main-structural-shapley}. In the first case (the ``large critical index'' case), we will enumerate all junta LTFs $g$ on $H$ whose weights are discretized to some precision $\gamma$.
In the second case (the ``small critical index'' case), we will enumerate all LTFs $g$ of a particular form.
We will start by considering a discretized version $f' = f'(x_H, x_T) = \sgn(w_H' \cdot x_H + w_T' \cdot x_T - \theta')$ of $f$ whose weights $w_1', \ldots, w_n'$ and threshold $\theta'$ are integer multiples of $\gamma$.
The goal of the algorithm will be to (approximately) recover $f'$, which by~\Cref{thm:discretization} is close in Shapley distance to $f$.

The head weights $w_H$ of $g$ are set to be those of $f'$, which are guessed to some precision $\gamma$ (in a similar way to the large critical index case). The tail weights $w_T$ of $g$ are set to be (roughly) affine functions of the input Shapley indices $\fc(i)$ by calling $\textsc{RecoverWeights}$ on input values $\alpha_i = \fc(i)$ for $i \in S \cap T$.
(Although our overall goal is to recover an approximation of $f$, it is useful to think of $f'$ as the ``ground truth'' function whose tail weights we're trying to recover via the call to the subroutine $\textsc{RecoverWeights}$, and of the input values $\alpha_i = \fc(i)$ to $\textsc{RecoverWeights}$ as noisy versions of $\fc'(i)$.)

By two applications of the triangle inequality,
\begin{align*}
\dPS{S}(f, g) &\leq \dPS{S}(f, f') + \dPS{S}(f', g) \\
              &\leq \dPS{S}(f, f') + \dPS{H \cap S}(f', g) + \dPS{T \cap S}(f', g) \ .
\end{align*}

We will show that $\dPS{S}(f, g) \leq O(\eps)$ by upper bounding each of the three terms on the right-hand side. 
Roughly speaking, we will show that $\dPS{S}(f,f') \leq O(\eps)$ by the discretization result in~\Cref{thm:discretization}, that $\dPS{H \cap S}(f',g) \leq O(\eps)$ by the head Shapley index stability result in~\Cref{thm:shapley-head-coeffs}, and that $\dPS{T \cap S}(f',g ) \leq O(\eps)$ by the result showing that tail weights are affine functions of their corresponding Shapley indices in~\Cref{thm:shapley-affine-tail}. 

We next present the full algorithm and analysis for the Partial Shapley Values Problem.

\bigskip\ \\

{\centering 
\fbox{
\begin{minipage}{16cm}
\begin{enumerate}
    \item
    \begin{enumerate}
        \item
 Define $\tau := \tfrac{\eps^2}{\log^4 n}$ and \red{$k := \max\{4 \tfrac{ \log^{9} n}{\eps^{4}},\tfrac{1}{\eps^{12}}\}$}.
as in~\Cref{thm:main-structural-shapley}. 
    Fix the granularity parameter $\gamma := 1/(n^2 \cdot k^{k/2})$ as in~\Cref{thm:discretization}.
   \item
    Guess the size of the head $\card{H} \in [k]$ and the size of $|H \cap S|$. Identify the $|H \cap S|$ elements of $S$ for which $\fc(i)$ is largest as the corresponding guess for $H \cap S$. Set $H$ equal to the union of $H \cap S$ and $\card{H} - \card{H \cap S}$ arbitrary indices not in $S$. Set $T$ equal to $[n] \setminus H$.
        \end{enumerate}
    \label{en:shap-alg-param-set}
    \item For each setting of $H$ in Step~\ref{en:shap-alg-param-set}, enumerate all LTFs of the following forms (corresponding to the two cases in~\Cref{thm:main-structural-shapley}):
    \begin{enumerate}
        \item \label{en:shap-alg-junta-case} Enumerate all junta LTFs $g$ on $H$.
        \item \label{en:shap-alg-non-junta-case} Enumerate all LTFs $g$ of the form
        \[
        g(x) = g(x_H, x_T) = \sign(w_H \cdot x_H + w_T \cdot x_T - \theta) \ ,
        \]
obtained by enumerating all combinations of a number of values, and then setting $w_1, \ldots, w_n, \theta$ according to the subsequent procedure.

\medskip
Enumerate the following:
				\begin{enumerate}
				\item Head weights $w_H$ with $w_i \in \set{0, \gamma, 2\gamma, \ldots, 1}$ for $i \in H$,
				\item The threshold $\theta \in \set{0, \gamma, 2\gamma, \ldots, n}$,
				\item $W_1 \in \set{\gamma, 2\gamma, \ldots, n}$,
				\item $W_2 \in \set{\gamma, \sqrt{2}\gamma, \ldots, n}$.
				\end{enumerate}
				\medskip
				Set $w_1, \ldots, w_n, \theta$ as follows:
				\begin{enumerate}
				\item Set the head weights $w_H$ and threshold $\theta$ equal to the enumerated values. 
				\item Compute $\dA = \dA(w_H, W_1, W_2, \theta, 1/n^2)$, $\dB = \dB(w_H, W_1, W_2,\theta, 1/n^2)$ using the formulas in~\Cref{eq:affine-constants}.
				\item Set the tail weights as\\ $w_T := \textsc{RecoverWeights}(\set{(\fc(i)), i) : i \in T}, \card{T}, \gamma, W_1, W_2, \tau, \dA, \dB)$. %
				\end{enumerate}
    \end{enumerate}
    \label{en:shap-alg-enumeration}
    \item For each candidate LTF $g$ generated in Step~\ref{en:shap-alg-enumeration}, compute an empirical estimate $\bar{g}(i)$ of each of the Shapley Values $\gc(i)$ for $i \in S$ so that $|\gc(i) - \bar{g}(i)| \leq \eps/\sqrt{|S|}$ with confidence $1 - \delta_{\fail}/(\card{S} \cdot M)$, where $M$ is the total number of LTFs enumerated in Step~\ref{en:main-alg-enumerate-candidate-ltfs}.
Output (a weights-based representation of) the first $g$ such that $\norm{(\fc(i))_{i \in S} - (\bar{g}(i))_{i \in S}} \leq O(\eps)$.
    \label{en:shap-alg-verification}
\end{enumerate}
\end{minipage}
}}
\bigskip
\bigskip
\begin{theorem}
There exists an algorithm for the Partial Inverse Shapley Index Problem with the following guarantees. It takes as input four things: (1) a set $\set{(i, \fc(i)) : i \in S}$ for some $\eta$-restricted, monotone increasing LTF $f : \pmo^n \to \pmo$ with $\eta \in (1/4, 1]$ and some $S \subseteq \set{0, 1, \ldots, n}$, (2) the length $n$ of the input to $f$, (3) an error parameter $\eps \in (\frac{1}{n^{1/14}}, \half)$, and (4) a confidence parameter $\delta_{\fail} > 0$. It outputs a weights-based representation of an LTF $g : \pmo^n \to \pmo$ such that $\dPS{S}(f, g) \leq O(\eps)$ with probability $1 - \delta_{\fail}$, and runs in time $2^{\tilde{O}(\log^{18} n/\eps^{24})} \cdot \log(1/\delta_{\fail})$.
\label{thm:main-algorithmic-shapley}
\end{theorem}

\begin{proof}
We start by arguing that the above algorithm is correct, beginning with analysis similar to that in the Chow algorithm.
By taking a union bound, it holds that all $|S| \cdot M$ estimates $\bar{g}(i)$ of the Shapley Indices $\gc(i)$ of candidate LTFs $g$ with $i \in S$ in Step~\ref{en:main-alg-check-candidate-ltfs} will be accurate to within a $\eps/\sqrt{|S|}$ additive error factor with probability at least $1 - \delta_{\fail}$.
In this case our estimates will all satisfy $\norm{(\gc(i))_{i \in S} - (\bar{g}(i))_{i \in S}} \leq \eps$, and hence by the triangle inequality $\norm{(\fh(i))_{i \in S} - (\bar{g}(i))_{i \in S}} - \eps \leq \dPS{S}(f, g) \leq \norm{(\fc(i))_{i \in S} - (\bar{g}(i))_{i \in S}} + \eps$ for every candidate LTF $g$. So, in this case, we will output a candidate LTF $g$ if and only if it satisfies $\dPS{S}(f, g) \leq O(\eps)$.

In terms of correctness, it remains to show that one of the enumerated LTFs $g$ satisfies\\ $\dPS{S}(f, g) \leq O(\eps)$.
By~\Cref{thm:discretization} we have that there exists an LTF $f'$ such that
$\dShapley(f, f') = O(\eps)$ with weights and a threshold which are integer multiples of $\gamma$ -- as defined in Step~1 of the algorithm, $\gamma := 1/(n^2 \cdot k^{k/2})$.

We consider the two cases in~\Cref{thm:main-structural-shapley} for function $f'(x) = \sign(\sum_{i=1}^n w_i' x_i - \theta')$.
In the first case, $\dShapley(f',g) = O(\eps)$ for some junta LTF $g$ on $H$. All such discretized junta LTFs on $H$ are enumerated in Step~\ref{en:shap-alg-junta-case}, and so the algorithm enumerates a $g$ satisfying $\dPS{S}(f, g) \leq \dShapley(f, f') + \dShapley(f', g) \leq O(\eps)$, as needed. 

In the second case, we're guaranteed that there exist constants $\dA', \dB'$ such that $\sum_{i \in T} (\fc'(i) - \dA' w_i' + \dB')^2 \leq O(\eps)$ \red{and 
$w_T'$ is $\tau-$regular, where constants $\dA', \dB'$ are given by~\Cref{eq:affine-constants} for $\delta=\frac{1}{n^2}$ and $w',\theta'$}. 
We will show that there exists a function $g$ enumerated in Case~\ref{en:shap-alg-non-junta-case} that satisfies $\dPS{S}(f, g) \leq O(\eps)$. To do this, we observe that by two applications of triangle inequality,
\begin{align*}
\dPS{S}(f, g) &\leq \dPS{S}(f, f') + \dPS{S}(f', g) \\
              &\leq \dPS{S}(f, f') + \dPS{H \cap S}(f', g) + \dPS{T \cap S}(f', g) \ .
\end{align*}
We will show that each of the three terms in the right hand side is upper bounded by $O(\eps)$ in turn.
As established earlier, $\dPS{S}(f, f')\leq O(\eps)$. 

Next, we argue that $\dPS{H \cap S}(f', g) \leq O(\eps)$. Because $f'$ has weights and a threshold that are integer multiples of $\gamma$, the algorithm will enumerate guesses $w_H$, $\theta$, $W_1$, and $W_2$ that are equal to the head weights $w_H'$, threshold $\theta'$, $\ell_1$ norm of the tail weights $\norm{w_T'}_1$, and $\ell_2$ norm of the tail weights $\norm{w_T'}_2$ of $f'$, respectively. For such correct guesses, the procedure $\textsc{RecoverWeights}$ will output tail weights $w_T$ of $g$ such that $\norm{w_T}_1 = \norm{w_T'}_1$ and $\norm{w_T}_2 = \norm{w_T'}_2$.
Observe that $w_T, w_T'$ are $\tau$-regular. Thus, 
by the head Shapley index stability result in~\Cref{thm:shapley-head-coeffs} (with $\delta = 1/n^2$), we get $\dPS{H \cap S}(f', g) \leq O(\eps)$.

{
\red{
}}

Finally, we argue that $\dPS{T \cap S}(f', g) \leq O(\eps)$. Recall that 
\begin{equation}\label{eq:fp-prop-wip}
\sum_{i \in T} (\fc'(i) - \dA' \cdot w_i' + \dB')^2 \leq O(\eps)   .
\end{equation}
Applying $\dPS{S}(f, f')\leq O(\eps)$, we get that 
\begin{equation}\label{eq:fp-prop-wip1}
\sum_{i \in T} (\fc(i) - \dA' \cdot w_i' + \dB')^2 \leq O(\eps)   .
\end{equation} 
Consider 
the case when the algorithm has correctly guessed the head weights $w_H'$, threshold $\theta'$, $\ell_1$ norm of the tail weights $\norm{w_T'}_1$, and $\ell_2$ norm of the tail weights $\norm{w_T'}_2$ of $f'$, respectively. By its correctness, $\textsc{RecoverWeights}$ will therefore output tail weights $w_T$ of $g$ and threshold $\theta$ which satisfy (i) $\norm{w_T'}_1 = \norm{w_T}_1$ and $\norm{w_T'}_2 = \norm{w_T}_2$ and (ii) $\theta = \theta'$  and the following holds: 
\begin{equation} \label{eq:fp-prop-wip3}
\sum_{i \in T} (\fc(i) - \dA' \cdot w_i + \dB')^2 \leq O(\eps) .
\end{equation} 
However, since by its correctness, $ \textsc{RecoverWeights}$ will also guarantee that the weight vector $w_T$ is $\tau$-regular, it will imply that (by ~\Cref{eq:affine-constants}), 
\begin{equation}
\sum_{i \in T} (\gc(i) - \dA' \cdot w_i + \dB')^2 \leq O(\eps)  .
\label{eq:fp-prop-wip2}
\end{equation} 
Thus, applying both \Cref{eq:fp-prop-wip3} and \Cref{eq:fp-prop-wip2}, 
\begin{equation}
\sum_{i \in T} (\gc(i) - \fc(i))^2 \leq O(\eps). 
\end{equation} 
Combining with  the fact that $\dShapley(f, f') \leq O(\eps)$ and applying triangle inequality, this implies 
$\dPS{T \cap S}(f', g) \leq O(\eps)$.

We now turn to analyzing the runtime of the algorithm.
We start by analyzing how many LTFs $M'$ are enumerated in Step~\ref{en:shap-alg-enumeration} for fixed guesses of $\card{H}$ and $\card{H \cap S}$ in Step~\ref{en:shap-alg-param-set}.
This number is asymptotically dominated by Case~\ref{en:shap-alg-non-junta-case},
where there are $O((1/\gamma)^k \cdot n/\gamma)$ possible choices for $\theta$ and $w_i$ for $i \in H$, $O(n/\gamma)$ choices for $W_1$, and $O(n^2/\gamma^2)$ choices for $W_2$.
Because $\gamma = 1/(n^2 \cdot k^{k/2})$ and $k \leq 4 \log^{9} n/\eps^{12}$, we enumerate a total of 
\[
M' = (1/\gamma)^{k+4} \cdot \poly(1/\eps) %
= n^{O(\log^9 n/\eps^{12})} \cdot (\log^9 n/\eps^{12})^{O(\log^{18} n/\eps^{24})}
= 2^{\tilde{O}(\log^{18} n/\eps^{24})} \ .
\]
LTFs in Step~\ref{en:shap-alg-enumeration} (for fixed $\card{H}$).

In Step~\ref{en:shap-alg-param-set}, we make $O(k^2) = O(n^2)$ guesses for $\card{H}$ and $\card{H \cap S}$, so we get that the total number of LTFs enumerated by the algorithm is 
\begin{equation}
M = M' \cdot O(n^2) = 2^{\tilde{O}(\log^{18} n/\eps^{24})} \ .
\label{eq:shapley-num-ltfs-ub}
\end{equation}

Computing $\dA$, $\dB$ requires evaluating the formulas in \Cref{eq:affine-constants}, which is efficient. 
For each guess of $w_i$ for $i \in H$, $\theta$, $W_1$, $W_2$, we call $\textsc{RecoverWeights}$, which runs in $\poly(n, 1/\gamma, W_1) = \poly(n, 1/\gamma) = 2^{\tilde{O}(\log^9 n/\eps^{12})}$ time (since $W_1 \leq 1$) which is asymptotically dominated by the upper bound on $M$ in~\Cref{eq:shapley-num-ltfs-ub}. So, the total time needed to enumerate the $M$ LTFs is also $2^{\tilde{O}(\log^{18} n/\eps^{24})}$.

Concluding, the algorithm enumerates $M = 2^{\tilde{O}(\log^{18} n/\eps^{24})}$ LTFs in $\tilde{O}(\log^{18} n/\eps^{24})$ time, so the algorithm's runtime is dominated by the time needed to compute  estimates $\bar{g}(i)$ of the Shapley indices in Step~\ref{en:shap-alg-verification} for each of the $M$ functions $g$ enumerated in Step~\ref{en:shap-alg-enumeration}. The total runtime of the algorithm is the same function of $M$ as in the Chow algorithm (\Cref{eq:chow-step-3-time}), which is 
\[
O\Big(\frac{n \cdot |S|^2}{\eps^2} \log\Big(\frac{|S| \cdot M}{\delta_{\fail}}\Big)\Big) M \Big) = 2^{\tilde{O}(\log^{18} n/\eps^{24})} \cdot \log(1/\delta_{\fail}) \ . \qedhere
\] 
\end{proof}

%
%
%
%
%
%
%
%
%

%
%
%
%

\bibliography{allrefs,shapley}
\bibliographystyle{alpha}

\appendix

\section{Fourier and Hermite analysis} \label{ap:fourier}

\subsection{Fourier analysis over $\bn$}

Viewing $\bn$ as endowed with the uniform probability distribution, the set of real-valued functions over $\bn$ forms a $2^n$-dimensional inner product space with inner product given by $\inner{f}{g} := \E_{\bx}[f(\bx)g(\bx)]$. The set of functions
$(\chi_S)_{S\subseteq [n]}$ defined by $\chi_S := \prod_{i\in S}x_i$ forms a complete orthonormal basis for this space.
Given a function $f : \pmo^n \to \pmo$ we define its \emph{Fourier coefficients} by $\fh(S) := \E_{\bx}[f(\bx)\chi_S(\bx)]$, and we have that the \emph{Fourier representation} of $f$ is
$f(x) = \sum_{S\subseteq [n]}\fh(S)\chi_S$ (note that this is the unique representation of $f$ as a multilinear real polynomial).

We will be particularly interested in $f$'s degree-$1$ coefficients, i.e., $\fh(S)$ for $|S| = 1$; we will write these as $\fh(i)$ rather than $\fh(\{i\})$, and we note that these correspond precisely to the Chow Parameters of $f$.
Finally, we recall \emph{Plancherel's identity}, which states that $\inner{f}{g}=\sum_{S\subseteq [n]}\fh(S)\gh(S)$, and the special case of \emph{Parseval's identity}, which states that $\E_{\bx}[(f(\bx))^2] = \sum_{S\subseteq [n]}\fh(S)^2 =1$.

\subsection{Hermite analysis over $\R^n$}

Here we consider functions $f : \R^n \to \R$, where we think of the inputs $x$ to $f$ as being
distributed according to the standard $n$-dimensional Gaussian  distribution $N(0,1)^n$. In this context we view the space of
all real-valued square-integrable functions as an inner product space with inner product $\inner{f}{h}= \E_{\bx \sim N(0,1)^n}[f(\bx)h(\bx)]$.
In the case $n = 1$, there is a sequence of Hermite polynomials $h_0(x) \equiv 1, h_1(x) = x, h_2(x) = (x^2 -1)/\sqrt{2},\ldots$ that form a complete orthonormal basis for the space.  These polynomials can be defined via $\exp(\lambda x-\lambda^2/2)=\sum{d=0}^{\infty}(\lambda^d/\sqrt{d!}) h_d(x)$. In the case of general $n$, we have that the collection of $n$-variate
polynomials $\{H_S(x) := \prod_{i=1}^{n} h_{S_i}(x_i)\}_{S \in \N^n}$ forms a complete orthonormal basis for the space. Given a square integrable function $f : \R^n \to \R$ we define its Hermite coefficients by \ft(S) = \inner{f}{H_S}, for $S\in \N^n$
and we have that $f(x) = \sum_{S}\ft(S)H_S(x)$ (with the equality holding in $\calL^2$). Again, we will be particularly interested in $f$'s degree-1 coefficients, i.e., $\ft(e_i)$, where $e_i$ is the vector which is 1 in the $i$-th coordinate and 0 elsewhere; observe that $\ft(e_i)$ is  $\E_{\bx \sim N(0,1)^n)}[f(\bx)\bx_i]$.  Plancherel's and Parseval's identities are easily seen to hold in this setting.

\section{Useful inequalities} \label{ap:inequalities}

In this section we record some useful elementary inequalities. 

\begin{fact}\label{fact:easy-bound}
Suppose that $A,B$ are non-negative and $|A-B|\le \eta$. Then $|\sqrt{A}-\sqrt{B}|\le \dfrac{\eta}{\sqrt{B}}$.
\end{fact}
\begin{proof}
$|\sqrt{A}-\sqrt{B}|=\dfrac{|A-B|}{|\sqrt{A}+\sqrt{B}|}\le \dfrac{\eta}{\sqrt{B}}$.
\end{proof}

\begin{fact}\label{fact:triangular-in-square}
Let $\vect{a},\vect{b},\vect{c}\in\R^n$ with $\normtwo{a-b}^2=O(\eps_1)$ and $\normtwo{b-c}^2=O(\eps_2)$.
Then \[\normtwo{a-c}^2\le \bigO{\eps_1+\eps_2}. \]
\end{fact}
\begin{proof}
It is easy to verify that $\normtwo{x-y}^2\le 2\normtwo{x}^2+2\normtwo{y}^2$, and consequently we have that
 $\normtwo{a-c}^2\le 2\normtwo{a-b}^2+2\normtwo{b-c}^2 \le O(\eps_1+\eps_2)$.
\end{proof}

Given any  vector $v\in \R^n$, let us write $v_{\|}$ to denote
\[
v_{\|}:=\frac{\sum_{i\in n}v_i}{n}(1,\cdots,1),
\]
which we call the \emph{centralized vector} of $v$.

\begin{fact}\label{fact:centering}
Let $\vect{a},\vect{b}\in\R^n$ be such that $\normtwo{\vect{a}-\vect{b}}\le \eta$.
Then for any constant $c\in [0,1]$ it holds that 
\[
c\normtwo{\vect{a}_{\|}-\vect{b}_{\|}}\le {\eta},  \quad \Big|\frac{\sum_{i\in n}a_i}{n} - 
\frac{\sum_{i\in n}b_i}{n}  \Big| \le \frac{\eta}{\sqrt{n}}, \quad \text{and} \quad
\normtwo{(a-c\vect{a}_{\|})-(b-c\vect{b}_{\|})}\le \bigO{\eta}.
\]
\end{fact}
\begin{proof} It suffices to prove only the first claim since the second and the third one can be obtained from the first via the triangle inequality. For the first we have
\[\normtwo{\vect{a}_{\|}-\vect{b}_{\|}}=\normtwo{(1,\cdots,1)} \cdot |\sum_{i\in n}\frac{a_i-b_i}{n} | \le \frac{1}{\sqrt{n}} 
\sum_{i\in n}|a_i-b_i|\le \frac{1}{\sqrt{n}} \|a-b\|_1 \le 
\frac{\sqrt{n}}{\sqrt{n}} \|a-b\|_2\le \eta.
\]
\end{proof}

\begin{fact}\label{fact:close-l2-norm}
Let $\vect{a},\vect{b}\in \R^m$ with $\normtwo{\vect{a}}\le 1,\normtwo{\vect{b}}\le 1$ such that $\normtwo{a-b}^2\le \eta$. Then 
\[
\Big|\normtwo{a}-\normtwo{b}\Big|\le \sqrt{\eta}
\quad \text{and} \quad
\Big|\normtwo{a}^2-\normtwo{b}^2\Big|\le 2\sqrt{\eta}.
\]
\end{fact}
\begin{proof}
The first claim holds by the triangle inequality, since
$\Big|\normtwo{a}-\normtwo{b}\Big|\le \normtwo{a-b} $. For the second claim we have that 
$\displaystyle\Big|\sum_{i=1}^{m}(a_i^2-b_i^2)\Big|=
\Big|\sum_{i=1}^{m}(a_i-b_i)(a_i+b_i)\Big|\le 
\sqrt{\sum_{i=1}^{m}(a_i+b_i)^2}
\sqrt{\sum_{i=1}^{m}(a_i-b_i)^2}\le \sqrt{2 \sum_{i=1}^{m}(a_i^2+b_i^2)}\normtwo{a-b}
\le 2\sqrt{\eta}.
$
\end{proof}

\begin{fact}\label{fact:holder}
Let $\vect{a},\vect{b},\vect{c}\in \R^m$ with $\norm{\vect{a}-\vect{b}}_1 \le \eta$ and $\norm{c}_\infty\le \bigO{1}$. Then
\[
\Big|(a-b) \cdot c \Big|\le \bigO{\eta}.
\]
\end{fact}
\begin{proof}
We have that
\[
|(a - b) \cdot c| = \abs*{\sum_i (a_i - b_i) c_i} \leq \sum_i |a_i - b_i| \cdot |c_i| \leq
\|c\|_\infty \cdot \sum_i |a_i - b_i| = \|c\|_\infty \cdot \|a - b\|_1 = O(\eta)
\]
as claimed.
\end{proof}

\section{Consequences and variants of the Berry-Esseen theorem for $p$-biased linear forms} \label{ap:BE}

Recall \Cref{fct:p-biased-linear-form}:

\medskip

\noindent {\bf \Cref{fct:p-biased-linear-form}.} 
\emph{
Let $0^n \neq \vec{w} \in \R^n$ be $\tau$-regular, and let $p \in (0, 1)$. Then we have the following:}

\begin{enumerate}

\item \emph{
For any interval $[a, b] \subseteq \R \union \set{\pm \infty}$, 
\[
\left|\Prx_{\vec{\rv{x}} \sim u_p^n}\left[\vec{w} \cdot \vec{\rv{x}} \in [a, b]\right] - \left(\Phi\left(\frac{b - \mu}{\sigma}\right) - \Phi\left(\frac{a - \mu}{\sigma}\right) \right)\right| \leq \frac{4 \tau}{\sigma_p} \ ,
\]
where $\mu =\ignore{\mu_p \cdot \E_{\rv{x} \sim u^n_1}[w \cdot \rv{x}] =}\mu_p\cdot\sum_{i=1}^n w_i$ and $\sigma = \sigma_p \cdot \norm{\vec{w}}_2$.
}

\item \emph{For any $\lambda$ and any $\theta\in\R$, we have
\[
\Prx_{\bx \sim u^n_p}\left[\left| w \cdot \bx-\theta \right|\le \lambda\right]\le 2\frac{\lambda}{\sigma_p\normtwo{w}}+2\frac{\tau}{\sigma_p}.
\]
In particular, if $\lambda=O(\tau)$ and $\|w\|_2=1,$ then we have
\[\Pr[|w \cdot \bx-\theta|\le \lambda]\le \frac{O(\tau)}{\sigma_p}.\]
}
\end{enumerate}
\begin{proof}
For part (1), we apply \Cref{thm:berry-esseen} to the random variables $\rv{Y}_1, \ldots, \rv{Y}_n$ where $\rv{Y}_i = w_i \rv{x}_i - \mu_p w_i$ for $i \in [n]$. It is straightforward to check that for each $i$ we have that $\E[\rv{Y}_i] = 0$, $\E[\rv{Y}_i^2] = \sigma_i^2 = \sigma_p^2 \cdot w_i^2$, and $\E[\abs{\rv{Y}_i}^3] = 8p(1-p) \cdot (p^2 + (1-p)^2) \cdot w_i^3 \leq 8p(1-p) \cdot w_i^3$. Therefore $\sigma = \sqrt{\sum_{i=1}^n \sigma_i^2} = \sigma_p \cdot \norm{\vec{w}}_2$ and $\rho = \sum_{i=1}^n \E[|\rv{Y}_i|^3] \leq 8p(1-p) \cdot \norm{\vec{w}}_3^3 \leq 8p(1-p) \cdot \norm{\vec{w}}_2^2 \cdot \norm{\vec{w}}_{\infty}$, and hence by \Cref{thm:berry-esseen} and the $\tau$-regularity of $\vec{w}$ it holds that for any $\theta \in \R$,\begin{equation}
\big|\Pr\big[\sigma^{-1} \cdot \sum_{i=1}^n \rv{Y}_i \leq \theta\big] - \Phi(\theta)\big| \leq \frac{8p(1-p) \cdot \norm{\vec{w}}_2^2 \cdot \norm{\vec{w}}_{\infty}}{\sigma_p^3 \cdot \norm{\vec{w}}_2^3} \leq \frac{2\tau}{\sigma_p} \ .
\label{eq:B-E-bound-1}
\end{equation}

We also have 
\begin{equation}
\Prx_{\vec{\rv{x}} \sim u_p^n}[\vec{w} \cdot \vec{\rv{x}} \leq \theta] = \Pr\left[\sigma^{-1} \cdot \sum_{i=1}^n \rv{Y}_i \leq \frac{\theta - \mu}{\sigma}\right]
\label{eq:B-E-bound-2}
\end{equation}
where $\mu = \E_{\vec{\rv{x}} \sim u_p^n}[\vec{w} \cdot \vec{\rv{x}}]$ (we note for later reference that if all coefficients of $w$ are non-negative, then this value is equal to $\mu_p \cdot \norm{\vec{w}}_1$). We get part (1) of the fact by combining Equations~\eqref{eq:B-E-bound-1} and~\eqref{eq:B-E-bound-2} twice, once setting $\theta = a$ and once setting $\theta = b$.

For part (2), we have

\begin{align*}
\Prx_{\vec{\rv{x}} \sim u_p^n}[|w \cdot \bx - \theta|\le \lambda]&=\Pr[w \cdot \bx \in [\theta-\lambda,\theta+\lambda]]\\
    &\le \Phi(\psi_p^{[w]}(\theta-\lambda),\psi_p^{[w]}(\theta+\lambda))+2\frac{\tau}{\sigma_p} \le 2\frac{\lambda}{\sigma_p\normtwo{w}}+2\frac{\tau}{\sigma_p}.\\
\end{align*}
\end{proof}

Recall \Cref{lemma:marginal-approximation}:

\medskip

\noindent {\bf \Cref{lemma:marginal-approximation}.} \emph{
For $w$ a $\tau$-regular LTF, we have
\[
\Ex_{\bx \sim u^n_p}[| w \cdot \bx-\theta|] \overset{\tau\normtwo{{w}}}{\approx}  \normtwo{{w}}\sigma_p 
\Ex_{\bx \sim N(0,1)}\Big[|\bx-\psi_p^{[w]}(\theta)|\Big]=\normtwo{{w}}\sigma_p\Big(2\phi(\psi_p^{[w]}(\theta))-\psi_p^{[w]}(\theta)m(\psi_p^{[w]}(\theta))\Big).
\]
}

\medskip
\begin{proof}
\ignore{Without loss of generality we will assume that ${w}\neq {0}$, since otherwise the result is certainly true.}
The proof closely follows the proof of Proposition~32 in \cite{MORS:10} with minor changes.
Using  the fact that $\E[\br]=\int_{0}^{+\infty} \Pr[\br>s]ds$ for any nonnegative random variable $\br$ for which $\E[\br]<+\infty$, we have that:
\begin{align}
\Ex_{\bx \sim u^n_p}[|w \cdot \bx-\theta|]&=    \int_{0}^{+\infty} \Prx_{\bx \sim u^n_p}[|w \cdot \bx-\theta|>s]ds \nonumber\\
&=\int_{0}^{+\infty} \Pr[w \cdot \bx>\theta+s]ds+ \int_{0}^{+\infty} \Pr[w \cdot \bx<\theta-s]ds \nonumber\\
&=\int_{0}^{+\infty} 1- \Pr[w \cdot \bx<\theta+s]ds+ \int_{0}^{+\infty} \Pr[w \cdot \bx<\theta-s]ds. \label{eq:cow}
\end{align}

It follows from the Berry-Esseen theorem (\Cref{thm:berry-esseen}, the more detailed bound) that $\abs{(\ref{eq:cow})-(A)}\le (B)$, where
\begin{align*}
(A) &=\displaystyle\int_{0}^{+\infty}1-\Phi(\psi_p^{[w]}(\theta+s))+\Phi(\psi_p^{[w]}(\theta-s))ds,\\
(B) &=\bigO{\frac{\tau}{\sigma_p}}\displaystyle\int_{0}^{+\infty}\dfrac{1}{1+|\psi_p^{[w]}(\theta+s)|^3}+
\dfrac{1}{1+|\psi_p^{[w]}(\theta-s)|^3}ds.
\end{align*}

We have that $(B)=\bigO{\frac{\tau}{\sigma_p}}\displaystyle\int_{0}^{+\infty}\dfrac{1}{1+|\psi_p^{[w]}(\theta+s)|^3}+
\dfrac{1}{1+|\psi_p^{[w]}(\theta-s)|^3}ds=\bigO{\tau\normtwo{{w}}}$.
Turning to $(A)$, we observe that $(A)$ can be reexpressed as 
\[
(A)=\int_{0}^{+\infty}\Prx_{\bx\sim N(0,1)}\Big[|\mu_p (\sum_i w_i)+\normtwo{{w}}\sigma_p \bx -\theta|>s\Big]ds=\E\Big[|\mu_p (\sum_i w_i)+\normtwo{{w}}\sigma_p \bx -\theta|\Big].
\]
Dividing by $\normtwo{{w}}\sigma_p$, we have
\[
(A)=\E\Big[|\mu_p (\sum_i w_i)+ \normtwo{{w}}\sigma_p \bx -\theta|\Big]=\normtwo{{w}}\sigma_p \E\Big[|\bx-\psi_p^{[w]}(\theta)|\Big].
\]
Using now part(1) of \Cref{proposition:gaussian properties} we get that 
\[(A)=\normtwo{{w}}\sigma_p\Big(2\phi(\psi_p^{[w]}(\theta))-\psi_p^{[w]}(\theta)m(\psi_p^{[w]}(\theta))\Big)\]
as desired.
\end{proof}

\subsection{Bivariate bounds.}

Recall \Cref{fct:d-clt:p-biased}:

\medskip

\noindent
{\bf Fact~\ref{fct:d-clt:p-biased}.}
\emph{Let $\bx \sim u^n_p$ be a $p$-biased random vector in $\bn$, and let $\by$ be a random vector in $\bn$ that is $\rho$-correlated with $\bx$ (meaning that each coordinate $\by_i$ is independently set to equal $\bx_i$ with probability $\rho$ and is set to a random draw from $u_p$ with probability $1-\rho$) for some $\rho$ that is bounded away from 1.\ignore{WAS:  ``Let $\rv{y}=(\rho \rv{x} + \sqrt{1-\rho^2}\rv{z})$, where $\rv{z}\sim \gaussian{0}{\mathbb{I}_n}$.''} Let $w \in \R^n$ be $\tau$-regular, and let $\ell(x)$ denote the linear form $\sum_{i=1}^{n}w_ix_i$.\ignore{ and $sum(w)=(\displaystyle\sum_i w_i)$.} Then for any two intervals $[a,b]$ and $[c,d]$ in $\R$, we have
\[
\Big|\Pr[(\ell(\rv{x}),\ell(\rv{y})) \in [a,b]\times[c,d]]-
\Phi_{0,V}\Big([\psi_p^{[w]}(a),\psi_p^{[w]}(b)]\times  
[\psi_p^{[w]}(c),\psi_p^{[w]}(d)]\Big)
\Big|\le O\left(\frac{\tau}{\sigma_p}\right),
\]
where $V=\begin{bmatrix}
1&\rho\\
\rho&1
\end{bmatrix}$ and $\Phi_{0,V}$ denotes the distribution of the bivariate Gaussian with zero mean and covariance matrix $V$.
}

\medskip

\Cref{fct:d-clt:p-biased} is a $p$-biased analogue of Theorem~68 of \cite{MORS:10}.
The proof uses the following multidimensional analogue of the Berry-Esseen theorem (the statement below can be found as Theorem 16 in \cite{KKMO07} and Corollary 16.3 in \cite{bhatrao}):

\begin{theorem}[Multi-dimensional Berry Esseen]\label{theorem:berry-essen:multi-d}
Let $\rv{X}_1,\cdots,\rv{X}_n$ be independent random vectors in $\R^2$ satisfying:
\begin{itemize}
    \item $\E[\rv{X_j}]={0}$ for all $j=1,\dots,n$, and
        \item $\rho_3=\frac{\sum_{j=1}^{n} \E[\normtwo{\rv{X}_j}^3]}{n}<\infty$.
\end{itemize}
Let ${V}:=\frac{\sum_{j=1}^{n} Cov(\rv{X}_j)}{n}$, where Cov denotes the covariance matrix, and let $\lambda$ be the smallest eigenvalue of $V$ and $\Lambda$ be the largest eigenvalue of $V$.
Let $Q_n$ denote the distribution of $\dfrac{\sum_{j=1}^n \rv{X}_j}{\sqrt{n}}$, let $\Phi_{{0},{V}}$ denote the distribution of the bivariate Gaussian with zero-vector mean and covariance matrix $V$ and let $\eta=C\rho_3\dfrac{1}{\sqrt{n\lambda^3}}$, where $C$ is a certain universal constant.
Then for any Borel set $A$, it holds that
\[|Q_n(A)-\Phi_{{0},{V}}(A)|\le \eta + Bound(A),\] 
where $Bound(A)$ is the following measure of the boundary of $A$: $Bound(A)=2\operatorname{sup}_{{y}\in \R^2}
\Phi_{{0},{V}}((\partial A)^{\eta^\prime} + y)$, where $\eta^{\prime} = \sqrt{\Lambda}\eta$ and $(\partial A)^{\eta^\prime}$ denotes the set of points within distance $\eta^\prime$ of the topological boundary of $A$.
\end{theorem}

\medskip
\noindent \emph{Proof of \Cref{fct:d-clt:p-biased}.}
We first rewrite 
$\Pr[(\ell(\rv{x}),\ell(\rv{y})) \in [a,b]\times[c,d]]$  as
\[
\Pr
\left[\left( \dfrac{\sum_i w_i (\bx_i- \mu_p)}{\normtwo{w}\sigma_p} , \dfrac{\sum_i w_i (\by_i- \mu_p)}{\normtwo{w}\sigma_p} \right)\in [\psi_p^{[w]}(a),\psi_p^{[w]}(b)]\times[\psi_p^{[w]}(c),\psi_p^{[w]}(d)]\right].
\]
We will apply \Cref{theorem:berry-essen:multi-d}. First we define some new random
variables: let 
\[
\bL_i:=(\bA_i,\bB_i):=\left(\dfrac{\sqrt{n}w_i}{\sigma_p\normtwo{w}}(\bx_i-\mu_p),\dfrac{\sqrt{n}w_i}{\sigma_p\normtwo{w}}(\by_i-\mu_p)\right)
\]
for $i=1,\dots,n$. Since each $\bx_i$ and $\by_i$ is individually a $p$-biased random variable over $\bits$, it is easy to see that $\E[\bL_i]=(0,0)$, and it is also straightforward to verify that the covariance matrix of $\bL_i$ is 
\[
\Cov(\bL_i)=
\begin{bmatrix}
\Cov(\bA_i,\bA_i)&\Cov(\bA_i,\bB_i)\\
\Cov(\bB_i,\bA_i)&\Cov(\bB_i,\bB_i)
\end{bmatrix}=\dfrac{nw_i^2}{\normtwo{w}^2}\begin{bmatrix}
1&\rho\\
\rho&1
\end{bmatrix}.
\]
It follows that 
\[
V=\dfrac{1}{n}\displaystyle\sum_{i=1}^{n} \Cov(L_i)=\dfrac{1}{n}\displaystyle\sum_{i=1}^{n}
\dfrac{nw_i^2}{\normtwo{w}^2}\begin{bmatrix}
1&\rho\\
\rho&1
\end{bmatrix}=\begin{bmatrix}
1&\rho\\
\rho&1
\end{bmatrix},
\]
and consequently the eigenvalues of $V$ are $\lambda = (1-\eta)$ and $\Lambda = (1+\eta)$. We note that $\normtwo{\bL_i}=\sqrt{\bA_i^2+\bB_i^2}=
\dfrac{\sqrt{n}w_i}{\sigma_p\normtwo{w}}\sqrt{(\bx_i-\mu_p)^2 +(\by_i-\mu_p)^2}$,
and hence $\E[\normtwo{\bL_i}^2] = {\frac {2 n w_i^2}{\|w\|_2^2}}$. Since
$\sqrt{(\bx_i-\mu_p)^2 +(\by_i-\mu_p)^2} \leq 2 \sqrt{2}$ with probability 1, we have that
$\normtwo{\bL_i} \leq {\frac {2\sqrt{2} \cdot \sqrt{n} w_i}{\sigma_p \|w\|_2}} \leq {\frac {2 \sqrt{2 n} \tau}{\sigma_p}}$ with probability 1. Consequently we have
\begin{align*}
\rho_3=\dfrac{\sum_{i=1}^{n}\E[\normtwo{\bL_i}^3]}{n}
&\le
\dfrac{\sum_{i=1}^{n}\E[\normtwo{\bL_i}^2]}{n} \cdot \max_{i \in [n]}\{\normtwo{\bL_i}\}\\
& = {\frac {2n} n} \cdot \left( \sum_{i=1}^n {\frac {w_i^2}{\|w\|_2^2}}\right) \cdot {\frac {2 \sqrt{2 n} \tau}{\sigma_p}} =  2^{5/2}n^{1/2}\frac{\tau}{\sigma_p}.
\end{align*}
Recalling the value of $\lambda$ and the definition of $\eta$, we get that $\eta = O{(1-\rho)^{-3/2}\dfrac{\tau}{\sigma_p}}$ and since $\rho$ is bounded away from $1$, this is $O({\tau}/{\sigma_p})$. 

It is easy to check that for any $y \in \R^2$, the measure under $\Phi_{0,V}$ of the $y$-translate of the set of points within distance $\eta'$ of the topological boundary of $[\psi_p^{[w]}(a),\psi_p^{[w]}(b)]\times[\psi_p^{[w]}(c),\psi_p^{[w]}(d)]$ is $O(\eta^\prime)$. Since $\eta^\prime=(1+\rho)^{1/2}\eta$, this is also $O({\tau}/{\sigma_p})$. 

Thus  it holds that
\[
\Big|\Pr[(\ell(\rv{x}),\ell(\rv{y})) \in [a,b]\times[c,d]]-
\Phi_{0,V}\Big([\psi_p^{[w]}(a),\psi_p^{[w]}(b)]\times  
[\psi_p^{[w]}(c),\psi_p^{[w]}(d)]\Big)
\Big|\le O\left(\frac{\tau}{\sigma_p}\right),
\]
which is the desired statement.
\qed

\section{\texorpdfstring{Proof of \Cref{lem:monotonicity-shapley}: Shapley indices are monotone in LTF weights}{Shapley indices are monotone in LTF weights}} \label{ap:rank}

Recall \Cref{lem:monotonicity-shapley}:

\noindent {\bf Lemma~\ref{lem:monotonicity-shapley}.} 
\emph{
Let $f(x) = \sign(\ell(x))$ be an LTF where $\ell(x) = \displaystyle\sum_{i=1}^n w_i x_i - \theta$ is a linear form with $w_1,\dots,w_n \geq 0$. Then for all $i \neq j \in [n]$, it holds that if $w_i \geq w_j$ then $\fc(i) \geq \fc(j).$
}
\begin{proof}

Rephrasing \Cref{eq:shapley-values}, the Shapley value for a voter can be expressed as the fraction of all $n!$ orderings of the $n$ voters in which she casts the pivotal vote.
More precisely, for a given ordering (permutation) $\pi\in \mathbb{S}_n$, an index $i$ is the unique pivotal index if starting from $x = (-1)^n$
and flipping coordinates of $x$ from $-1$ to 1 in the order specified by $\pi$, flipping $x_i$ changes $f (x)$ from $-1$ to $1$.  We thus have
\[
\fc(i)=\dfrac{2}{n!}\cdot \displaystyle\sum_{\pi \in \mathbb{S}_n}
\indic{i \text{ is the pivotal index in }\pi \text{ order} }
=
\dfrac{2 \cdot |\{\pi \in \mathbb{S}_n: i \text{ is the pivotal index in }\pi \text{ order} \}|}{n!}.
\]

Let $h:\mathbb{S}_n\to \mathbb{S}_n$ be the following swapping involution:
\[
h(\pi) = \pi^\prime = \begin{cases}
\pi^\prime(x)=\pi(x) & x\not\in\{i,j\}\\
\pi^\prime(i)=\pi(j) \\
\pi^\prime(j)=\pi(i) 
\end{cases}.
\]

We will show that if  $j$ is the pivotal index in permutation $\pi$, then $i$ is the pivotal index in permutation $h(\pi)$. For simplicity of notation in the proof, we write $PR(\pi,k)$ to denote the predecessors of $k$ in permutation $\pi$, i.e
$PR(\pi,k) := \{\ell\in [n]: \pi(\ell) < \pi(k) \} $. Thus equivalently we would like to show that:
\[
\text{if~}
\begin{cases}
\displaystyle\sum_{k \in PR(\pi, j)} w_k <\theta  \text{~and}\\
\displaystyle\sum_{k \in PR(\pi, j)} w_k + w_j \ge \theta, 
\end{cases}
\text{~then~}
\begin{cases}
\displaystyle\sum_{k \in PR(h(\pi), i)} w_k <\theta \text{~and} \\
\displaystyle\sum_{k \in PR(h(\pi), i)} w_k + w_i \ge \theta.
\end{cases}
\]

To complete the exchange argument, we split the permutations where $j$ is the pivotal index into two cases: whether or not $i$ is the predecessor of $j$ in $\pi$.
\begin{enumerate}
\item[{\bf Case 1:}] $ i \in PR(\pi, j) $. By definition of the swapping involution, $ j \in PR(h(\pi), i) $. 

Additionally, it easy to check that
\[
\begin{cases}
PR(\pi, j)\setminus\{i\}=PR(h(\pi), i)\setminus\{j\} \ \ (a)\\
PR(\pi, j)\setminus\{i\}\cup\{j\}=PR(h(\pi),i) \ \ (b)
\end{cases}.
\]
Since $j$ is pivotal in $\pi$, we have that:
\begin{align*}
\{\text{ $j$ is the pivotal index in permutation $\pi$ } \} &\Leftrightarrow& 
\begin{cases}
\displaystyle\sum_{k \in PR(\pi, j)} w_k <\theta \\
\displaystyle\sum_{k \in PR(\pi, j)} w_k + w_j \ge \theta 
\end{cases}\Leftrightarrow\\
\begin{cases}
\displaystyle\sum_{k \in PR(\pi, j)\setminus\{i\}} w_k +w_i<\theta \\
\displaystyle\sum_{k \in PR(\pi, j)\setminus\{i\}} w_k +w_i + w_j \ge \theta 
\end{cases}&\overset{(a)}{\Leftrightarrow}&
\begin{cases}
\displaystyle\sum_{k \in PR(h(\pi), i)\setminus\{j\}} w_k +w_i<\theta \\
\displaystyle\sum_{k \in PR(h(\pi), i)\setminus\{j\}} w_k +w_i + w_j \ge \theta 
\end{cases}\Leftrightarrow
\end{align*}
\begin{align*}
\{\text{ $j$ is the pivotal index in permutation $\pi$ } \} &\overset{w_j<w_i}{\Leftrightarrow}&
\begin{cases}
\displaystyle\sum_{k \in PR(h(\pi), i)\setminus\{j\}} w_k +w_j<\theta \\
\displaystyle\sum_{k \in PR(h(\pi), i)\setminus\{j\}} w_k +w_i + w_j \ge \theta 
\end{cases}\Leftrightarrow\\
\begin{cases}
\displaystyle\sum_{k \in PR(h(\pi), i)} w_k<\theta \\
\displaystyle\sum_{k \in PR(h(\pi), i)} w_k +w_i  \ge \theta 
\end{cases}&\Leftrightarrow&\{\text{ $i$ is the pivotal index in permutation $h(\pi)$ } \}.
\end{align*}

\item[{\bf Case 2:}] $ i \not\in PR(\pi, j) $. By definition of the swapping involution, $ j \not\in PR(h(\pi), i) $. 
Additionally, it easy to check that 
\[PR(\pi, j)=PR(h(\pi), i) \ \ (c)\]
Since $j$ is pivotal in $\pi$, we have that:
\begin{align*}
\{\text{ $j$ is the pivotal index in permutation $\pi$ } \} &\Leftrightarrow& 
\begin{cases}
\displaystyle\sum_{k \in PR(\pi, j)} w_k <\theta \\
\displaystyle\sum_{k \in PR(\pi, j)} w_k + w_j \ge \theta 
\end{cases}\overset{w_j<w_i}{\Leftrightarrow}\\
\begin{cases}
\displaystyle\sum_{k \in PR(\pi, j)} w_k <\theta \\
\displaystyle\sum_{k \in PR(\pi, j)} w_k + w_i \ge \theta 
\end{cases}&\overset{(c)}{\Leftrightarrow}&
\begin{cases}
\displaystyle\sum_{k \in PR(h(\pi), i)} w_k <\theta \\
\displaystyle\sum_{k \in PR(h(\pi), i)} w_k + w_i \ge \theta 
\end{cases}\Leftrightarrow\\
\{\text{ $i$ is the pivotal index in permutation $h(\pi)$ } \}.
\end{align*} \qedhere
\end{enumerate}
\end{proof}

\section{\texorpdfstring{Proof of \Cref{lem:approximation-of-shapley1} and \Cref{lem:approximation-of-shapley2}: The unnormalized $\mathcal{Q}(\delta)$ measure approximates $\DShap$ to high accuracy}{The unnormalized Q\_delta measure approximates D\_Shap to high accuracy}}\label{ap:approx}

The following useful result  intuitively says that the measure given by $\frac{1}{C_\delta}\mathcal{Q}(\delta)$  can take the place of the Shapley distribution $\DShap$ and incur only small error:
\begin{lemma}\label{lemma:intro:tv:bound}
For $\delta>0$, we have
\[
\dtv\parens*{\frac{1}{C_\delta}\mathcal{Q}(\delta),\DShap}\le \dfrac{1}{\Lambda(n)} \cdot O\left(\displaystyle\sum_{k=1}^{n/2} (n\delta)^{k}\right)+
\frac{1}{C_\delta}\Prx_{\rv{x}\sim \mathcal{Q}(\delta)}[\rv{x}=-1^n]+
\frac{1}{C_\delta}\Prx_{\rv{x}\sim \mathcal{Q}(\delta)}[\rv{x}=1^n],
\]
and consequently for \deltaExpression, it holds that
\[
\displaystyle\displaystyle\sum_{k=1}^{n-1}\displaystyle\sum_{x \in \{-1,1\}_{=k}^{n}}\left|\frac{1}{C_\delta}\Prx_{\rv{x}\sim \mathcal{Q}(\delta)}[\rv{x}=x]-\Prx_{\rv{x}\sim\DShap}[\rv{x}=x]\right|\le \dfrac{\bigO{n\delta}}{\Lambda(n)}.
\]
\end{lemma}

\begin{proof}
Recalling that $\Prx_{\rv{x}\sim \DShap}[\rv{x}=\pm 1^n]=0,$ let us fix an $x \in \bn$ such that  $0<\weight(x)<n$. Let $k=\weight(x).$  Then we have that:

\begin{align*}
& \left|\frac{1}{C_\delta}\Prx_{\rv{x}\sim \mathcal{Q}(\delta)}[\rv{x}=x]-\Prx_{\rv{x}\sim\DShap}[\rv{x}=x]\right|
=\left|\displaystyle\int_{\delta}^{1-\delta}\frac{\frac{1}{p}+\frac{1}{1-p}}{\Lambda(n)} p^{k}(1-p)^{n-k} {d}p
-\displaystyle\int_{0}^{1}\frac{\frac{1}{p}+\frac{1}{1-p}}{\Lambda(n)} p^{k}(1-p)^{n-k} {d}p
\right|\\
&=
\left|\displaystyle\int_{[0,\delta]\cup[1-\delta,1]}\frac{\frac{1}{p}+\frac{1}{1-p}}{\Lambda(n)} p^{k}(1-p)^{n-k} {d}p\right|=\displaystyle\int_{[0,\delta]}\frac{\frac{1}{p}+\frac{1}{1-p}}{\Lambda(n)} \Big(p^{k}(1-p)^{n-k} + p^{n-k}(1-p)^{k}\Big) {d}p\\
&=\frac{1}{\Lambda(n)}\displaystyle\int_{[0,\delta]} p^{k-1}(1-p)^{n-k} + p^{k}(1-p)^{n-k-1}
+  (1-p)^{k-1}p^{n-k} + (1-p)^{k}p^{n-k-1} {d}p\\
&\le \frac{1}{\Lambda(n)} \Big( \delta^{k}+\delta^{k+1}+ \delta^{n-k}+\delta^{n-k+1}\Big).
\end{align*}

Consequently for any $1 \leq k \leq n-1$ we have that
\[
\displaystyle\sum_{x \in \{-1,1\}_{=k}^{n}}\left|\frac{1}{C_\delta}\Prx_{\rv{x}\sim \mathcal{Q}(\delta)}[\rv{x}=1^n]-\Prx_{\rv{x}\sim\DShap}[\rv{x}=x]\right|\le \frac{\binom{n}{k}}{\Lambda(n)}\Big( \delta^{k}+\delta^{k+1}+ \delta^{n-k}+\delta^{n-k+1}\Big),
\]
which yields (assuming without loss of generality for simplicity that $n$ is odd)
\begin{align*}
\displaystyle\displaystyle\sum_{k=1}^{n-1}\displaystyle\sum_{x \in \{-1,1\}_{=k}^{n}}\left|\frac{1}{C_\delta}\Prx_{\rv{x}\sim \mathcal{Q}(\delta)}[\rv{x}=x]-\Prx_{\rv{x}\sim\DShap}[\rv{x}=x]\right|
\ignore{\le \displaystyle\sum_{k=1}^{n-1} \frac{\binom{n}{k}}{\Lambda(n)}\Big( \delta^{k}+\delta^{k+1}+ \delta^{n-k}+\delta^{n-k+1}\Big)}
&\le \displaystyle\sum_{k=1}^{n-1} \frac{\binom{n}{k}}{\Lambda(n)}\Big( \delta^{k}+\delta^{k+1}+ \delta^{n-k}+\delta^{n-k+1}\Big)\\
&= 2\displaystyle\sum_{k=1}^{(n-1)/2} \frac{\binom{n}{k}}{\Lambda(n)}\Big( \delta^{k}+\delta^{k+1}+ \delta^{n-k}+\delta^{n-k+1}\Big)\\
&\le\displaystyle\sum_{k=1}^{\red{n/2}}\frac{4n^k}{\Lambda(n)}\Big( \delta^{k}+\delta^{k+1}\Big),
\end{align*}
Thus, it holds that
\[
\dtv\parens*{\frac{1}{C_\delta}\mathcal{Q}(\delta),\DShap}\le \dfrac{1}{\Lambda(n)} \cdot O\left(\displaystyle\sum_{k=1}^{n/2} (n\delta)^{k}\right)+
\frac{1}{C_\delta}\Prx_{\rv{x}\sim \mathcal{Q}(\delta)}[\rv{x}=-1^n]+
\frac{1}{C_\delta}\Prx_{\rv{x}\sim \mathcal{Q}(\delta)}[\rv{x}=1^n]
\]
and the lemma is proved.
\end{proof}

Now we are ready to prove \Cref{lem:approximation-of-shapley1} and \Cref{lem:approximation-of-shapley2}: 

\medskip

\noindent {\bf \Cref{lem:approximation-of-shapley1}}. 
\emph{Let \deltaExpression and let $f:\pmo^n\to \R$ be such that $\|f\|_\infty \le O(1)$ and  $f(-1)^n=f(1^n)=0$.Then it holds that
\[  \Big|\Ex_{\rv{x}\sim \DShap}[f(\rv{x})]-\frac{1}{C_\delta}\Ex_{\rv{x}\sim \mathcal{Q}(\delta)}[f(\rv{x})]\Big| \le \dfrac{\bigO{n\delta}}{\Lambda(n)}.
\]
}
\begin{proof}
\begin{align*}
\abs*{    \Ex_{\rv{x}\sim \DShap}[f(\rv{x})]-\frac{1}{C_\delta}\Ex_{\rv{x}\sim \mathcal{Q}(\delta)}[f(\rv{x})]}
    &=\abs*{\displaystyle\sum_{x\in\bn} f(x)\Prx_{\rv{x}\sim \DShap}[\rv{x}=x]-f(x)\frac{1}{C_\delta}\Prx_{\rv{x}\sim \mathcal{Q}(\delta)}[\rv{x}=x]}\\
    &=
    \abs*{\displaystyle\sum_{x\in\bn\setminus\{-1^n,1^n\}} f(x)\Prx_{\rv{x}\sim \DShap}[\rv{x}=x]-f(x)\frac{1}{C_\delta}\Prx_{\rv{x}\sim \mathcal{Q}(\delta)}[\rv{x}=x]}\\
    &\le \dfrac{\bigO{n\delta}}{\Lambda(n)} \cdot \bigO{1},
\end{align*}
where the second equality uses $f((-1)^n)=f(1^n)=0$ and the last inequality is by \Cref{lemma:intro:tv:bound} together with a straightforward application of \Cref{fact:holder}.
\end{proof}

\medskip

\noindent {\bf \Cref{lem:approximation-of-shapley2}}. 
\emph{Let $f$ be any nontrivial monotone LTF (so $f((-1)^n)=-1$ and $f(1^n) = 1$).  Then for \deltaExpression, for each $i \in [n]$, the value $\fc(i)$ is additively $O(n \delta)$-close to the quantity $\upsilon_i$ defined below:
\begin{align*}
\fc(i) &{\stackrel{ \bigO{n\delta} }{\approx}} \frac{2}{n} + \frac{\Lambda(n)}{2} \cdot 
\dfrac{1}{C_\delta}\cdot \Ex_{\rv{p}\sim \mathcal{K}(\delta)}\bracks*{ f^*(i,\rv{p})  - {\frac 1 n} \sum_{j=1}^n f^*(j,\rv{p}) } =:\upsilon_i \\
& = \frac{2}{n} + \frac{\Lambda(n)}{2} \cdot 
\dfrac{1}{C_\delta}\cdot \Ex_{\rv{p}\sim \mathcal{K}(\delta)}\bracks*{\sigma_{\rv{p}} \fh(i,\rv{p})  - {\frac 1 n} \sum_{j=1}^n \sigma_{\rv{p}}\fh(j,\rv{p})}.
\end{align*}
}

\begin{proof}
\begin{align*}
\fc(i) &= \frac{f(1^n ) - 
f((-1)^n)}{n} + \frac{\Lambda(n)}{2} \cdot \left( f^*(i)  - {\frac 1 n}
\sum_{j=1}^n f^*(j)\right) \tag{\Cref{lem:expshap}}\\
&=\frac{2}{n} + \frac{\Lambda(n)}{2} \cdot \left( \Ex_{\rv{x}\sim \DShap}\left[f(\rv{x})\left(\rv{x}_i-\frac{\sum_{k\in[n]}\rv{x}_k}{n}\right)\right]\right)   \tag{Definition of $f^*(i)$}\\
&\stackrel{ \bigO{n\delta} }{\approx} \frac{2}{n} + \frac{\Lambda(n)}{2} \cdot \left( \dfrac{1}{C_\delta}\cdot \Ex_{\rv{x}\sim \mathcal{Q}(\delta)}\left[f(\bx)\left(\bx_i-\frac{\sum_{k\in[n]}\bx_k}{n}\right)\right]\right) \tag{\Cref{lem:approximation-of-shapley1}}\\
&=\ignore{\stackrel{ \bigO{n\delta} }{\approx}} \frac{2}{n} + \frac{\Lambda(n)}{2} \cdot 
\dfrac{1}{C_\delta}\cdot \Ex_{\rv{p}\sim \mathcal{K}(\delta)}\left[ f^*(i,\rv{p})  - {\frac 1 n} \sum_{j=1}^n f^*(j,\rv{p}) \right] = \upsilon_i,
\end{align*}
giving the first claimed approximation (where the last equality holds recalling that $f^*(i,p)  =  \Ex_{\bx\sim u_p^n}  \left[f(\bx)\bx_i\right]$, recall Table~1).
For the second statement, observe that as a straightforward consequence of the definition of $f^*(i,p)$ we have that
\[f^*(i,p) =\sigma_p \fh(i,p) + \Ex_{\bx\sim u_p^n}  \left[f(\bx)\right]\mu_p.
\]
Using the above equivalent definition we get that:
\begin{align*}
f^*(i,p)   - {\frac 1 n} \sum_{j=1}^n f^*(i,p) &= \sigma_p\fh(i,p)+\Ex_{\bx \sim u^n_p}[f(\rv{x})]\mu_p- {\frac 1 n} \sum_{j=1}^n  \parens*{\sigma_p\fh(j,p)+\Ex_{\bx \sim u^n_p}[f(\rv{x})]\mu_p }\\
&= 
 \sigma_p \parens*{ \fh(i,{p})  - {\frac 1 n} \sum_{j=1}^n \fh(j,{p}) }
 \end{align*}
which gives the second statement as claimed.
\ignore{
Thus we get 
\[
\fc(i) {\stackrel{ \bigO{n\delta} }{\approx}} \frac{2}{n} + \frac{\Lambda(n)}{2} \cdot 
\dfrac{1}{C_\delta}\cdot \Ex_{\rv{p}\sim \mathcal{K}(\delta)}\left[ \sigma_p \fh(i,p)  - {\frac 1 n} \sum_{j=1}^n \sigma_p\fh(j,p) \right]
\]
}
\end{proof}

\end{document}